\documentclass[lettersize,journal]{IEEEtran}
\usepackage{amsmath,amssymb, amsfonts}
\usepackage{amsthm}
\usepackage{array}
\usepackage[linesnumbered,ruled,vlined]{algorithm2e}
\usepackage{textcomp}
\usepackage{stfloats}
\usepackage{url}
\usepackage{verbatim}
\usepackage{graphicx}
\usepackage{cite}
\usepackage{xcolor}
\usepackage{tabularray}
\usepackage{subcaption}
\usepackage{booktabs,caption,siunitx}
\newtheorem{theorem}{Theorem}
\newtheorem{corollary}{Corollary}
\newtheorem{lemma}{Lemma}
\newtheorem{assumption}{Assumption}

\usepackage{pifont}% http://ctan.org/pkg/pifont
\newcommand{\cmark}{\ding{51}}%
\newcommand{\xmark}{\ding{55}}%

\usepackage{etoolbox}
\makeatletter
\patchcmd{\@makecaption}
  {\scshape}
  {}
  {}
  {}
\makeatletter
\patchcmd{\@makecaption}
  {\\}
  {.\ }
  {}
  {}
\makeatother

\allowdisplaybreaks

\SetKwInput{KwConstants}{Constants}
\SetKwInput{KwOutput}{Output}
\SetKwInput{KwInput}{Input}
\SetKwFor{ClientFor}{each client}{in parallel:}{}
\SetKwFor{SubnetFor}{each subnet}{in parallel:}{}
\SetKwFor{AggrFor}{all sampled clients}{aggregate:}{}

\begin{document}

% \title{Taming Subnet-Drift in D2D-Enabled Fog Learning: A Hierarchical Gradient Tracking Approach}
\title{A Hierarchical Gradient Tracking Algorithm for Mitigating Subnet-Drift in Fog Learning Networks}

\author{Evan Chen,~\IEEEmembership{Student Member,~IEEE},
~Shiqiang Wang,~\IEEEmembership{Fellow,~IEEE},
\\ Christopher G. Brinton,~\IEEEmembership{Senior Member,~IEEE}
\thanks{This work has been accepted for publication in IEEE/ACM Transactions on Networking (ToN).}
\thanks{Evan Chen and Christopher G. Brinton are with the Elmore Family School of Electrical and Computer Engineering, Purdue University, West Lafayette, IN, 47907, USA. E-mail: \{chen4388,cgb\}@purdue.edu.}% <-this % stops a space
\thanks{Shiqiang Wang is with the Department of Computer Science, University of Exeter, EX4 4RN, UK. E-mail: shiqiang.wang@ieee.org.}
\thanks{A preliminary version of this work appeared in the 2024 IEEE Conference on Computer Communications (INFOCOM) \cite{chen2024taming}.}
\thanks{This work was supported in part by the National Science Foundation (NSF) under grants CPS-2313109 and CNS-2212565, by DARPA under grant D22AP00168, by the Office of Naval Research (ONR) under grant N000142212305, and by the Air Force Office of Scientific Research (AFOSR) under grant FA9550-24-1-0083.}}



\maketitle

\begin{abstract}
Federated learning (FL) encounters scalability challenges when implemented over fog networks that do not follow FL's conventional star topology architecture. Semi-decentralized FL (SD-FL) has proposed a solution for device-to-device (D2D) enabled networks that divides model cooperation into two stages: at the lower stage, D2D communications is employed for local model aggregations within subnetworks (subnets), while the upper stage handles device-server (DS) communications for global model aggregations. However, existing SD-FL schemes are based on gradient diversity assumptions that become performance bottlenecks as data distributions become more heterogeneous. In this work, we develop semi-decentralized gradient tracking ({\tt SD-GT}), the first SD-FL methodology that removes the need for such assumptions by incorporating tracking terms into device updates for each communication layer. Our analytical characterization of {\tt SD-GT} reveals  upper bounds on convergence for non-convex, convex, and strongly-convex problems. We show how the bounds enable the development of an optimization algorithm that navigates the performance-efficiency trade-off by tuning subnet sampling rate and D2D rounds for each global training interval.
%We employ the resulting bounds in the development of a co-optimization algorithm for optimizing subnet sampling rates and D2D rounds according to a performance-efficiency trade-off. 
Our subsequent numerical evaluations demonstrate that {\tt SD-GT} obtains substantial improvements in trained model quality and communication cost relative to baselines in SD-FL and gradient tracking on several datasets.
\end{abstract}


\section{Introduction}
\noindent Federated learning (FL) has emerged as a promising technique for distributed machine learning (ML) over networked systems \cite{kairouz2021advances}. FL aims to solve problems of the following form:
\begin{align}
    \min_{x\in \mathbb{R}^d}f(x) &= \frac{1}{n}\sum_{i=1}^n f_i(x)\\
    \textrm{for } f_i(x) &= \mathbb{E}_{\xi_i \sim \mathcal{D}_i}f_i(x;\xi_i),
\end{align}
where $n$ is the total number of clients (typically edge devices) in the system, $f_i(x)$ is the local ML loss function computed at client $i$ for model parameters $x \in \mathbb{R}^d$, $\mathcal{D}_i$ is the local data distribution at client $i$, and $\xi_i$ is a random sample from $\mathcal{D}_i$. 

\begin{figure}
% \centerline{\includegraphics[width=0.5\textwidth]{Images/image3.pdf}}
\centerline{\includegraphics[width=0.5\textwidth]{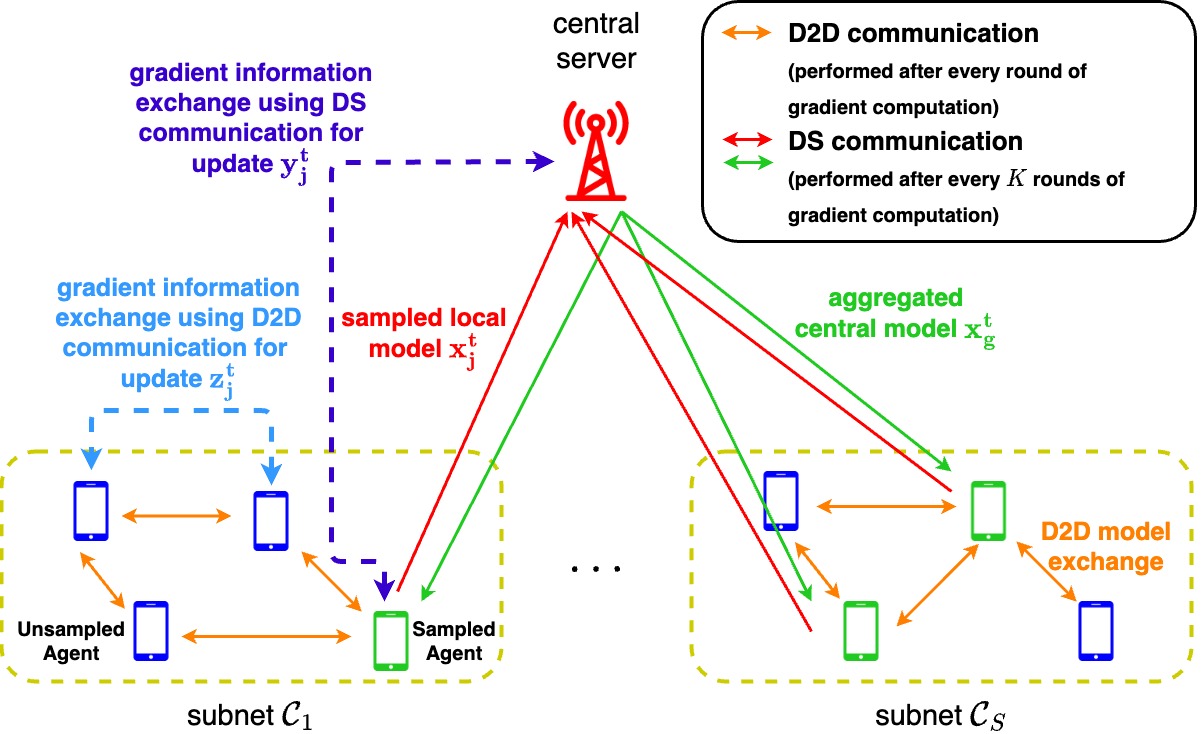}}
\caption{Illustration of semi-decentralized FL with gradient tracking. Clients in each subnet communicate via iterative low-cost D2D communications to conduct local aggregations. Once they have converged towards a consensus within the subnet, the central server conducts a global aggregation across sampled devices using DS communication. Two separate terms related to gradient tracking are maintained, corresponding to within subnet and between subnet gradient information, respectively.}
\vspace{-0.2in}
\label{fig1}
\end{figure}

Conventionally, FL employs a two-step iterative algorithm to solve this optimization: (i) \textit{local model update}, where gradient information computed on the local device dataset is used to update the local model, and (ii) \textit{global model aggregation}, where a central server forms a consensus model across all devices. The server is assumed to be connected to each device over a star topology. In wireless networks, however, device-server (DS) communications for the global aggregation step can be expensive, especially for large ML models over long DS distances. Much research in FL has been devoted to improving the communication efficiency of this step, with typical approaches including model sparsification/quantization \cite{wang2022federated,amiri2020federated}, device sampling~\cite{wang2024device}, and aggregation frequency minimization \cite{karimireddy2020scaffold,mishchenko2022proxskip}.

\IEEEpubidadjcol
Recent research has considered how decentralizing FL's client-server star topology can improve communication overhead, e.g., by introducing more localized communications wherever possible. In the extreme case of severless FL, model aggregations are conducted entirely through short range device-to-device (D2D) communications \cite{koloskova2019decentralized,lian2017can,zehtabi2022decentralized}. More generally, \textit{fog learning} \cite{hosseinalipour2020federated,nguyen2022fedfog,hosseinalipour2022multi,brinton2025key} considers distributing FL over fog computing architectures. Then fog computing, a hierarchy of network elements separates the edge and the cloud, which makes the learning process more complex, as it must consider communication between computing nodes within the same level and spanning multiple levels of the hierarchy.

%complementing DS interactions with more localized communications among edge devices and fog nodes. This encompasses
%For other settings where a main server still exists despite being costly to use, fog learning has emerged to combine the advantages of centralized and decentralized model aggregations \cite{hosseinalipour2020federated}. Fog learning is a multi-layer hybrid learning framework where multiple-hierarchical layers of aggregation is performed to gather information to a central server. In this paper, we will focus on a specific type of fog learning architecture, semi-decentralized FL (SD-FL), which is composed of two hierarchical network layers \cite{lin2021semi}.

%Under this network structure, we are able to mitigate the problem that the central server can not simultaneously communicate with large amount of clients while also using the advantage that the central server is able to link to all clients in the network.

\subsection{Semi-Decentralized FL}
\label{ssec:examples}
\textit{Semi-decentralized FL} (SD-FL) has emerged as a canonical architecture of fog learning~\cite{lin2021semi,yemini2022semi}. Its overall structure is depicted in Figure~\ref{fig1}. Devices are grouped into subnetworks (subnets) of close physical proximity, according to their ability to form D2D connections \cite{suraci2021trusted}. The central server performs global aggregation by sampling partial sets from each subnet. To enhance communication efficiency, the model aggregation in SD-FL is conducted in two stages: (i) \textit{iterative cooperative consensus formation} of local models within subnets, and (ii) \textit{DS communication among sampled devices} for global aggregation. The idea is that frequent, low-cost within-subnet model aggregations should reduce the burden placed on global, cross-subnet aggregations, as they can occur less frequently engaging fewer clients. 

The following are some practical examples of networks and learning tasks applicable to SD-FL:
%The following are some practical applications of where the network setting contains D2D connections, hence are some examples that can perform learning tasks in a fog network setting:

\textbf{Sensor Networks in Smart Manufacturing.}
Consider a set of smart sensors distributed across manufacturing plants. The sensors within each plant may establish D2D communications to perform learning tasks such as power control and predicting production delays across assembly lines \cite{varga20205g}. Moreover, the backbone network providing connectivity to each plant may be able to facilitate global modeling of multiple plants to take advantage of diverse training environments.

\textbf{Mobile Augmented Reality (AR) Devices.} For AR applications based on mobile phones, machine learning tasks like image classification \cite{chatzopoulos2017mobile} can be collaboratively trained by communicating with nearby phones through D2D protocols. Moreover, since devices participating in an AR learning task may be too geographically dispersed to be connected by a single D2D topology, base stations with central servers to aggregate information are also essential. 

\textbf{Connected and Automated Vehicles (CAV).}
Distributed networks of CAVs employ machine learning to optimize various control, planning, and self-driving functions \cite{chellapandi2023federated}. Vehicle-to-Vehicle (V2V) and Vehicle-to-Infrastructure (V2I) communication are important to facilitating any distributed learning task that involves training across vehicles spanning multiple geographical locales (e.g., cities).
% Machine learning tasks are an essential way for CAVs to learn control, planning, and self-driving algorithms in a distributed network. For vehicles that are constantly transporting around a city or even between cities, a hierarchical network topology with the ability to perform Vehicle-to-Vehicle (V2V) and Vehicle-to-Infrastructure (V2I) communication will be very important to trains any distributed learning tasks \cite{chellapandi2023federated,subedi2022synchronizing}.
% There are multiple existing applications that can take advantage of D2D communication and perform learning tasks in a fog network setting:
\subsection{The Subnet-Drift Problem}
%SD-FL, however, encounters a fundamental challenge in terms of managing \textit{gradient diversity} across subnets.
A fundamental challenge in SD-FL is managing the \textit{gradient diversity} across subnets. 
Specifically, the more within-subnet aggregations we perform, the more the global model drifts away from the global optimum, towards a linear combination of local optimums of each subnet. This ``subnet-drift" manifests from the client-drift problem in FL, due to non-i.i.d. local datasets across clients \cite{karimireddy2020scaffold,liu2023decentralized,mishchenko2022proxskip}. For example, in the AR use case above, users in an urban areas may be presented with different distributions of environmental data (e.g. densely populated blocks of pedestrians) compared to users in sub-urban areas. In other words, when users are grouped into subnets based on their physical location, the data distribution between subnets may vary significantly.

In this work, we are interested in addressing the subnet-drift challenge for SD-FL. Although some existing works alleviate client drift by letting clients share a portion of their datasets with their neighbors and/or the server \cite{wang2021network,wang2024device}, such approaches present privacy issues that FL aims to avoid. To rectify this, we turn to concepts in \textit{gradient tracking}, which have been successful in mitigating data heterogeneity challenges in fully decentralized learning, and do not require data sharing \cite{koloskova2019decentralized,lian2017can}. However, the hierarchical nature of SD-FL presents two key research challenges to gradient tracking methods. First, the \textit{differing timescales} of D2D and DS communications necessitate careful consideration on how a client should employ gradient information from the server versus from its neighbors. Second, \textit{randomness in client participation} for DS communication may create biases in aggregated gradient information. We thus pose the following research question:
\begin{center}
\textbf{\textit{How do we alleviate subnet drift in semi-decentralized FL through gradient tracking while ensuring gradient information is well mixed throughout the system?}}
\end{center}
% The main idea of gradient tracking is for each client to track gradient information from its neighbors during the learning communication process. 
% \shiqiang{How is our gradient tracking different from well-known gradient tracking approaches such as SCAFFOLD?}
% \evan{This is mentioned in the next paragraph, where I explain how we address the first research challenge.}
To address this, a key component of our design is to introduce two separate gradient tracking terms, one corresponding to information sharing within each subnet (facilitated by D2D communications), and one for information sharing across subnets (facilitated by the server). Since incoming information to each client from within the subnet (through D2D communication) and across subnets (through DS communication) may have differing statistical properties, maintaining two gradient tracking terms can help adjust the local gradient update direction in a more stable manner compared with only using one.
% each communication stage of SD-FL, which treat the incoming information differently.
 %Another research challenge is that the randomness of client participation
% \shiqiang{Why does the server need to sample clients? Can they decide whether to participate on their own? Also, I think we should be consistent with the term ``client'' instead of using ``client'', unless you mean something different by ``client''.}
%for DS communication will create a bias term on the aggregated gradient information.
Our convergence analysis and subsequent experiments demonstrate how this indeed stabilizes the global learning process, and does so \textit{internally to the algorithm itself}, i.e., without having to increase the communication frequency between clients and the server to limit the drift.
%We address this problem by cancelling out the gradient tracking information when global aggregation is performed, making the convergence more stable.
% We show that we're able to recover several gradient tracking based methods based on SD-GT on specific network structures, showing that our algorithm is a more generalized version. We also present convergence analysis for SD-GT, showing that by an appropriate choice of step size, we're able to converge in a sub-linear rate under non-convex settings and in a linear rate under strong-convexity, while having heterogeneity-independent complexity.
% Base on our convergence rate, we also propose a control algorithm to trade-off between energy cost and convergence speed: The more clients the central server samples during global aggregation, the faster the convergence speed will be, but it will cost more energy to perform one round of global aggregation. Our control algorithm allows the server to balance between the number of sampled clients from each subnet, the number of D2D communication rounds performed between two global aggregation, and the convergence speed.
\subsection{Outline and Summary of Contributions}
\label{ssec:contribution}
\begin{itemize}
    \item We propose Semi-Decentralized Gradient Tracking ({\tt SD-GT}), the first work which integrates gradient tracking into SD-FL, building robustness to data heterogeneity across clients and subnets. 
    %\chris{Through careful design of two gradient tracking terms ...}
    Through careful design of two gradient tracking terms,
    {\tt SD-GT} can tolerate a large number of D2D communications between two global aggregation rounds without risking convergence to a sub-optimal solution (Sec. \ref{sec:III}).
    %\item We showed that under certain network structures, several GT-based algorithms can be viewed as a special case of SD-GT (Sec. \ref{sec:IIIC}).
    
    \item We conduct a Lyapunov-based convergence analysis for {\tt SD-GT}, obtaining upper bounds on convergence for non-convex, weakly convex and strongly convex functions. 
    %\chris{Key takeaways of convergence analysis is }
    We show that our algorithm converges with a sub-linear rate under non-convex and weakly convex problems, while for a strongly-convex problem with deterministic gradients, it obtains a linear convergence rate. Importantly, compared with prior works in SD-FL, our convergence bounds do not depend on data heterogeneity constants (Sec. \ref{sec:IV}).
    
    \item We show how our convergence bounds can be employed in a co-optimization of convergence speed and communication efficiency via adapting D2D communication rounds and subnet sampling rates in {\tt SD-GT}.
    % \chris{This shows how our method enables simpler optimization procedures compared to other control algorithms ...}
    Through the gradient tracking mechanism in our algorithm, this co-optimization is directly solvable via geometric programming techniques, as opposed to prior SD-FL control optimization algorithms which must carefully adapt based on data heterogeneity 
    (Sec. \ref{sec:IV}). {\color{black} We further derived a non-convex convergence rate of our algorithm under the co-optimization control (Appendix \ref{appen:ctrl_convergence}).}
 
    \item Our experiments verify that {\tt SD-GT} obtains substantial improvements in trained model quality and convergence speed relative to baselines in the SD-FL and gradient tracking literature. They also demonstrate consistency in our improvements as factors like subnet composition and data heterogeneity vary. Moreover, we verify the behavior of our co-optimization optimization in adapting to the relative cost of D2D vs. DS communications (Sec. \ref{sec:V}).
    
%    \item We experimentally verify the theoretical results on both on strongly convex and non-convex functions, showing how different number of D2D communication rounds and number of samples the central server gets from each subnet effects convergence speed, and show the effectiveness of our proposed control algorithm (Sec. \ref{sec:VI}).
\end{itemize}
% \subsection{Placeholder}

\textcolor{black}{This paper is an extension of our conference paper version of this work\cite{chen2024taming}. Compared to \cite{chen2024taming}, we make the following additional contributions: (1) A new theorem discussing the case of weakly convex ($\mu = 0$) objectives is included. (2) We develop a new co-optimization control algorithm that adapts the device sampling rate and D2D communication frequencies after each global aggregation. (3) We conduct a more comprehensive set of experiments to discuss the effect of different network structures, data partitions, and other factors. (4) For the mathematical claims (lemmas, theorems, etc.), we include sketch proofs in the main text and full proofs in the supplemental material.}

\vspace{-0.1in}
\section{Related Works}

\label{sec:II}
% \textcolor{red}{TODO:} Add more papers on FL part.
\vspace{-0.1in}
\subsection{Communication Efficient Federated Learning}
Multiple works on FL have considered optimizing communication resource efficiency. One of the earliest ideas was to allow for multiple local gradient updates between two consecutive communications with the central server \cite{haddadpour2019convergence, mcmahan2017communication}. However, when local datasets exhibit significant heterogeneity, the local updates may lead the individual models to converge toward locally optimal solutions. Techniques for mitigating this include adapting the interval between consecutive aggregations according estimates of the degree of heterogeneity \cite{wang2019adaptive} as well as intelligently selecting clients to participate in each aggregation \cite{chen2020optimal,nguyen2020fast,ribero2020communication,luo2022tackling}, to more judiciously make use of available DS communication resources. Additionally, there are also studies on how model quantification and model sparsification techniques can be used to optimize resource efficiency\cite{wang2022federated,amiri2020federated}.

Other works have aimed to address this challenge by directly reducing the impact of data heterogeneity across clients, thereby allowing for less frequent aggregations. For instance, some have proposed quantifying the similarity between each client's local dataset and the global distribution, and adding regularizers to the local training process for reducing the impact of data heterogeneity \cite{li2020federated, sun2022distributed}. 
Others have opted for sharing a subset of local datasets between devices and/or with the central server \cite{wang2021network, wang2024device}. However, direct sharing of raw data over the network naturally raises privacy concerns. Therefore, it is preferable to mitigate this problem by transmitting only model variables and gradient information.

\vspace{-0.1in}
\subsection{Hierarchical Federated Learning}
% Aiming to improve of resource efficiency of communication while also aggregating the global information towards a central server, the idea of a hierarchical network FL was created\cite{liu2020client,hosseinalipour2022multi}.
%\subsection{Federated Learning}
% \textbf{Federated Learning:} Multiple works on federated learning have been trying to optimize the resource efficiency of communication. There are works that performs multiple local gradient updates between two consecutive communications with the central server . There are also studies on how model quantification and model sparsification techniques can be used to optimize resource efficiency\cite{wang2022federated,amiri2020federated,sattler2019robust}.
% There are works that tries to deal with data-heterogeneity among clients.  However, sharing raw data between the server and the clients may cause privacy concerns, hence it's best if we can mitigate this problem using only model variables and gradient information.
%\subsection{Hierarchical Network Federated Learning}
%, employing fog computing nodes to facilitate local model relaying and/or aggregations
Hierarchical FL has received considerable attention for scaling up model training across large numbers of edge devices. Most of the works have considered a multi-stage tree extension of FL \cite{wang2021resource,liu2020client,hosseinalipour2022multi,wang2022infedge, wang2022demystifying}, i.e., with ``parent'' nodes at each stage of the hierarchy responsible for its own local FL star topology comprised of its ``child'' nodes. A commonly considered use case has been the three-tier hierarchy involving devices, base stations, and cloud encountered in cellular networks. Optimization of the aggregation frequencies across the hierarchy stages have demonstrated significant improvements in convergence speed and communication efficiency. In a separate domain, these concepts have been employed for model personalization in cross-silo FL~\cite{zhou2023hierarchical, chu2024multi}.

%each cluster of clients update towards a consensual direction and the central server sparsely sample information from the cluster\cite{lin2021semi,yemini2022semi,zhou2023hierarchical}. Our work focuses specifically on semi-decentralized FL, where 
Our work focuses on semi-decentralized FL (SD-FL), where subnets of edge devices conduct local aggregations via D2D-enabled cooperative consensus formation, and the central server samples models from each subnet \cite{lin2021semi,yemini2022semi,parasnis2024energy}. SD-FL is intended for settings where DS communications are costly, e.g., due to long edge-cloud distances. The convergence behavior of SD-FL was formally studied and a corresponding control algorithm was proposed to maintain convergence based on approximations of data-related parameters \cite{lin2021semi}. The work in \cite{parasnis2024energy} developed SD-FL based on more general models of subnet topologies that may be time-varying and directed. 
A main issue with all current SD-FL papers is that their theoretical bounds assume that either the gradient, gradient diversity, or data-heterogeneity are bounded. Both \cite{lin2021semi} and \cite{parasnis2024energy} even require knowledge on the connectivity of each subnet to run the algorithm. In our work, through gradient tracking, we will not require knowledge on the relationship between different local data distributions for deriving convergence guarantees, and will not need any parameter related to network topology for our algorithm. With less information required to control the communication-related parameters, we are able to develop control algorithms that are easier to implement in practice.
\vspace{-0.1in}
\subsection{Gradient Tracking for Communication Efficiency}
Gradient tracking (GT) methods \cite{di2016next,nedic2017achieving,tian2018asy,koloskova2021improved,sun2022distributed} were proposed to mitigate data heterogeneity in decentralized optimization algorithms. The main idea is to track the gradient information from neighbors every time communication is performed. GT has become particularly popular in centralized \cite{karimireddy2020scaffold,mishchenko2022proxskip} and serverless \cite{liu2023decentralized,alghunaim2024local,ge2023gradient,berahas2023balancing,zhang2021low} FL settings where communication costs are high, as it enables algorithms to reach the optimum point while increasing the interval between synchronization. These works have demonstrated that assumptions on data heterogeneity can be lifted under proper initialization of gradient tracking variables.
% \shiqiang{What is variable initialization? Do you mean choice of learning rates?}.\evan{The GT terms (like the y,z in our paper) needs specific initalization. Like they have to be all zero at the initialization in our paper.}\shiqiang{Suggest changing it to: proper initialization of gradient tracking variables}\evan{ok}

Our work instead considers GT under a semi-decentralized network setting. In this respect, \cite{huang2022tackling} discussed GT under a hierarchical network structure, where they assumed a topology consisting of (i) random edge activation within subgraphs and (ii) all subgraphs being connected by a higher layer graph that communicates after every gradient update. This is different from the SD-FL setting, where D2D communication usually is cheaper than DS communication and thus occurs at a much higher frequency. In this paper, we develop a GT methodology that accounts for the diversity in information mixing speeds between D2D and DS communications, and track this difference by maintaining two separate GT terms.

\vspace{-0.1in}
\section{Proposed Method}
\label{sec:III}
% \subsection{Notations}
% Under the setting of semi-decentralized FL, the network performs $K$ rounds of D2D communication within each cluster, then the central server samples a number of clients $h_{c_i}$ from each cluster to update the global variables.
In this section, we first introduce the overall network structure of SD-FL (Sec. \ref{sec:IIIA}). Then we develop our {\tt SD-GT} algorithm, explaining the usage of each tracking variable and how they solve the subnet-drift problem (Sec. \ref{sec:IIIB}). Finally, we show that our method encapsulates two existing methods under specific network topologies (Sec. \ref{sec:IIIC}). 
% Finally, we provide the convergence analysis, showing that with a constant step size that is determined by the $L$-smoothness constant, the strongly-convex constant, and the number of D2D communication rounds, we are able to achieve sub-linear rate under non-convex and linear-rate convergence under strongly-convex (Sec \ref{sec:IV}).
\vspace{-0.1in}
\subsection{Network Model and Timescales}
\label{sec:IIIA}
We consider a network containing a central server connected upstream from $n$ clients (edge devices), indexed $i = 1,...,n$. As shown in Figure \ref{fig1}, the devices are partitioned into $S$ disjoint subnets 
% \shiqiang{is a subnet the same as a cluster?}\evan{Yes, in this current version we should be using subnet instead of cluster throughout the whole paper}
$\mathcal{C}_1, \ldots, \mathcal{C}_S$. Subnet $s$ contains $m_{s} = |\mathcal{C}_s|$ clients, where $\sum_{s=1}^S m_{s} = n$. Similar to existing works in SD-FL~\cite{lin2021semi,yemini2022semi}, we do not presume any particular mechanism by which clients have been grouped into subnets, except that clients within the same subnet are capable of engaging in D2D communications according to a wireless protocol, e.g., smart sensors communicating through 5G/6G in a manufacturing plant (see Sec.~\ref{ssec:examples}).

For every client $i \in \mathcal{C}_s$, we let $\mathcal{N}_i \subseteq \mathcal{C}_s$
% \shiqiang{does $in$ in the superscript refer to both the variables $i$ and $n$? If yes, add a comma between them, if not use ``$\mathrm{in}$'' (note the font type) if needed, but single letter sub-/superscripts are generally preferred. Also, capital $C_i$ is not defined, do you mean $c_i$ instead?}
be the set of clients that can transmit updates to client $i$ using D2D transmissions. Considering all clients $i, j \in \mathcal{C}_s$, we define $W_{s} = [w_{ij}] \in \mathbb{R}^{m_{s}\times m_{s}}$ to be the D2D communication matrix for subnet $s$, where $0 < w_{ij} \leq 1$ if $j \in \mathcal{N}_i$, and $w_{ij} = 0$ otherwise. As we will see in Sec.~\ref{sec:IIIB}, $w_{ij}$ is the weight that client $i$ will apply to information received from client $j$. \textcolor{black}{Thus, these communication matrices enable distributed average consensus among devices within each subnet through D2D communication. These local model aggregations, together with the gradient tracking procedure introduced in Sec.~\ref{sec:IIIB}, help ensure that the local models in each subnet do not drift too far apart between global aggregations by the main server.}

\textcolor{black}{Similar to existing SD-FL literature~\cite{yemini2022semi,parasnis2024energy}, will primarily assume the topology for each subnet $\mathcal{C}_s$ is fixed, and the corresponding communication matrix $W_s$ is static and pre-defined. With an undirected and strongly connected communication graph for each subnet, standard approaches such as the Metropolis-Hastings algorithm can be employed to compute a doubly stochastic weight matrix $W_s$, assigning appropriate weights to each communication link among devices within the subnet~\cite{parasnis2024energy,zehtabi2022decentralized}. In Sec.~\ref{sec:IV}, we will discuss how our results can be extended to the time-varying case.}

We then can define the network-wide D2D matrix
\begin{equation}
\label{eq:W}
W = \begin{bmatrix}
    W_{1}& \ldots& 0\\
    \vdots & \ddots & \vdots\\
    0 & \ldots & W_{S}
\end{bmatrix} \in \mathbb{R}^{n\times n},
\end{equation}
which is block-diagonal given that the subnets do not directly communicate. %to be a block diagonal matrix and denote $w_{ij}\in[0,1]$ to be the entry of W at the $i$\textsuperscript{th} row and $j$\textsuperscript{th} column. The matrix $W$ is a block diagonal matrix because the graph formed by D2D communications are $l$ separate graphs, forming $l$ blocks.
In Sec.~\ref{sec:IV-A}, we will discuss further assumptions on the subnet matrices $W_s$ for our convergence analysis.

The {\tt SD-GT} training process consists of two timescales. The outer timescale, $t = 1,2,\ldots,T$, indexes global aggregations carried out through DS communications. The inner timescale, $k = 1, \ldots, K$, indexes local training and aggregation rounds carried out via D2D communications. 
% \st{Without loss of generality}\shiqiang{There is a strict mathematical definition of ``without loss of generality'' (see related Wikipedia page), which I don't think applies there},
We assume a constant $K$ local rounds occur between consecutive global aggregations.

\textcolor{black}{Consistent with prior work in semi-decentralized FL~\cite{lin2021semi,hosseinalipour2022multi,wang2021network}, we adopt a synchronous communication model. With an orthogonal resource allocation protocol like OFDMA or FDMA, devices can transmit their local models to neighbors concurrently. Each device then waits to receive updates from its neighbors before performing its local aggregation (described in Sec. III-B), synchronizing the update round index $k$ across the network based on the slowest link. Such synchronization models are common in OFDM-based D2D systems~\cite{kang2014adaptive,brinton2025key}. Future work can investigate the impact of enabling asynchronous operation due to heterogeneity in communication and computation capabilities.
% Consistent with prior work in semi-decentralized FL~\cite{lin2021semi,hosseinalipour2022multi,wang2021network}, we adopt a synchronous communication model to enable tractable analysis and controlled evaluation. In practice, devices can transmit local models to neighbors over dedicated or orthogonal resources (e.g., FDMA/OFDMA), allowing parallel transmissions. Each device then waits to receive updates from its neighbors before performing local aggregation, synchronizing the update round index $k$ across the network based on the slowest link. While this introduces delays tied to worst-case latency, it does not cause protocol-level asynchrony. Such synchronization is practical in OFDM-based D2D systems~\cite{kang2014adaptive,brinton2025key}.
}

\vspace{-0.1in}
\subsection{Learning Model}
\label{sec:IIIB}
Algorithm \ref{alg:1} summarizes the full {\tt SD-GT} procedure. Each client maintains two gradient tracking terms, $y_i^t \in \mathbb{R}^d$ and $z_i^t \in \mathbb{R}^d$, which track (i) the gradient information between different subnets and (ii) the gradient information inside each subnet, respectively. These two variables act as corrections to the local gradients so that the update direction can guarantee convergence towards global optimum, visualized in Figure \ref{fig2}.
 
\begin{algorithm}[t]
{\small
\caption{{\tt SD-GT}: Semi-Decentralized Gradient Tracking}
\label{alg:1}
\KwConstants{step size $\gamma > 0$, initial model parameter $x^0$}
\KwOutput{$x_\mathrm{g} ^t$}
\textbf{Local parameter initialization}\\
$x_i^{1,1} \gets x^0, \quad \forall i$\\
$y_{i}^1 \gets \frac{1}{n}\sum_{j=1}^n \nabla f_i(x^0, \xi_i^0) - \frac{1}{|\mathcal{C}_s|}\sum_{j\in \mathcal{C}_s} \nabla f_i(x^0, \xi_i^0), \quad \forall i \in \mathcal{C}_s, \forall s \in S$,\\
$z_i^1 \gets \frac{1}{|\mathcal{C}_s|}\sum_{j\in \mathcal{C}_s} \nabla f_i(x^0, \xi_i^0) - \nabla f_i(x^0, \xi_i^0), \quad \forall i \in \mathcal{C}_s, \forall s \in S$\\
\textbf{Server parameter initialization}\\
$x_\mathrm{g} ^1 \gets x^0$\\
$\psi_{s}^1 \gets 0, \quad \forall s$\\
\For{$t\leftarrow 1,\ldots,T$}{
\tcc{Step 1: Within-subnet Model Update}
\For{$k\leftarrow 1,\ldots,K$}{
\ClientFor{$i \gets 1,\ldots,n$}{
Perform within-subnet update on the local model $x_i^{t,k}$ based on \eqref{eq3} and \eqref{eq4}.

}
}
\tcc{Step 2: Local Tracking Term Update}
\ClientFor{$i \gets 1,\cdots,n$}{
    Update the within-subnet gradient tracking variable $z_i^t$ based on \eqref{eq5} and \eqref{eq6}.
}
\tcc{Step 3: Global Aggregation}
sample $h_{s}$ random clients $x_{s,j} \sim \mathcal{C}_s$ from every subnet\\
\SubnetFor{$s\gets 1,\ldots,S$}{
\AggrFor{$j \leftarrow 1,\ldots,h_{s}$}{
Use \eqref{eq7} to \eqref{eq10} to update variables $x_\mathrm{g} ^t$ and $\psi_{s}^t$, $\forall s$ on the server, then broadcast them to all sampled clients.
}
}
% \textbf{All clients that are not sampled:}\\
\If{client $i$ is not sampled}{
    $x_i^{t+1,1} \gets x_i^{t,K+1}$\\
    $y_i^{t+1} \gets y_i^t$
}
}

}\end{algorithm}

\textbf{Within-subnet Updates.} We denote $x_i^{t,k} \in \mathbb{R}^d$ as the model parameter vector stored at client $i$ during the $t$\textsuperscript{th} global aggregation round and $k$\textsuperscript{th} D2D communication round. Each device conducts its local model update in two steps: (i) updating its local model using its local gradient and gradient tracking terms, and then (ii) computing a linear combination with models received from its neighbors, also known the Adapt-Then-Combine (ATC) scheme. ATC is known to have a better performance compared to other mixing schemes\cite{tu2012diffusion}. The update direction not only includes the gradient direction computed from the local dataset $\mathcal{D}_i$ but also the gradient tracking terms $z_i^t$ and $y_i^t$. In other words:
% \shiqiang{Equations should be part of a sentence and have a punctuation at the end, see equation (1) as an example}
\begin{align}
    \textstyle x_{i}^{t,k+\frac{1}{2}} &\textstyle= x_{i}^{t,k} - \gamma \left( \nabla f_i(x_{i}^{t,k}, \xi_i^{t,k}) + y_{i}^{t} + z_i^{t}\right), && \forall i
    \label{eq3}\\
    \textstyle  x_{i}^{t,k+1} &\textstyle= \sum_{j \in \mathcal{N}_i\cup \{i\}} w_{ij}x_{j}^{t,k+\frac{1}{2}}, && \forall i,
    \label{eq4}
\end{align}
% \shiqiang{Perhaps mention that $\nabla f_i(x_{i}^{t,k}, \xi_i^{t,k})$ denotes the stochastic gradient}
where $\nabla f_i(x_{i}^{t,k}, \xi_i^{t,k})$ denotes the stochastic gradient of $\nabla f_i(x_{i}^{t,k})$.

\begin{figure}
\centering
\includegraphics[width=0.5\textwidth]{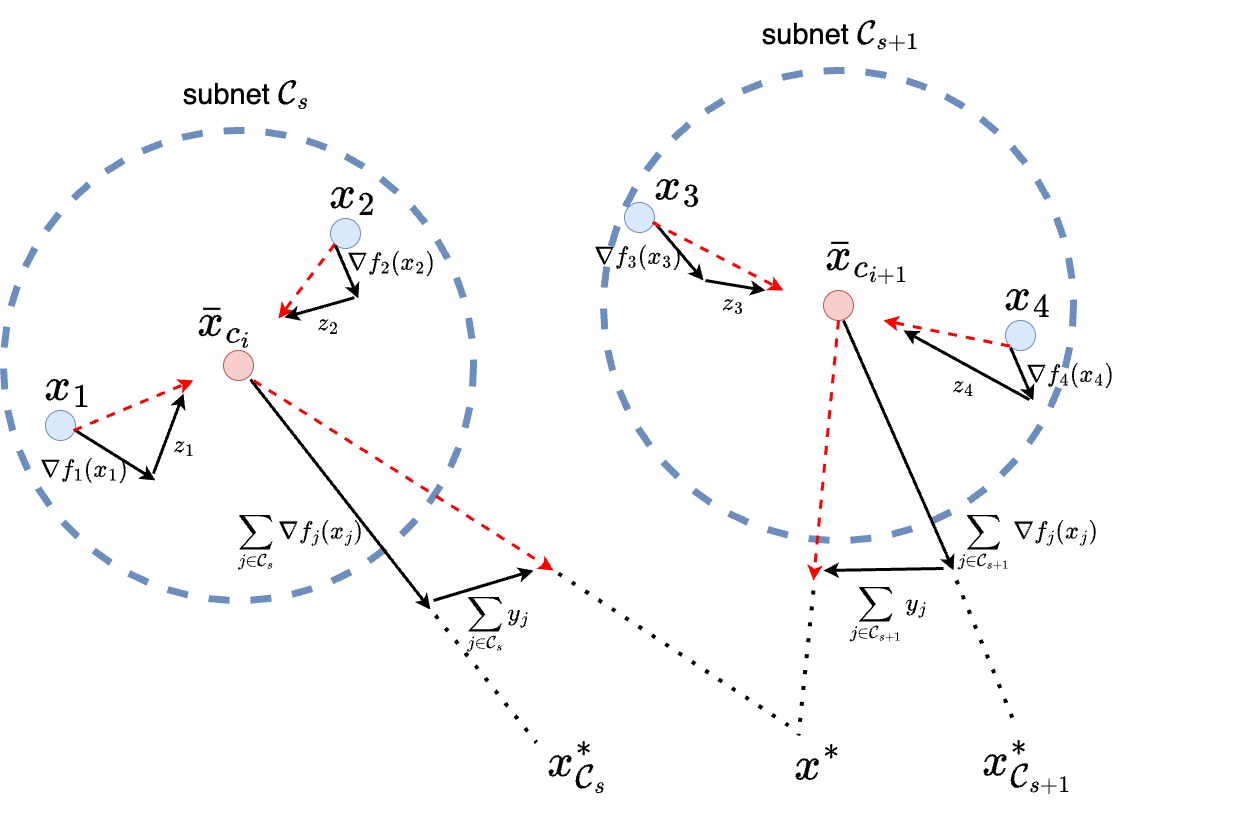}
\caption{\textcolor{black}{An illustration of how {\tt SD-GT} mitigates subnet-drift. By introducing the within-subnet GT term $z_i^t$, all devices within each subnet are able to course-correct towards a consensual location of the subnet, mitigating the gradient difference between local updates and the average subnet direction. Further, the between-subnet GT term $y_i^t$ corrects the drift between subnets and the global gradient direction that arises due to inter-subnet heterogeneity, so that each subnet $\mathcal{C}_s$ no longer converges towards the optimal solution $x_{\mathcal{C}_s}^{\star}$, but to the optimal solution $x^{\star}$ of the whole network. Both gradient tracking terms are added during each round of local update, steering model update directions toward the global optimum.}\vspace{-0.15in}}
\label{fig2}
\end{figure}

Additionally, after every D2D communication, each client computes $\Tilde{z}_i^{t,k}$, a vector which measures the change in its local model $x_i^{t,k}$ excluding the impact of the global tracking term $y_i^t$. Formally,
\begin{align}
    \textstyle \Tilde{z}_i^{t,k}  &\textstyle= x_{i}^{t,k+\frac{1}{2}} - x_{i}^{t,k} + \gamma y_i^t, && \forall i,
    \label{eq5}
\end{align}
where $\Tilde{z}_i^{t,k}$ is needed to update the within-subnet tracking term $z_i^t$. Although $\Tilde{z}_i^{t,k}$ is computed after every D2D communication, we only pass the sum of all $K$ rounds to client $i$'s neighbor set $\mathcal{N}_i$ once every global iteration $t$ to update $z_i^{t+1}$. 
%After the $K$ rounds of D2D communications, each client updates its own in-subnet gradient tracking term $z_i^t$ using $\Tilde{z}_i^{t,k}$, which is the measured change in the variable $x_i^{t,k}$ every time a D2D communication within the subnet is performed.
This update is given by
% \evan{align the equal signs}
% \evan{change $\delta$ z to Tilde z?}
\begin{align}
    \textstyle z_i^{t+1} &\textstyle= z_i^t + \frac{1}{K\gamma}\sum_{k=1}^K\left(\Tilde{z}_i^{t,k} - \sum_{j \in \mathcal{N}_i \cup \{i\}} w_{ij}\Tilde{z}_j^{t,k}\right), && \forall i.\label{eq6}
\end{align}
The global GT term helps correct for the drift between subnets and server gradients due to inter-subnet heterogeneity, while the local GT term compensates for discrepancies caused by local updates and D2D consensus. Together, these components cooperatively guide the local models toward the global optimum by maintaining alignment with the global descent direction throughout the training process.

\textbf{Global Aggregation.} The central server will choose a number of clients $h_{s} \in \{1,...,m_s\}$ to sample from each subnet $s$, e.g., based on the DS communication budget. These will be a set of variables in our optimization considered in Sec.~\ref{sec:IV-E}. Those clients that are not sampled by the central server for round $t$ will not update their parameters, and maintain $x_i^t$ and $y_i^t$ into the next communication round. 
% \shiqiang{What happens to clients that are not sampled? Do they still compute their in-subnet updates as above?}\evan{They will maintain their values into the next communication round (mentioned next page at the ``All the clients that are'' paragraph)}\shiqiang{My question is whether they will still compute or not. It is not clear from the writing. This is an important aspect that should be mentioned before this paragraph when talking about in-subnet updates} 
%The sample number $h_{s} \in \{1,...,m_s\}$ is at least one and at most the total number of clients in the subnet. 

We denote $x_\mathrm{g}^t$ as the global model that is stored at the server. At each global aggregation $t$, the server accumulates the gradient information collected from each subnet, and updates the between-subnet gradient tracking terms that are stored on the server, which we denote by $\psi_{s}^t$ for subnet $s$. Specifically, for each subnet $s$, let the sampled clients for round $t$ be the set $\mathcal{C}'_s \subseteq \mathcal{C}_s$ with size $|\mathcal{C}'_s| = h_{s}$. To update $x_\mathrm{g}^t$ and $\psi_{s}^t$, we first compute two intermediate quantities, $\textstyle\Tilde{x}_{i}^{t+1}$ (managed at the client-side) and $\textstyle\Tilde{x}_\mathrm{g} ^{t+1}$ (managed at server-side):
\begin{align}
    &\textstyle\Tilde{x}_{i}^{t+1} = x_{{i}}^{t, K+1} - x_{{i}}^{t, 1}+ K\gamma y_{{i}}^t,&& \forall i, \label{eq7}\\
    &\textstyle\Tilde{x}_\mathrm{g} ^{t+1} = \frac{1}{S\cdot h_{s}}\sum_{s = 1}^{S}\sum_{j \in \mathcal{C}'_s}  \Tilde{x}_{{j}}^{t+1}. \label{eq8}
\end{align}
In \eqref{eq7}, each client conducts a cancellation of the between-subnet gradient tracking information so that the global model receives unbiased gradient information. In \eqref{eq8}, all the aggregated information is averaged for the server to update the global model. Then, we have:
\begin{align}
        &\textstyle x_\mathrm{g} ^{t+1} = x_\mathrm{g} ^t + \Tilde{x}_\mathrm{g}^{t+1},\label{eq9}\\
        &\textstyle \psi_{{s}}^{t+1} = \frac{1}{K\gamma}(\frac{1}{h_{s}}\sum_{j\in \mathcal{C}'_s}\Tilde{x}_{{j}}^{t+1} - \Tilde{x}_\mathrm{g} ^{t+1}), && \forall s,\label{eq10}
\end{align}
% \shiqiang{You may need $\forall$ before the ranges of $j$ and $s$ on the right of equations. Also, why are the ranges given here but not for the equations in in-subnet updates and not for the following equations either?} \evan{Do you think I should remove them from these equations or add them to the in-subnet updates}\evan{ok}
where  \eqref{eq9} and \eqref{eq10} are both computed at the server. Finally, the server broadcasts the updated global model $x_\mathrm{g}^{t+1}$ and between-subnet gradient tracking terms $y_{i}^{t+1}$ to the sampled clients to complete the synchronization:
\begin{align*}
        \textstyle x_{{i}}^{t+1,1} &\textstyle= x_\mathrm{g} ^{t+1}, \quad \forall i \in \mathcal{C}'_s, \forall s \in S,\\
        \textstyle y_{{i}}^{t+1} &\textstyle= \psi_{s}^{t+1}, \quad \forall i \in \mathcal{C}'_s, \forall s \in S.
\end{align*}

Maintaining two gradient tracking terms is an essential feature of {\tt SD-GT} for stabilizing convergence, and is one of our contributions mentioned in Sec.~\ref{ssec:contribution}. In particular, if we only used the between-subnet measure $y_i^t$ to track gradient information, then the gradient information within each subnet would deviate from the average of the subnet, preventing the system from converging towards the global minimum. On the other hand, if we only used the within-subnet measure $z_i^t$ to track gradient information, then each subnet $s$ will tend to converge towards its local minimum $x_{s}^\star$ instead of the global minimum.
% \evan{explain how the two GT terms controls the convergence. The uniqueness.} \evan{perhaps add the numbers, and in alg 1 refer to the numbers (but need more explanations).}
\vspace{-0.1in}
\subsection{Connection with Existing Methods}
\label{sec:IIIC}
Some existing methods can be shown to be special cases of {\tt SD-GT} under certain network structures.
% In this subsection, we show the connection of our algorithm with existing methods: Under certain network topology, one of the two gradient tracking terms $z_i$ and $y_i$ will be zero, and the update process will align with Proxskip\cite{mishchenko2022proxskip} and gradient tracking\cite{sun2016distributed,di2016next}.

\textbf{Case 1: Conventional FL with gradient tracking ($S = n$).}
Under this setting (when $S = n$), the server always samples the full subnet since each subnet contains only one client, forming the conventional FL star topology. Then we can see that with the initialization $z_1^1 = z_2^1 = \ldots = z_n^1 = 0$, the in-subnet gradient tracking terms are always zero:
\begin{align*}
    \textstyle z_i^{t+1} = z_i^t + \frac{1}{K\gamma}\sum_{k=1}^K(\Tilde{z}_i^{t,k} - \Tilde{z}_i^{t,k}) = z_i^t = 0, \quad \forall i.
\end{align*}
The global gradient tracking term $y_{i}^t$ can be formulated as:
\begin{align*}
    \textstyle y_{i}^{t+1} = \psi_{s}^{t+1} & \textstyle= \frac{1}{K\gamma}(\sum_{j \in \mathcal{C}'_s} \frac{1}{h_{s}}\Tilde{x}_{{j}}^{t+1} - \Tilde{x}_\mathrm{g} ^{t+1})\notag\\
    &\textstyle= y_{i}^{t} + \frac{1}{K\gamma}(x_{s}^{t,K+1} - x_\mathrm{g} ^{t+1}), \quad \forall s.
\end{align*}
With the in-subnet gradient tracking term being zero and the global gradient tracking term in the form above, our algorithm aligns with ProxSkip~\cite{mishchenko2022proxskip} under a deterministic communication frequency.

\textbf{Case 2: Fully decentralized learning with gradient tracking ($S = 1, K = 1$).}
Under this setting (when $S = 1$, $K = 1$), the global gradient tracking term is always zero since $\sum_{j=1}^{h_s} \frac{1}{h_s}\Tilde{x}_{s,j}^{t+1} = \Tilde{x}_\mathrm{g} ^{t+1}$:
\begin{equation*}
    \textstyle \psi_{s}^{t+1} = \psi_{s}^{t} + \frac{1}{\gamma}(\sum_{j\in \mathcal{C}'_s}^{h_{s}} \frac{1}{h_c}\Tilde{x}_{{j}}^{t+1} - \Tilde{x}_\mathrm{g} ^{t+1}) = y_{s}^{t} = 0, \quad \forall s.
\end{equation*}
Since the number of D2D rounds are set to one, the update of $z_i^t$ for every client $i$ can be formulated as:
\begin{align*}
    \textstyle z_i^{t+1} \textstyle=& z_i^t + \frac{1}{K\gamma}\sum_{k=1}^K(\Tilde{z}_i^{t,k} - \sum_{j \in \mathcal{N}_i \cup \{i\}} w_{ij}\Tilde{z}_j^{t,k})\\
    =&\textstyle \sum_{j \in \mathcal{N}_i \cup \{i\}} w_{ij}(z_j^t+ \nabla f_j(x_j^t)) - \nabla f_i(x_i^t),\quad\forall i .
\end{align*}
If we define $\hat{z}_i^t = z_i^t + \nabla f_i(x_i^t)$, we are left with:
\begin{align*}
    \textstyle x_i^{t+1} &=\textstyle \sum_{j \in \mathcal{N}_i \cup \{i\}} w_{ij}(x_j^{t} - \gamma \hat{z}_i^t), &&\forall i,\\
    \textstyle\hat{z}_i^{t+1} &= \sum_{j \in \mathcal{N}_i \cup \{i\}} w_{ij}\hat{z}_j^t + \nabla f_i(x_i^{t+1}) -  \nabla f_i(x_i^t),&&\forall i,
\end{align*}
% \shiqiang{It is not clear why this equation has $\{$ on the left while the others don't. I think you can remove $\{$}
which aligns with the conventional gradient tracking algorithm \cite{sun2016distributed,di2016next}, assuming that {\tt SD-GT} does not conduct a global aggregation.

\vspace{-0.1in}
\section{Analysis and Optimization}
\label{sec:IV}
We first outline the assumptions used for the proofs (Sec. \ref{sec:IV-A}), then we provide the convergence analysis of {\tt SD-GT} under non-convex, weakly-convex, and strongly-convex ML loss functions, without making any assumptions on data heterogeneity (Sec. \ref{sec:IV-B}). Then, we provide a proof sketch of the theorems (Sec. \ref{sec:IV-C}). Finally, we derive a co-optimization algorithm that considers the trade-off between communication cost and performance (Sec. \ref{sec:IV-E}).

Note that we defer the full proofs of all theorems, corollaries, and lemmas to the supplemental material.
\vspace{-0.1in}
\subsection{Convergence Analysis Assumptions}
\label{sec:IV-A}
We first present the general assumptions that are applied to all three theorems in this paper.
% \chris{Give references to related works using these assumptions}

\begin{assumption}\textbf{(L-smooth)} Each local objective function $f_i$:
$\mathbb{R}^d \rightarrow \mathbb{R}$ is L-smooth:
\begin{equation*}
    \|\nabla f_i(x) - \nabla f_i(y)\|_2 \leq L\|y - x\|_2, \quad \forall x,y\in \mathbb{R}^d.
\end{equation*}

\label{asmp1}
\end{assumption}
\begin{assumption}\textbf{(Mixing Rate)} Each subnet $\mathcal{C}_1, \ldots, \mathcal{C}_S$ has a strongly connected graph, with doubly stochastic weight matrix $W_{s} \in \mathbb{R}^{m_{s} \times m_{s}}$. Defining $J_{s} = \frac{\mathbf{11^\top}}{m_{s}}$, there exists a constant $\rho_{s}\in (0,1]$ such that
\begin{equation*}
    \| X (W_{s} - J_{s})\|_F^2 \leq (1 - \rho_{s})\| X(I - J_{s})\|_F^2, \quad \forall X \in \mathbb{R}^{d \times m_{s}}.
\end{equation*}

\label{asmp2}
\end{assumption}
\begin{assumption}\textbf{(Bounded Variance)}
Variances of each client's stochastic gradients are uniformly bounded.
\begin{align*}
    \mathbb{E}_{\xi \sim \mathcal{D}_i}\|\nabla f_i(x;\xi) - \nabla f_i(x)\|_2^2 \leq \sigma^2, \forall i \in [1, n], \forall x \in \mathbb{R}^d.
\end{align*}

\label{assmp3}
\end{assumption}

\textcolor{black}{These assumptions are widely used in the theoretical analysis of FL~\cite{wang2019adaptive,lin2021semi,karimireddy2020scaffold} and of stochastic optimization and gradient-based methods more generally~\cite{nesterov2013introductory,stich2018local,mishchenko2022proxskip}. While computing exact global 
$L$-smoothness constants for deep networks is generally intractable, this issue can be mitigated via smooth activation functions or spectral normalization. More importantly, our algorithm does not require knowledge of the value of smoothness constant in practice.}

\textcolor{black}{
Assumption \ref{asmp2} can be extended to strongly connected time varying-networks by assuming $\rho_s = \min_t \{\rho_s^t\}$, where $\rho_s^t$ is the connectivity constant for subnet $s$ during round $t$. A more detailed discussion regarding time-varying node mobility and link bandwidth limitations is given in Appendix~\ref{appen:discuss_topology}.
}
\vspace{-0.1in}
\subsection{Convergence Analysis Results}
\label{sec:IV-B}
Here we provide the obtained convergence bounds for our algorithm in non-convex, weakly-convex and strongly-convex settings. To facilitate this, we define $\overline{x}_g^t = \mathbb{E}_{\mathcal{C}'_s}[x_g^t]$ to be the expected server model taken over the sample sets $\mathcal{C}'_s$ across subnets $s$. %A sketch of the proof will be presented in Sec. \ref{sec:IV-D}.
% \textcolor{red}{TODO:} Redo the proof and write sketch.

\begin{theorem}\textbf{(Non-convex)}
Under Assumptions \ref{asmp1}, \ref{asmp2}, and \ref{assmp3}, let $\beta_{s} = \frac{m_{s} - h_{s}}{m_{s}}$ be the ratio of unsampled clients from each subnet. Define $p = \min(1 - \beta_{1}^2 , \ldots,1 - \beta_{S}^2)\in (0,1]$, $q = \min(\rho_{1}, \ldots, \rho_{S})\in (0,1]$, and the function value optimality gap for time $t$ as $\mathcal{E}_t = \mathbb{E}f(\overline{x}_\mathrm{g} ^t) - f(x^\star)$. Then, with a constant step size $\gamma < \mathcal{O}(\frac{p^2q}{KL})$, for any $T > K$, we have:
{\begin{align}
    \textstyle \frac{1}{T}\sum_{t=1}^T\mathbb{E}\|\nabla f(\overline{x}_\mathrm{g} ^t)\|^2 
    \leq &\textstyle\mathcal{O}\bigg(\frac{\mathcal{E}_1}{TK\gamma} + \frac{L\gamma\sigma^2}{2n}     +\frac{L^2K\gamma^2\sigma^2}{p^4q^2}\bigg).
\end{align}}
\label{thm1}
\end{theorem}

In Theorem \ref{thm1}, $p$ captures the lowest sampling ratio across subnets by the server, and $q$ captures the lowest information mixing ability of D2D communications across subnets. Large $p$ indicates the server samples a large amount of clients, and a large $q$ indicates the connectivity of every subnet is high.
% Theorem \ref{thm1} shows that our algorithm converges to a radius around some stationary
% point and that the radius can be controlled by the step size \shiqiang{It's a bit risky to mention this since it may give the impression that your algorithm may not converge to zero error}.\evan{But isn't this is a common result for stochastic gradient?}\shiqiang{People who are not familiar with details of such analysis may not know about it.} \evan{OK, I'll remove this part.}

\textcolor{black}{The step size condition $\gamma < O\left(\frac{p^2 q}{KL}\right)$ arises naturally from the Lyapunov-based convergence analysis that we will describe in Sec.~\ref{sec:IV-C}. It can be compared with existing literature in federated optimization, where smoothness constants $L$, local update rounds $K$, and network topology terms $q$ appear in convergence bound conditions~\cite{karimireddy2020scaffold,lin2021semi}.}
By carefully choosing the step size, we can further achieve the following result, guaranteeing convergence to a stationary point.
\begin{corollary}
\label{corr:noncon}
Under the same conditions as in Theorem \ref{thm1}, by choosing a constant step size
$$\gamma = \min\left\{\left(\frac{2\mathcal{E}_1 Kn}{L\sigma^2 T}\right)^{1/2}, \left(\frac{2\mathcal{E}_1 p^4q^2 K}{L\sigma^2 T}\right)^{1/3}, \frac{p^2q^2}{945 KL}\right\},$$ {\tt SD-GT} obtains the following rate:
\begin{align}
    &\textstyle\frac{1}{T}\sum_{t=1}^T\mathbb{E}\|\nabla f(\overline{x}_\mathrm{g} ^t)\|^2 
    = \mathcal{O}\bigg(\sqrt{\frac{\mathcal{E}_1\sigma^2L}{nTK}} + (\frac{\mathcal{E}_1L\sigma}{\sqrt{K}Tp^2q})^{\frac{2}{3}} + \frac{\mathcal{E}_1L}{Tp^2q}\bigg).
\label{eq28}
\end{align}
\vspace{-0.15in}

\label{cor1}
\end{corollary}
% \begin{proof}
%     Follows from plugging Lemma \ref{unroll_lem} into Theorem \ref{thm1}.
%     % \evan{capitalize the word lemma, add space between numbers}
% \end{proof}
% Theorem \ref{thm1} shows the sublinear convergence of the central server's model $x_g^t$ towards a stationary point when $T$ is sufficiently large. 
% We're also able to remove the need of any data-heterogeneity assumption, which means that even if the dataset distributed throughout the network is non-IID, the convergence results will still be bounded under this rate.
In \eqref{eq28}, the first two terms on the right side capture the effect of stochastic gradient variance $\sigma$ on convergence. Choosing a larger number of local D2D rounds $K$ decreases the effect of stochasticity on the convergence. We can also observe that the bound becomes better when the values of $p$ and $q$ are large. Note that the network topology which would give the smallest bound is $p = q = 1$, which is also the most resource-inefficient case: $p = 1$ means the server samples all clients from each subnet every time DS communication is performed, and $q = 1$ means the topology of every in-subnet D2D communication is a fully connected graph. This emphasizes the importance of learning-efficiency co-optimization for {\tt SD-GT}, as we will consider through our control algorithm in Sec.~\ref{sec:IV-E}.
% \chris{We need to try to spice this up a bit. Take a look at pages 15-16 in https://arxiv.org/ab/2305.13503. Try to bring out some of the interesting points in the different terms. Maybe the fact that $p_1$ and $p_2$ are multiplied? A joint effect?}

\textcolor{black}{
We now present a stricter assumption~\cite{tian2018asy} that guarantees a better convergence rate in Theorem \ref{thm2} and \ref{thm3}. 
}

\begin{assumption}\textbf{($\mu$-strongly-convex)}
    Every local objective function $f_i: \mathbb{R}^d \rightarrow \mathbb{R}$ is $\mu$-strongly convex with $0 \leq \mu \leq L$.
\begin{equation*}
    f(y) \leq f(x) + \nabla f(x)^\top(y-x) + \frac{\mu}{2}\|y - x\|^2, \quad \forall x,y\in \mathbb{R}^d.
\end{equation*}

\label{assmp4}
\end{assumption}
\textcolor{black}{Although most overparameterized neural networks have non-convex loss landscapes, strong convexity remains relevant in domains employing shallow neural networks to prioritize fast convergence and solution stability. As we will see, it is possible to show a linear convergence rate for {\tt SD-GT} in this special case (Theorem 3).}
By adding the assumption of convexity, we are able to get a stronger convergence result.
\begin{theorem}\textbf{(Weakly-convex)}
Under Assumptions \ref{asmp1}, \ref{asmp2}, \ref{assmp3}, and \ref{assmp4} with $\mu = 0$, with a constant step size $\gamma < \mathcal{O}(\frac{p^2q^2}{K^{3/2}L})$ and any $T > K$, we have
\begin{align}
       &\textstyle\frac{1}{T}\sum_{t=1}^{T}\mathbb{E} (f(\Bar{x}_g^t) - f(x^\star))\notag\\
       % &\mathcal{O}\bigg(\frac{\mathbb{E}||\overline{x}_g^1 - x^*||^2}{\gamma KT} + \frac{\mathcal{G}_1}{T} + \frac{\gamma  \sigma^2}{n}  + \frac{\gamma^2 KL \sigma^2}{p^4q^4}\bigg)\\
       =& \mathcal{O}\bigg(\textstyle\frac{\mathbb{E}\|\overline{x}_g^1 - x^\star\|^2}{\gamma KT}
       \textstyle+ \textstyle\frac{\gamma  \sigma^2}{n} + \frac{\gamma^2 KL \sigma^2}{p^4 q^4}\bigg).
\end{align}
\label{thm2}
\end{theorem}
The convergence rate of Theorem \ref{thm2} is similar to that of Theorem~\ref{thm1}, but the bound indicates convergence towards the global optimal solution instead of a stationary point. 
By choosing a specific step size, we obtain the following result, which shows convergence to the optimal global model.
\begin{corollary} Under the same conditions as in Theorem \ref{thm2}, by choosing a constant step size
\begin{align}\gamma =& \min\bigg\{ \bigg( \frac{\mathbb{E}\|\overline{x}_g^1 - x^\star\|^2 Kn}{\sigma^2 T} \bigg)^{1/2},\\\notag &\frac{1}{39} \bigg(\frac{\mathbb{E}\|\overline{x}_g^1 - x^\star\|^2 K p^4 q^4}{L\sigma^2 T} \bigg)^{1/3}, \frac{p^2q^2}{101 K^{3/2}L}\bigg\},\end{align} {\tt SD-GT} obtains the following rate:
\begin{align}
    \textstyle
    \frac{1}{T}&\sum_{t=1}^{T}\mathbb{E} (f(\Bar{x}_g^t) - f(x^\star))
       = \mathcal{O}\bigg(\frac{ L\mathbb{E}\|\overline{x}_g^1 - x^\star\|^2}{T} \notag\\
       \textstyle&+\bigg(\frac{\sqrt{L}\mathbb{E}\|\overline{x}_g^1 - x^\star\|^2\sigma}{p^2q^2\sqrt{K}T}\bigg)^{\frac{2}{3}} + \sqrt{\frac{\mathbb{E}\|\overline{x}_g^1 - x^\star\|^2\sigma^2}{nKT}}\bigg).
       \label{new_eq14}
\end{align}
\label{cor2}
\end{corollary}
From \eqref{new_eq14}, we can see that the bound has a similar structure to the bound for the non-convex case in \eqref{eq28}, where there are two terms related to the stochasticity of the gradients, and one term $\mathcal{O}(\frac{1}{T})$ proportional to the initial error. 

If we further assume strong convexity, we can obtain a faster rate of convergence, provided the gradient noise is small.

\begin{theorem}\textbf{(Strongly-convex)}
Under Assumptions \ref{asmp1}, \ref{asmp2},\ref{assmp3}, and \ref{assmp4} with $\mu > 0$, for a constant step size $\gamma \leq \overline{\gamma}$, where
\begin{align}
    \overline{\gamma} &= \min\bigg\{\frac{\min(p,q)\mu}{K(14L^2+240L^3)},\frac{1}{18KL}, \frac{4}{\mu K},\notag\\
    &\frac{\min(p, q)p}{2(45KL + 108KL^2 + K\mu)}, \frac{\min(p,q)^2}{2(86K^2+864K^2L+K\mu)}\bigg\}, \notag\\
\end{align}
then we have
\begin{align}
\textstyle\mathbb{E}\|\overline{x}_\mathrm{g} ^{T+1} \!- \!x^\star\|^2 &= \mathcal{O}\bigg( (1 - \frac{\mu K\gamma}{4})^T\mathbb{E}\|\overline{x}_\mathrm{g} ^1 - x^\star\|^2 \notag\\
    \textstyle&+ (1 - \frac{\mu K\gamma}{4})^T \gamma \sigma^2 + \|(I - A)^{-1}b\|_1^2\bigg),
\label{eq:24}
\end{align}
\label{thm3}
where $I$ is the identity matrix,
% \evan{expand the spacing between rows}
\small
\setlength{\arraycolsep}{1pt} % default: 5pt
\medmuskip = 1mu
%
% \begin{gather*}
%     \renewcommand*{\arraystretch}{2}
%     A=\begin{bmatrix}
%         1 - \mu K\gamma/2 & 9\gamma KL(1 \!-\! p) & 72K^3L\gamma^2 & 72K^3L\gamma^2\\
%         \frac{14}{p}\gamma^2K^2L^2 & 1 - \frac{p}{2} + \frac{36}{p}\gamma^2KL^2 & \frac{14}{p} K^2\gamma & \frac{14}{p} K^2\gamma\\
%         \frac{72}{p}\gamma^3K^3L^4 & \frac{30}{p}L^2K\gamma & 1 - \frac{p}{2} + \frac{240}{p}K^3\gamma^2L^2 & \frac{240}{p}K^3\gamma^2L^2\\
%         \frac{168}{q}\gamma^3K^3L^4 & \frac{78}{q}L^2K\gamma & \frac{624}{q}K^3\gamma^2L^2 & 1 - \frac{q}{2} + \frac{624}{q}K^3\gamma^2L^2\\
%     \end{bmatrix}
% \end{gather*}
%
\begin{gather*}
    \renewcommand*{\arraystretch}{1.5}
    A=\left(1-\frac{p}{2}\right)I +
    \gamma KL\begin{bmatrix}
        \frac{p-\mu \gamma K}{2\gamma KL} & 9(1 \!-\! p) & 72K^2\gamma & 72K^2\gamma\\
        \frac{14\gamma KL}{p} &  \frac{36\gamma L}{p} & \frac{14K}{pL} & \frac{14K}{pL} \\
        \frac{72\gamma^2K^2L^3}{p} & \frac{30L}{p} & \frac{240K^2\gamma L}{p} & \frac{240K^2\gamma L}{p}\\
        \frac{168\gamma^2K^2L^3}{q} & \frac{78L}{q} & \frac{624K^2\gamma L}{q} & \frac{624 K^2\gamma L}{q}\\
    \end{bmatrix},
\end{gather*}
\begin{gather*}
\renewcommand*{\arraystretch}{1.5}
    b = \begin{bmatrix}
        \frac{2\gamma^2K}{n} + 9K^2\gamma^3L\\
        K\gamma^2 + 3K^3\gamma^4L^2\\
        \frac{2\gamma}{qK} + \frac{30K^3\gamma^3L^2}{q}\\
        \frac{2\gamma}{qK} + \frac{78K^3\gamma^3L^2}{q}\\
    \end{bmatrix}\sigma^2.
\end{gather*}
\normalsize 
\end{theorem}

The second term on the RHS of \eqref{eq:24} is proportional to stochastic gradient noise $\sigma^2$, but can be controlled through an appropriate choice of step size $\gamma$. In the third term, the size of $b$ can also be limited through $\gamma$. By choosing a specific step size, we obtain the following result.
\begin{corollary}
    Under the same conditions as in Theorem \ref{thm3}, by choosing a constant step size
    \begin{align}
    \gamma = &\min\bigg\{\Bar{\gamma}, \frac{\ln(\max(1, \mu K(\mathbb{E}\|\overline{x}_\mathrm{g} ^{1} - x^\star\|^2 + \Bar{\gamma} 2\sigma^2)T/\sigma^2))}{\mu K T}\bigg\}\notag,\end{align} {\tt SD-GT} obtains the following rate:
{\begin{align}
    \textstyle\mathbb{E}\|\overline{x}_\mathrm{g} ^{T+1} - x^\star\|^2 \leq&\textstyle \Tilde{\mathcal{O}}\bigg(\exp( - pq\mu T)\cdot(\mathbb{E}\|\overline{x}_\mathrm{g} ^1 - x^\star\|^2 + \sigma^2) \notag\\
    &\textstyle+ \frac{\sigma^2}{\mu KT} + \frac{L^5\sigma^2}{\mu^5T^5pq}\bigg).
\label{eq:cor2}
\end{align}}
\label{cor3}
\end{corollary} 

Thus, under strong convexity,~\eqref{eq:cor2} demonstrates that our algorithm has a linear convergence rate under $\sigma^2 = 0$. In both Corollaries \ref{cor2} and \ref{cor3}, we observe a convergence towards the optimal solution when $T$ goes to infinity.

% \chris{Under proper initialization of the gradient tracking terms, all terms related to the heterogeneity will be either cancelled out or have the same order as the two stochastic noise related terms.}
\vspace{-0.1in}

\subsection{Proof Sketch and Key Intermediate Results}
\label{sec:IV-C}
Here we provide a proof sketch for Theorems \ref{thm1}, \ref{thm2}, and \ref{thm3}. The details of the proofs can be found in the supplemental material.

% \evan{Add the explanations/justifications of the assumptions.}

Recall that $\overline{x}_g^t$ is the expected global model with respect to the random client sampling process of the server. All of our theoretical analysis is based on iterative behavior of this expected model. When observing the progression of our algorithm, there are a few terms that characterize its deviation from our desired update direction. First, we have $\Gamma_t = \frac{1}{n}\sum_{i=1}^n\mathbb{E}\|x_i^{t-1,K+1} - \overline{x}_\mathrm{g} ^t\|^2$, which captures the error induced by the server only sampling part of the network instead of all devices. Next, let $Y^t = [y_1^t, \ldots, y_n^t] \in \mathbb{R}^{d\times n}$ and $Z^t = [z_1^t, \ldots, z_n^t] \in \mathbb{R}^{d\times n}$ be the collection of all gradient tracking variables $y_i^t$ and $z_i^t$ at iteration $t$. Further, we let $J_c = \textrm{diag}\{J_1, \ldots, J_S\} \in \mathbb{R}^{n\times n}$, with $J_s = \frac{11^\top}{m_s}$ being a full average matrix of subnet $s$, and $J = \frac{11^\top}{n} \in \mathbb{R}^{n\times n}$ being the full average matrix of the entire network. We define $\mathcal{Z}_t = \frac{1}{n}\mathbb{E}\|Z^t + \nabla F(\overline{x}_\mathrm{g} ^t)(I - J_c)\|^2_F$ and $\mathcal{Y}_t = \frac{1}{n}\mathbb{E}\|{Y^t} + \nabla F(\overline{x}_\mathrm{g} ^t)(J_c - J)\|_F^2$, which captures the performance of our within-subnet and between-subnet gradient tracking terms, respectively. Finally, we have $\Delta_t = \frac{1}{n}\sum_{i=1}^n\sum_{k=1}^K\mathbb{E}\|x_i^{t, k} - \overline{x}_\mathrm{g} ^{t}\|^2$ as the client drift term, which captures the total deviation across device models after performing $K$ rounds of D2D communication.

Now we can write out two descent lemmas that are used to obtain the convergence rate of our algorithm under non-convex vs. convex settings (the proofs of all lemmas are provided in Appendix~\ref{appendixA}).
\begin{lemma}(Descent lemma for non-convex settings)
     Under Assumption \ref{asmp1}, with a step size $\gamma \leq \frac{1}{4KL}$, the iteration of the global aggregation term $\overline{x}_g^t$ can be expressed as:
    \begin{align}
        \mathbb{E}f(\overline{x}_g^{t+1}) \leq & \mathbb{E}f(\overline{x}_g^{t})-\frac{\gamma K}{4}\mathbb{E}\|\nabla f(\overline{x}_g^t)\|^2 \notag\\
        &+ \gamma L^2\Delta_t + \frac{L\gamma^2K}{2n}\sigma^2.
    \end{align}
\label{lem5}
\end{lemma}
\begin{lemma} (Descent lemma for convex settings)
    Under Assumptions \ref{asmp1} and \ref{assmp4}, with a constant step size $\gamma < \frac{1}{36KL}$:
        \begin{align}
        \mathbb{E}\|\Bar{x}_g^{t+1} - x^\star\|^2
        \leq & \mathbb{E}\|\Bar{x}_g^t - x^\star\|^2 - \frac{\mu\gamma}{2}\mathbb{E}\|\Bar{x}_g^t - x^\star\|^2 \notag\\
        &- \gamma K  \mathbb{E} (f(\Bar{x}_g^t) - f(x^\star))\notag\\
        &+9 (1-p)\gamma KL \Gamma_t + 72K^3L\gamma^3 (\mathcal{Y}_t + \mathcal{Z}_t)\notag\\
        &+ \frac{2\gamma^2 K \sigma^2}{n} + 9K\gamma^3L\sigma^2.
\end{align}
\vspace{-0.1in}
\label{lem6}
\end{lemma}
Next, we need to bound iterative effects of quantities $\Delta_t$, $\Gamma_t$, $\mathcal{Y}_t$, and $\mathcal{Z}_t$ that appear in Lemmas~\ref{lem5} and \ref{lem6} under some appropriate choices of step size. Under $L$-smoothness, we can control $\Delta_t$ as follows.
\begin{lemma} (Deviation lemma)
    Under Assumption \ref{asmp1}, by selecting a step size $\gamma < \frac{1}{8KL}$, we have
\begin{align}
    \textstyle\Delta_t
        \leq&     \textstyle3(1 - p)K\Gamma_t
        + 24K^3\gamma^2\mathcal{Y}_t   \textstyle+ 24K^3\gamma^2\mathcal{Z}_t \notag\\
        &+ 6K^3\gamma^2\mathbb{E}\|\nabla f(\overline{x}_\mathrm{g} ^t)\|^2 + 3K^2\gamma^2 \sigma^2.
\end{align}
\vspace{-0.1in}
\label{lem1}
\end{lemma}
This lemma implies that by choosing an appropriate step size, the deviation at iteration $t$ of all client models from the average model after running $K$ rounds of D2D communication can be bounded by the other error terms at iteration $t$. We then can derive the bounds for the remaining terms.
\begin{lemma} (Between-subnet GT lemma)
Under Assumption \ref{asmp1}, with constant step size $\gamma < \frac{1}{\sqrt{6}KL}$, and $0 < p \leq 1$ defined as in Theorem~\ref{thm1}, we have: 
\begin{align}
        \mathcal{Y}_t \leq&  (1 - \frac{p}{2})\mathcal{Y}_{t-1} \! + \!\frac{10L^2}{p}\Delta_{t-1} \notag\\
        &\!+\! \frac{12}{p}\gamma^2L^2K^2\mathbb{E}\|\nabla f(\overline{x}_\mathrm{g} ^{t-1})\|^2 \!+\! \frac{2\sigma^2}{pK}.
\end{align}
\vspace{-0.1in}
\label{lem2}
\end{lemma}
\begin{lemma}
    (In-subnet GT lemma) Under Assumption \ref{asmp1}, with constant step size $\gamma < \frac{1}{\sqrt{6}KL}$, and $0 < q \leq 1$ defined as in Theorem~\ref{thm1}, we have: 
\begin{align}
        \mathcal{Z}_t \leq& (1-\frac{q}{2})\mathcal{Z}_{t-1} + \frac{26L^2}{q}\Delta_{t-1}
        \notag\\
        &+ \frac{12K^2L^2\gamma^2}{q}\mathbb{E}\|\nabla f(\overline{x}_\mathrm{g} ^t)\|^2 + \frac{2\sigma^2}{qK}.
        \label{eq:isgt}
    \end{align}
    \vspace{-0.1in}
\label{lem3}
\end{lemma}
\begin{lemma} (Sample Gap lemma)
Under Assumption \ref{asmp1}, and $0 < p \leq 1$ defined as in Theorem~\ref{thm1}, we have:
{\begin{align}
        \Gamma_t
        \leq& (1 - \frac{p}{2})\Gamma_{t-1} + \frac{12}{p}\gamma^2KL^2\Delta_{t-1} \notag\\
        &+ \frac{12}{p}\gamma^2K^2\mathcal{Y}_{t-1} + \frac{12}{p}\gamma^2K^2\mathcal{Z}_{t-1} \notag\\
        &+ \frac{12}{p}\gamma^2K^2\mathbb{E}\|\nabla f(\overline{x}_\mathrm{g} ^{t-1})\|^2+ K\gamma^2\sigma^2.
\end{align}
\vspace{-0.1in}}
\label{lem4}
% \label{lem7}
\end{lemma}
\textcolor{black}{Equipped with these lemmas, we are now ready to derive the theorems. For the non-convex case, we start with Lemma~\ref{lem5} and use the term $-\frac{\gamma K}{4}\mathbb{E}\|\nabla f(\overline{x}_g^t)\|^2$ as the descent direction. We then construct the following Lyapunov function capturing the between-subnet, within-subnet, and sampling errors:
\begin{equation}
\mathcal{H}_t = \mathcal{E}_t + c_0 K^3\gamma^3 \bigg(\frac{1}{p}\mathcal{Y}_t + \frac{1}{q}\mathcal{Z}_t\bigg) + c_1 \frac{K\gamma}{p}\Gamma_t.
\label{eq:lyapunov}
\end{equation}
}

\textcolor{black}{
This Lyapunov function is designed such that the joint effect of the two tracking terms $\mathcal{Y}_t$ and $\mathcal{Z}_t$ converges to zero. Through appropriate choice of step size, we also ensure the sample gap term $\Gamma_t$ converges to zero. In particular, we show in Appendix B that this Lyapunov function converges through an appropriate step size such that the following holds for some $c_0, c_1, c_2 > 0$:
\begin{align*}
        \mathcal{H}_{t+1} - \mathcal{H}_{t} \leq& -c_1 K\gamma \mathbb{E}\|\nabla f(\overline{x}_g^{t})\|^2 \notag\\
        &+ c_2 (K\gamma)\sigma^2 + c_3(K^2\gamma^2)\sigma^2.
\end{align*}
By unrolling the recursion and defining suitable values of $c_0, c_1, c_2, \gamma$, we obtain Theorem 1.
}

\textcolor{black}{The multi-term gradient tracking structure encountered here for semi-decentralized FL introduces novel challenges that we address in our theoretical development. The difficulty of this analysis lies in carefully constructing the terms $\mathcal{Z}_t$ and $\mathcal{Y}_t$. In non-hierarchical gradient tracking FL~\cite{karimireddy2020scaffold}, only one gradient tracking term exists, leading to a relatively straightforward treatment for ensuring convergence. However, in our hierarchical setting, improper definition of $\mathcal{Z}_t$ and $\mathcal{Y}_t$ may induce convergence towards non-zero values. In our work, we show that with our choices of $\mathcal{Z}_t$ and $\mathcal{Y}_t$, we can derive results showing that those terms jointly converge with proper choice of step size (see Lemmas \ref{lem2} and \ref{lem3} and their proofs in Appendix A). Combining these results with Lemma~\ref{lem4}, where we control the behavior of the errors introduced by sampling $\Gamma_t$, leads to convergence of the Lyapunov function.}

For the convex case, we start with Lemma \ref{lem6}. We use $- \gamma K  \mathbb{E} (f(\Bar{x}_g^t) - f(x^\star))$ as the descent direction for the weakly convex case ($\mu = 0$), and $- \frac{\mu\gamma}{2}\mathbb{E}\|\Bar{x}_g^t - x^\star\|^2 $ as the descent direction for the strongly convex case ($\mu > 0$). For the weakly convex case, the process is similar to the non-convex case, but instead of unrolling the recursion using $\mathbb{E}\|\nabla f(\overline{x}_g^{t})\|^2$, we use $\mathbb{E} (f(\Bar{x}_g^t) - f(x^\star))$, which gives us convergence towards the global optimum instead of a stationary point. We then expand the recursion and bound all deviating terms with Lemmas \ref{lem1}, \ref{lem2}, \ref{lem3}, and \ref{lem4} and derive Theorem \ref{thm2}. For the strongly convex case, we construct a matrix that iterates the terms $\mathbb{E}\|\overline{x}_\mathrm{g} ^t - x^\star\|^2 $, $\Gamma_t $, $\mathcal{Y}_t$, and $\gamma \mathcal{Z}_t$. Then, we bound the $\ell_1$-norm of this matrix, and choose an appropriate step size so that we can ensure a linear rate when $\sigma = 0$. This leads to the results in Theorem \ref{thm3}.

In all three theorems, we can also see the importance of initialization for gradient tracking terms. If we have arbitrary initialization of the gradient tracking terms under the constraint that $\mathcal{Y}_1\frac{11^\top}{n} = \mathcal{Z}_1\frac{11^\top}{n} = \mathbf{0}$, then we will inevitably have the norm values $\mathcal{Y}_1$ and $\mathcal{Z}_1$ in the convergence bounds. The values of these two terms depend on the data heterogeneity level, which we aim to avoid in our results. On the other hand, when we use the initialization defined in the {\tt SD-GT} algorithm, we can show that $\mathcal{Y}_1 \leq \sigma^2$ and $\mathcal{Z}_1 \leq \sigma^2$, which removes the dependency on gradient dissimilarity across devices.
%This shows that for gradient tracking methods to remove dependency on data-heterogeneity, proper initialization on the gradient tracking terms is required.

\vspace{-0.1in}

\subsection{Learning-Efficiency Co-Optimization}
\label{sec:IV-E}
% \evan{change all the control algorithm into co-optimization alg}
Based on Corollary \ref{cor1}, we can see that the convergence speed under general non-convex problems is highly affected by the ratio $\frac{1}{p^2}$. Without any additional conditions, the best choice to maximize the convergence speed is setting $p = 1$.
% where $p_1 = 1$ means the central server samples all clients from each subnet and $p_2 = 1$ means each subnet has a fully connected (centralized) network structure.
However, recall that one of the main objectives of semi-decentralized FL is to reduce communication costs between devices and the main server; for this, we wish to minimize $p$. We therefore propose a method for the central server to trade-off between convergence speed and communication cost, e.g., the energy or monetary cost for wireless bandwidth usage.
% \evan{change the energy cost to communication cost e.g. energy/monetary cost for bandwidth usage}

To facilitate this, assume that for each subnet $\mathcal{C}_s$, the server estimates a communication cost of $E_{s}$ for pulling and pushing variables from the subnet $\mathcal{C}_s$, and a communication cost of $E_{s}^{\mathrm{D2D}}$ for every round of D2D communication performed by the subnet $\mathcal{C}_s$. As in the case of Theorem \ref{thm1}, we define $\beta_{s} = \frac{m_{s} - h_{s}}{m_{s}}$ to be the ratio of unsampled clients from each subnet. Based on Corollary \ref{cor1}, then, one possibility is for the server to consider the following optimization objective for determining the number of D2D communication rounds $K$ and $\beta_{s}$:
\begin{equation}
    \begin{aligned}
     &\textstyle \frac{\lambda_1}{p^2} +  \sqrt{\frac{\lambda_2}{K}} + (\frac{\lambda_2}{Kp^2})^{\frac{2}{3}} \\
    &\textstyle+ \lambda_3\Big(\sum_{s=1}^S(1 - \beta_{s})\cdot E_{s}\textstyle+ K\sum_{s=1}^S E^{D2D}_{s}\Big),
    \end{aligned}
    \label{eq16}
\end{equation}
where $\lambda_1,\lambda_2,\lambda_3 > 0$. $\lambda_1, \lambda_2$ replace unknown constants in Corollary \ref{cor1}, and $\lambda_3$ balances the importance of the communication costs. However, using~\eqref{eq16} directly would not account for the fact that the gradient norm $\mathbb{E}\|\nabla f(\overline{x}_\mathrm{g} ^t)\|^2$ will decrease over the course of training, giving an opportunity to favor a more communication efficient solution (i.e., higher $K$ and $\beta_s$) in later training stages. To this end, let $K_t$ be the number of D2D rounds each global iteration $t$, $\beta_{1}^{t}, \ldots, 
\beta_{S}^{t}$ be the rate of unsampled clients from each subnet at $t$, and define $p_t = \min(1 - (\beta_{1}^{t})^2, \ldots, 1 - (\beta_{s}^{t})^2)$. We aim to adaptively control these variables at the server each global round $t$.

% By using the terms from Theorem \ref{thm1}, and simply replacing all coefficients that are not computable by the server with some balancing constants. The server can try to the form a minimization problem based on the equation below with the balance terms $\lambda_1,\lambda_2,\lambda_3 > 0$.
% {\begin{equation}
%     \begin{aligned}
%      &\textstyle \frac{\lambda_1}{p^2} +  \sqrt{\frac{\lambda_2}{K}} + (\frac{\lambda_2}{Kp^2})^{\frac{2}{3}} \\
%     &\textstyle+ \lambda_3\Big(\sum_{s=1}^S(1 - \beta_{s})\cdot E_{s}\textstyle+ K\sum_{s=1}^S E^{D2D}_{s}\Big),
%     \end{aligned}
%     \label{eq16}
% \end{equation}}

Corollary \ref{cor1} captures the behavior of our algorithm under a specific initialization. In order to update the values of $\beta_{i}^{t}$, $p_t$ and $K_t$ dynamically throughout the training process, we will consider the Lyapunov function $\mathcal{H}_t$ as the bound used for estimation instead of only the function value gap $\mathcal{E}_t$. Note that $\mathcal{H}_t$ is a linear combination of $\mathcal{E}_t$ and three other training performance terms from Sec.~\ref{sec:IV-C}: $\mathcal{Y}_t$, $\mathcal{Z}_t$, and $\Gamma_t$. It is unclear how to obtain a reliable estimate for the in-subnet gradient tracking performance $\mathcal{Z}_t$ at the server, since in {\tt SD-GT} there is no edge server that can aggregate the information of all $z_i^t$ within subnets. As a result, we focus on estimating the remaining three terms. The server can obtain an approximation of the between-subnet gradient tracking performance $\mathcal{Y}_t$ using the change in $\psi_{s}^{t}$ values~\eqref{eq10} it stores for each subnet from $t$ to $t + 1$:
\begin{equation}
        \textstyle\hat{Y}_t = \frac{1}{S} \sum_{s=1}^{S}\|\psi_{s}^{t} - \psi_{s}^{t+1}\|^2.
        \label{yhat}
\end{equation}
For the client deviation term $\Gamma_t$, since the server does not have access to all clients, it can make an estimate based on its sampled set:
\begin{equation}
    \textstyle\hat{\Gamma}_t = \frac{1}{S\cdot h_s}\sum_{s=1}^S \sum_{j=1}^{h_s} \left\|x_{s,j}^{t,K+1} - x_g^{t+1}\right\|^2.
    \label{gammahat}
\end{equation}
Finally, recall that $\mathcal{E}_t$ measures how close the global model is to the optimal solution. However, the server does not know $x^\star$, and also cannot compute the expectation $\mathbb{E}f(\cdot)$ over the system's data distribution. Hence, we resort to Corollary~\ref{cor1}'s indication of a sub-linear convergence rate, and use a decaying term $\frac{1}{t}$ to approximate the effect of $\mathcal{E}_t$. Putting this together, we have the following approximate Lyapunov function value:
\begin{equation}
    \hat{\mathcal{H}}_t = \frac{1}{t} + \lambda_1^2 \Big( \frac{(K_t)^3\gamma^3}{(p_t)^2}\hat{Y}_t + \frac{K_t\gamma}{p_t}\hat{\Gamma}_t\Big).
    \label{Hhat}
\end{equation}
Using this estimate together with~\eqref{eq16}, at the end of global iteration $t$, the server can adapt the next iteration's sampling rates and D2D rounds by solving the following optimization:
\begin{align}
    % \begin{aligned}
    \min_{\beta_{1}, \ldots, \beta_{S}, K}\quad &\textstyle \frac{\lambda_1\hat{H}_t}{p^2} +  \sqrt{\frac{\lambda_2\hat{H}_t}{K}} + (\frac{\lambda_2\hat{H}_t}{Kp^2})^{\frac{2}{3}} \notag\\
    &\textstyle+ \lambda_3\Big(\sum_{s=1}^S(1 - \beta_{s})\cdot E_{s}\textstyle+ K\sum_{s=1}^S E^{D2D}_{s}\Big),\notag\\
    \textrm{subject to} \quad &\textstyle 0 \leq \beta_{s} \leq \frac{m_{s} - 1}{m_{s}},\quad 1\leq i \leq l,\notag\\
    &\textstyle p = \min(1 - \beta_{1}^2, \ldots, 1 - \beta_{s}^2).    \label{control_alg_prob}
    % \end{aligned}
\end{align}

\begin{table*}
\normalsize
    \centering
    \caption{Comparison of the number of communication rounds required for different methods to reach an $\epsilon$-optimal solution. For the gradient correction methods which effectively contain one subnet, $\rho$ denotes the connectivity of the network (based on the spectral radius), with $\rho = 1$ being a fully connected network. H-SGD considers a two layer tree-like network with each subnet containing an edge server that can aggregate model parameters from all clients in the subnet. Here, $\eta$ and $\eta_s$ are data-heterogeneity related variables: $\eta$ is the upper bound of the gradient divergence between subnets, and $\eta_s$ is the upper bound of the gradient divergence within subnet $s$.}
    \begin{tabular}{cccc}
        \hline
       Method  & non-convex rate  &  network topology & uses gradient correction \\
        \hline
        \hline
       \vspace{0.05in}
       
       SCAFFOLD \cite{mishchenko2022proxskip} & $\mathcal{O}\Big( \frac{\sigma^2}{nK\epsilon^2} + \frac{1}{\epsilon}\Big)$ & centralized & \cmark\\
        \hline
       \vspace{0.05in}
       K-GT \cite{liu2023decentralized}  & $\mathcal{O}\Big( \frac{\sigma^2}{nK\epsilon^2} + \big(\frac{\sigma}{\rho^2 \sqrt{K}}\big)\frac{1}{\epsilon^{\frac{3}{2}}} + \frac{1}{\rho^2\epsilon}\Big)$ & decentralized & \cmark\\
        \hline
       \vspace{0.05in}
       LED \cite{alghunaim2024local}& $\mathcal{O}\Big( \frac{\sigma^2}{nK\epsilon^2} + \big(\frac{\sigma}{ \sqrt{\rho K}}\big)\frac{1}{\epsilon^{\frac{3}{2}}} + \frac{1}{\rho\epsilon}\Big)$& decentralized& \cmark\\
        \hline
       \vspace{0.05in}
       H-SGD \cite{wang2022demystifying}&  $\mathcal{O}\Big(\frac{\sigma^2}{n\epsilon^2} + (\sqrt{K}\sigma + K\eta + \sum_{s=1}^{\mathcal{S}}\eta_s)\frac{1}{\epsilon^{\frac{3}{2}}} + \frac{K}{\epsilon}\Big)$ & hierarchical& \xmark\\
        \hline
       \vspace{0.05in}
        
       SD-GT (ours) &  $\mathcal{O}\Big( \frac{\sigma^2}{nK\epsilon^2} + \big(\frac{\sigma}{p^2 q \sqrt{K}}\big)\frac{1}{\epsilon^{\frac{3}{2}}} + \frac{1}{p^2q\epsilon}\Big)$ & semi-decentralized& \cmark\\
        \hline
       
    \end{tabular}
    \vspace{0.15in}
    \label{table:1}
\end{table*}

\begin{algorithm}[t]
{
\caption{Dynamic Control Algorithm for {\tt SD-GT} at iteration $t$.}
\label{alg:2}
\KwInput{$t, p_t, K_t, \psi_{s}^t, \psi_x^{t+1}, x_{s,j}^{t,K+1}, x_g^{t+1}$}
\KwOutput{$\beta_{1}^{t+1}, \ldots, 
\beta_{S}^{t+1}, p_{t+1},  K_{t+1}$}
\KwConstants{$\lambda_1, \lambda_2, \lambda_3$}
Compute $\hat{Y}_t$ and $\hat{\Gamma_t}$ using \eqref{yhat} and \eqref{gammahat}.\\
Compute $\hat{\mathcal{H}_t}$ using \eqref{Hhat}.\\
Solve the minimization problem \eqref{control_alg_prob}, e.g., using GP, to obtain $ K_{t+1}$ and $\beta_{1}^{t+1}, \ldots, 
\beta_{S}^{t+1}$.\\
Compute $p_{t+1} = \min(1 - (\beta_{1}^{t+1})^2, \ldots, 1 - (\beta_{s}^{t+1})^2)$.\\
Server broadcasts the value $K_{t+1}$ to all clients.
}
\end{algorithm}
%Base on the minimization problem \eqref{control_alg_prob}, we propose 

% In the minimization problem \eqref{control_alg_prob}, the ratio between $E^{D2D}_{s}$ and $E_{s}$ is set to  $E_{s}^{D2D}/E_{s} = \delta$
% % with $\delta \ll 1$ \shiqiang{Why is this needed? It seems $E^{D2D}_{\mathcal{C}_s}$ and $E_{\mathcal{C}_s}$ are input parameters to the problem, which should be possible to take any positive value}
% . A small $\delta$ means that the cost for a client to communicate with its neighbors using D2D communication is much cheaper than to perform a centralized aggregation.

The optimization problem \eqref{control_alg_prob} 
% \shiqiang{Add word ``Equation'' only when equation number appears at the beginning of a sentence (IEEE style)} 
has a similar form to a geometric program (GP), with the exception of the last constraint $p = \min(1 - \beta_{1}^2, \ldots, 1 - \beta_{S}^2)$. By relaxing it to $p \leq \min(1 - \beta_{1}^2, \ldots, 1 - \beta_{S}^2)$, it becomes a posynomial inequality constraint which can be handled by GP~\cite{boyd2007tutorial}.
The solution of the relaxed problem will remain the same as \eqref{control_alg_prob}, which is provable via a straightforward contradiction. It is worth comparing this procedure to the more complex control algorithm for SD-FL proposed in \cite{lin2021semi}, where one must adaptively choose a smaller $K$ when client models in a subnet deviate from each other too quickly. With our {\tt SD-GT} methodology, the subnet drift is inherently controlled, allowing a relatively stable $K_t$ throughout the whole training process.

\textcolor{black}{The computational complexity of solving this GP using interior-point methods (IPMs) is approximately $\mathcal{O}(S^3 \log(S/\epsilon))$, where $\epsilon$ is the desired accuracy~\cite{boyd2004convex}; in particular, since our formulation optimizes over $S + 1$ variables, each Newton iteration incurs $\mathcal{O}(S^3)$ cost. Importantly, the overhead of solving this optimization is small compared to the other server operations of model aggregation, broadcasting, and gradient tracking that scale with the model dimension $d$.}
%Thus, the control step incurs only a tiny fraction of the total server computational cost.}

% \chris{Add some discussion here on how our approach based on gradient tracking leads to a simpler control optimization than~\cite{lin2021semi} because we explicitly integrate subnet-drift into the learning algorithm.}

Algorithm \ref{alg:2} summarizes the dynamic control procedure for {\tt SD-GT} developed in this section. It can be executed at any global iteration $t \in [1, T]$ during the training process to update $K_t$ and the sampling rates. \textcolor{black}{
We provide both a convergence analysis of {\tt SD-GT} under Algorithm~\ref{alg:2} and bounds on the associated $K_t$ values in Appendix~\ref{appen:ctrl_convergence}. Overall, those results show that Algorithm~\ref{alg:2} achieves the same convergence rate as Corollary~\ref{cor1}, with time-varying effects from adaptive parameters bounded by their minimum observed values during training. Moreover, higher D2D costs are shown to restrict feasible choices of $K_t$, while lower costs allow larger values, demonstrating that our dynamic control of $K_t$ balances communication and optimization objectives as intended.}

% After solving the minimization problem, the server has the ability to dynamically select more clients from clusters that cost less to perform global aggregation and select less clients from clusters that cost more to perform global aggregation. If the server wants the same sampling rate $\frac{h_{c_i}}{m_{c_i}}$ throughout all clusters, it can simply set $E_{c_1} = E_{c_2} = \ldots = E_{c_l}$.

\vspace{-0.1in}
\section{Performance Evaluation}
\label{sec:V}
In this section, \textcolor{black}{we first compare the non-convex convergence rate with  existing works (Sec. \ref{sec:IV-D})}. Then, we introduce our experimental setup (Sec. \ref{sec:V-A}), evaluate the performance compared to baselines under different aggregation frequencies (Sec. \ref{sec:V-B}), and compare results varying different network parameters (Sec. \ref{sec:V-C}). We then validate the linear rate shown in our theory for strongly convex problems with deterministic gradients (Sec. \ref{sec:V-D}). Finally, we show the performance of our co-optimization control algorithm (Sec. \ref{sec:V-E}). \textcolor{black}{Additional experiments considering weakly-convex problems and stochastic noise are given in Appendix~\ref{appen:additional_experiments}.}
\vspace{-0.1in}
\subsection{Rate Comparisons with Related Works}
\label{sec:IV-D}

To compare with existing methods, we consider the communication rounds required to reach an $\epsilon$-optimal solution. For methods considering FL's star topology and its hierarchical extension, this corresponds to DS communication rounds, while for methods considering fully decentralized FL, this corresponds to D2D rounds. Since {\tt SD-GT} operates in the semi-decentralized setting containing both DS and D2D, we focus on the more energy-intensive DS rounds for our setting. Table \ref{table:1} summarizes the rates of our work and existing works:
% This is a reasonable comparison the D2D communications are considered to be low cost communication that are likely to be as energy-efficient as computing local gradients.

\textbf{Comparison with gradient correction methods.} First, we compare our results with the gradient tracking-based methods SCAFFOLD and K-GT \cite{mishchenko2022proxskip, liu2023decentralized}, and the exact-diffusion method LED \cite{alghunaim2024local}. The terms $p$ and $q$ in our convergence rate reflect the impact of the hierarchical structure considered by {\tt SD-GT}, different from the other gradient correction methods that do not consider this structure. Here, $p$ is determined by the lowest sampling rate of each subnet. If the server samples all clients from each subnet, we will have $p = 1$ and maximize the convergence rate, at the expense of more DS communication. The quantity $q$ is the lowest connectivity of all subnets: if one subnet $\mathcal{C}_s$ is sparse with a very low connectivity $\rho_s$, this will still affect the required communication rounds to reach $\epsilon$-optimal solution.

For a direct rate comparison, we can consider the special case of {\tt SD-GT} where the hierarchical network only has one subnet and the server aggregates from all clients in the subnet. In this case, we have $p = 1$ and $q = \rho$. In this special case, the rate of our algorithm reduces to:
\begin{equation}
    \textstyle\mathcal{O}\left( \frac{\sigma^2}{nK\epsilon^2} + \left(\frac{\sigma}{\rho \sqrt{K}}\right)\frac{1}{\epsilon^{\frac{3}{2}}} + \frac{1}{\rho\epsilon}\right).
\end{equation}
Compared with K-GT and LED, we can see that our first term $\mathcal{O}\left(\frac{\sigma^2}{nK\epsilon^2}\right)$ recovers their result. On the other hand, our second term $\mathcal{O}\left(\big(\frac{\sigma}{\rho \sqrt{K}}\big)\frac{1}{\epsilon^{\frac{3}{2}}}\right)$ is in between the rate of K-GT and LED, while our third term $\mathcal{O}(\frac{1}{\rho\epsilon})$ is better than K-GT and the same as LED. This result is consistent to observations from previous works where exact-diffusion methods have better network-dependent rates than gradient tracking methods \cite{alghunaim2022unified,alghunaim2024local}.

\textbf{Comparison with hierarchical FL methods.} Next, we compare our rate with H-SGD, a recent hierarchical FL framework that does not employ gradient tracking \cite{wang2022demystifying}. There are several differences between their setting and our setting: 1) Unlike H-SGD, we do not assume the existence of edge servers that can aggregate information within a subnet. 2) We perform D2D communication every time each client performs one round of gradient update, while H-SGD performs in-subnet aggregations after $I_s$ rounds of local gradient updates. Hence, for a direct comparison, we set $I_s = 1$ to observe the convergence rates in Table \ref{table:1}. There are several differences between the two bounds. First, our bound does not require assumptions on data heterogeneity, while their bound requires some $\eta$ and $\eta_s$ that bounds the gradient divergence of the network. Second, we see that our rate improves with increasing the period $K$ between global aggregations,\footnote{The caveat is that the step size must decrease proportional to $K$, as shown in Theorem~\ref{thm1}, consistent with the other gradient tracking methods.} while H-SGD's gets slower. This is due to {\tt SD-GT} employing D2D communications for local aggregations, as well as our removal of the client drift impact through gradient tracking.

\begin{figure*}
% \centerline{\includegraphics[width=\textwidth]{MNIST/setting3_MNIST-09-07-22-05-22.jpg}}
\centerline{\includegraphics[width=\textwidth]{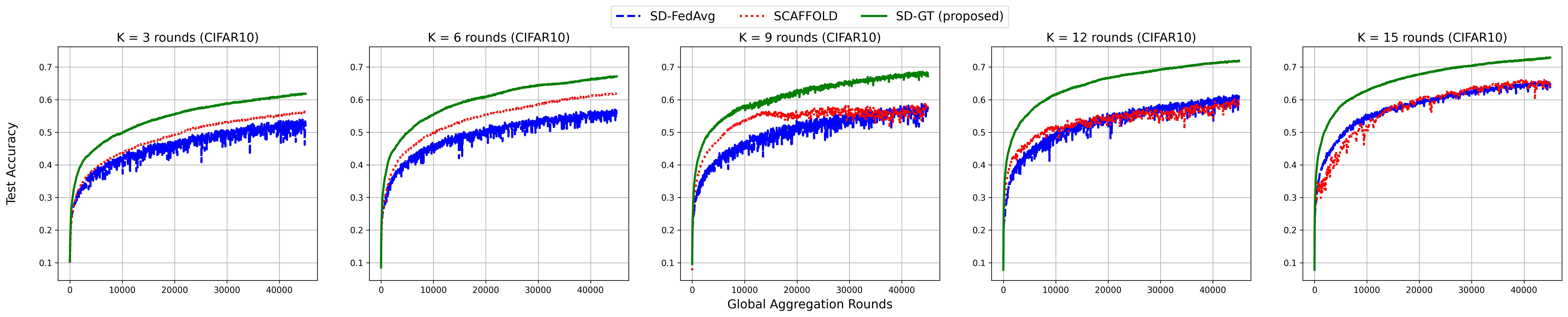}}
\caption{Comparison between algorithms on CIFAR10 datasets when changing the number of local client updates and D2D consensus rounds $K$ between global aggregations. Each experiment is conducted with 30 clients and 3 subnets. As $K$ increases, {\tt SD-GT} is able to take advantage of multiple in-subnet model update and consensus iterations while correcting for client drift to achieve better convergence speed, particularly on CIFAR10.\vspace{-0.15in}}
\label{fig:CIFAR10_exp1}
\end{figure*}

\begin{figure}
\center
{
\includegraphics[width=0.88\linewidth]{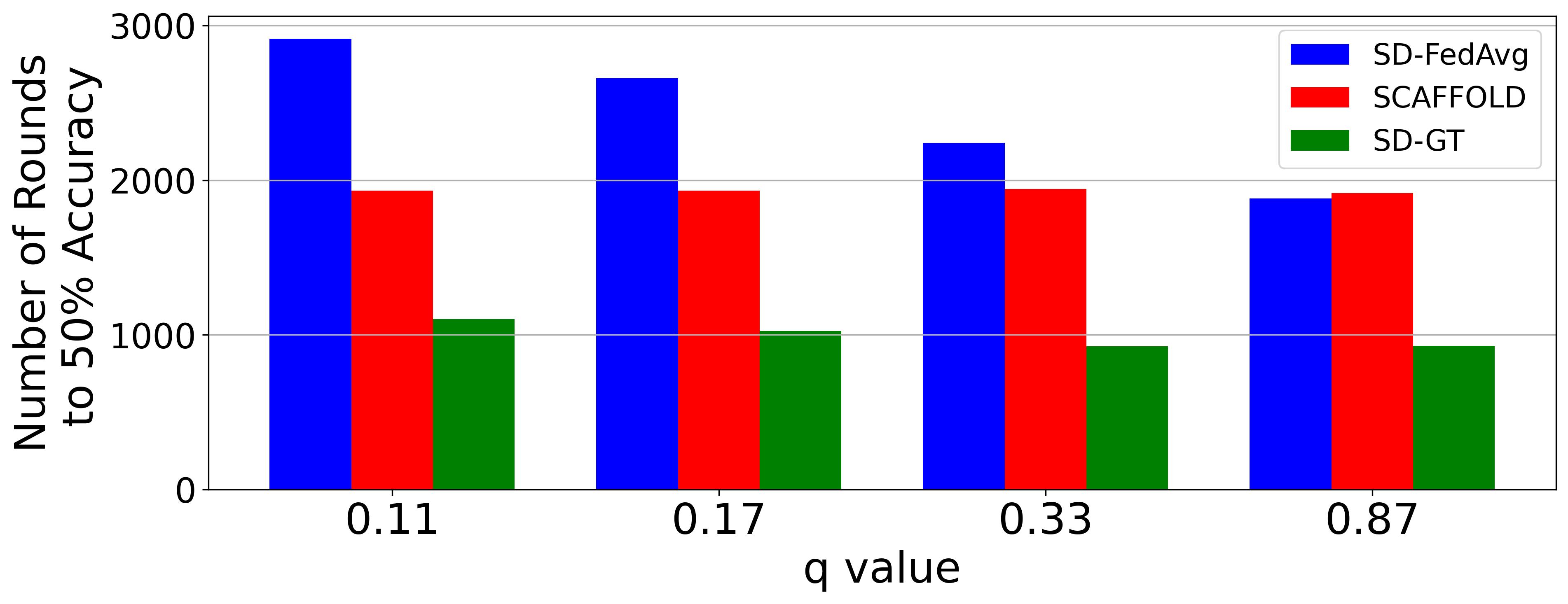}
}
\caption{\textcolor{black}{Comparison when changing the radius range ($q = 0.1, 0.17, 0.3, 0.87$) of the geometric graph for D2D communication links, on the CIFAR10 dataset. All experiments have $K = 3$, with 30 total clients and 3 subnets. Compared with SD-FedAvg, our method that combines both D2D communications and gradient tracking obtains robust performance over different subnet connectivity levels. SCAFFOLD's performance is unaffected since it does not employ D2D communications.} \vspace{-0.15in}}
\label{fig:setting4}
\end{figure}

\vspace{-0.1in}
\subsection{Experimental Setup}
\label{sec:V-A}
% \chris{Rather than setting 1 vs. setting 2, I think we should call it 4 datasets/ML tasks. The setup can have 1) System setup; 2) Datasets/ML tasks; 3) Baselines; 4) Parameter settings}
\textbf{System Setup.} By default, we consider $S = 3$ subnets and $m_s = 10$ clients per subnet. For each subnet $s$, we generate the D2D communication graph structure among the set $\mathcal{C}_s$ of clients using a random
geometric graph, as in prior works~\cite{lin2021semi,parasnis2024energy}. Specifically, \textcolor{black}{we generate a realization of $n$ device coordinates sampled over a uniform grid. The clients are then grouped into subnets base on their distance using $K$-means clustering}. Then, a sensing radius is generated for each client, and when two clients are within each other's range, a D2D link is established between them. By default, we choose the radius for each client randomly in the range $[0.5, 3.5]$. Subsequently, the weight matrices $W_1,...,W_S$ for each subnet are generated using the standard Metropolis-Hasting rule \cite{sayed2014adaptation}, which enforces double stochasticity on each $W_s$. In Sec.~\ref{sec:V-C}, we will consider several network variations, where we change the connectivity of subnets, the size of each subnet, and other aspects to demonstrate the importance of employing gradient tracking terms between-subnets and within-subnets.

% 1) Fixing the total number of clients in the system, and changing the number of subnets, where each subnet holds the same number of clients; 2) Fixing the number of clients and the sampling rate of the server sampling from each subnet, and changing the number of subnets.
%; 3) Fixing the number of clients and their network topology, and changing the number of D2D communication rounds $K$ performed between global aggregations.

\textbf{Datasets.} There are in total four datasets that we use in our experiments. The first three are standard image classification tasks used in FL, and last one one is a synthetic dataset.

\textit{(i) Real-world Datasets:} We consider three neural network classification tasks using a cross entropy loss function on image datasets (MNIST\cite{lecun1998gradient}, CIFAR10, CIFAR100\cite{krizhevsky2009learning}). Let $\mathcal{D}_i$ be the dataset allocated on client $i$; we set $|\mathcal{D}_i|$ to be the same $\forall i$. 
To simulate data heterogeneity, by default, each client in the MNIST and CIFAR10 experiments only hold data of one class (out of 10), and each client in the CIFAR100 experiments only hold data from at most six classes (out of 100). Subnets of given sizes are then formed among clients with proximity distance. \textcolor{black}{We then experiment with different non-i.i.d. allocations and subnetting strategies in Sec.~\ref{sec:V-C}.}
% \evan{cite the papers of the dataset}

\textit{(ii) Synthetic Dataset:} For Sec.~\ref{sec:V-D}, we consider a strongly convex Least Squares (LS) problem using synthetic data, to demonstrate the linear convergence of our algorithm under strong-convexity. In the LS problem, each client $i$ estimates an unknown signal $x_0 \in\mathbb{R}^d$ through linear measurements $b_i = A^ix_0 + n_i$, where $A^i \in \mathbb{R}^{{|\mathcal{D}_i|}\times d}$ is the sensing matrix, and $n_i \in \mathbb{R}^{|\mathcal{D}_i|}$ is the additive noise. 
% The LS problem cost function for each local client will be in the form:
% \begin{equation}
%     f_i(x) = \|A_ix - b_i\|^2.
% \end{equation}
Here, $d = 200$ and $x_0$ is a vector of i.i.d. random variables drawn from the standard normal distribution $\mathcal{N}(0,1)$. Each client $i$ receives $|\mathcal{D}_i| = 30$ observations. All additive noise is sampled from $\mathcal{N}(0, 0.04)$. The sensing matrix $A^i$ is generated as follows: For the $j$\textsuperscript{th} row of $A^i$, denoted by $a^i_j \in \mathbb{R}^d$, let $z_1, ..., z_d$ be an i.i.d. sequence of $\mathcal{N}(0,1)$ variables, and fix a correlation parameter $\omega \in [0, 1)$. For each row $j$ in $A^i$, we initialize by setting the first element to be $a^i_{j,1} = \frac{z_1}{\sqrt{1 - \omega^2}}$, and generate the remaining entries by applying the recursive relation
$a^i_{j,l+1} = \omega a^i_{j,l} + z_{l+1}$, for $l = 1, 2, \ldots, d-1$.

\textbf{Baselines.} We consider two baselines for comparison:

\textit{(i) SD-FedAvg:} Our first baseline is a semi-decentralized version of FedAvg\cite{lin2021semi} (denoted as SD-FedAvg). Each client updates using only its local gradient, and communicates with its nearby neighbors within each subnet using D2D communication after every gradient computation. A global aggregation is conducted after every $K$ rounds of gradient computation. This baseline does not contain gradient tracking, thus allowing us to assess this component of our methodology.

\textit{(ii) SCAFFOLD:} We also run a comparison with SCAFFOLD\cite{karimireddy2020scaffold} from Table~\ref{table:1}, a centralized gradient tracking algorithm that considers star topology aggregations. This baseline does not conduct D2D-based model aggregations. Thus, within each global round, SD-FedAvg and {\tt SD-GT} perform $K$ rounds of D2D communications and model updates, while SCAFFOLD computes $K$ rounds of local on-device updates. Comparing with SCAFFOLD allows us to assess the benefit of our methodology co-designing semi-decentralized updates with gradient tracking.\footnote{Recall that K-GT and LED, the other gradient tracking methods in Table~\ref{table:1}, consider the fully decentralized setting without any server, hence isn't comparable with our algorithm which considers a central server.}

\textbf{Parameter Settings.} We use a constant step size of $\gamma = 1 \times10^{-2}$ for MNIST, $\gamma = 1 \times 10^{-3}$ for CIFAR10, and $\gamma = 3 \times 10^{-4}$ for CIFAR100. For MNIST, we consider training a two layer fully-connected neural network with ReLU activations. For CIFAR10 and CIFAR100, we use a two layer convolutional neural network with ReLU activations. For experiments using the synthetic dataset, we tuned the value of $\omega$ to create two different condition numbers $\kappa$: (i) $\kappa \approx 80$, which is a simpler learning task, and (ii) $\kappa \approx 800$, which is more difficult. The step size for synthetic dataset experiments is set to $\gamma = 10^{-4}$. We use deterministic gradients for synthetic dataset and batch-size of $512$ for the image datasets.

For experiments on our co-optimization algorithm, the communication cost $E_{s}$ for each subnet is chosen uniformly at random between $1$ and $100$. The balance terms are set to $\lambda_1 = 1$,  $\lambda_2 = 10^{-1}$ and $\lambda_3 = 10^{-5}$. We compare different cost ratios $\delta = E_{s}^{\mathrm{D2D}}/E_{s}$. When $\delta << 1$, it indicates that D2D communication is much cheaper than global aggregation.
%We set $\delta = 1$ and $10^{-2}$ side by side to observe the importance of D2D communication. 

\begin{figure*}
\centerline{\includegraphics[width=\textwidth]{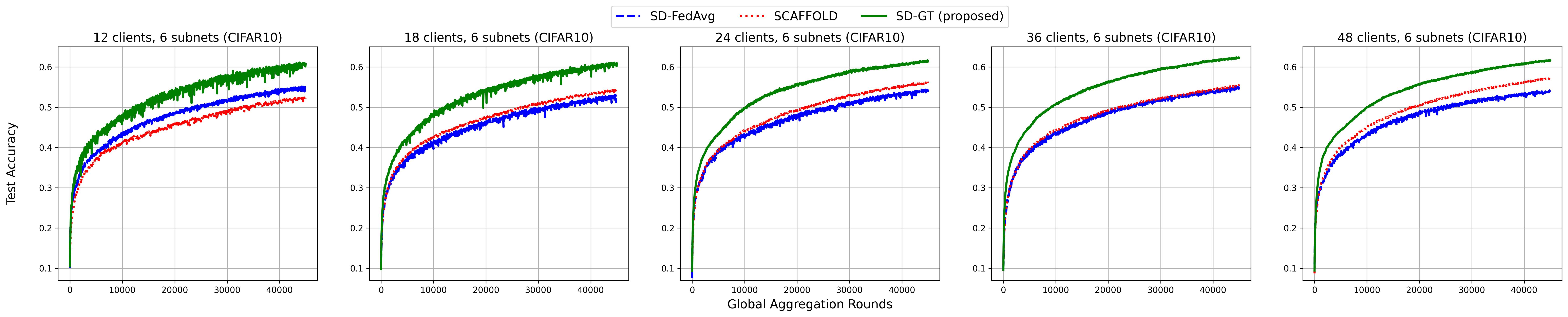}}
\centerline{\includegraphics[width=\textwidth]{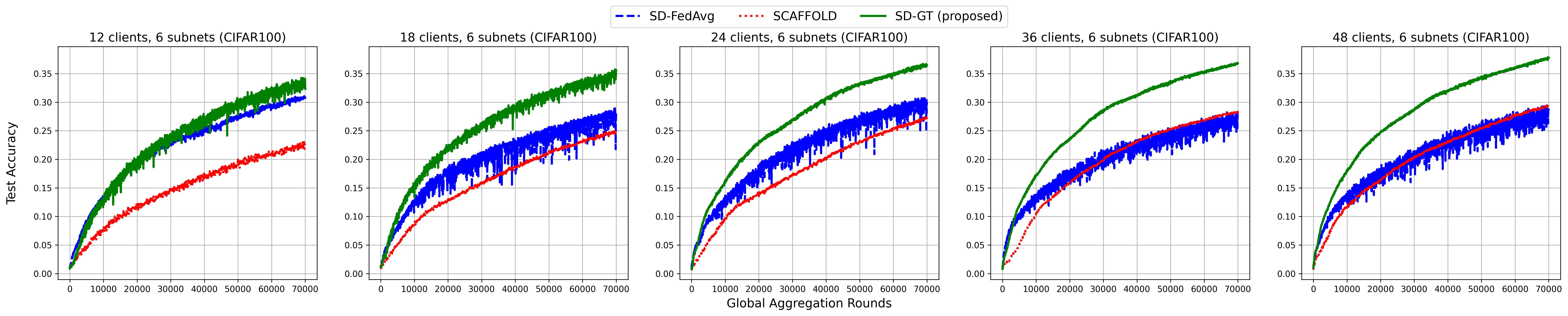}}
\caption{Impact of varying the total number of clients for $S = 6$ subnets ($K = 3$). As the size of each subnet increases, the within-subnet data heterogeneity increases. We see that {\tt SD-GT} is able to obtain larger improvements over the baselines as the number of clients grows larger (particularly for CIFAR100, an intrinsically more heterogeneous dataset due to having more labels), due to its inclusion of within-subnet gradient tracking terms.\vspace{-0.15in}}
%When the network connectivity and data heterogeneity within each subnet increases, {\tt SD-GT} is able to mix up information in each subnet better than SD-FedAvg.\vspace{-0.15in}}
\label{fig:CIFAR10_exp3}
\end{figure*}

\vspace{-0.1in}
\subsection{Varying Device-Server Communication Frequencies}

\label{sec:V-B}
% \textbf{Experiment 1: Convergence results for varying D2D rounds.} 
In Figure \ref{fig:CIFAR10_exp1}, we compare the model training performance of the algorithms on CIFAR10 as $K$, i.e., the number of D2D rounds and local updates conducted between global aggregations, is varied. In these experiments, the central server samples $40\%$ of the clients from each subnet. We experiment with $K = 3$ to $K = 15$ to observe the effect of performing multiple rounds of in-subnet consensus operations, where the frequency of global aggregations decreases as $K$ gets large. We can see that for CIFAR10, the SD-FedAvg baseline struggles to gain improvement from increasing the number of D2D communications. On the other hand, {\tt SD-GT} obtains better results when using a larger number of D2D rounds, since it tracks the gradient information through the network and corrects the client drifts accordingly. We also see that the convergence of SCAFFOLD is slower than {\tt SD-GT}, even though both algorithms uses gradient tracking to correct the update direction. This is due to the fact that SCAFFOLD does not consider D2D communications in-between global aggregations; {\tt SD-GT} exploits the fact that D2D is generally much less expensive than DS to conduct in-subnet aggregations.
%For the case where $K = 15$, we can observe that for the CIFAR10 dataset, the stability of convergence will drop, which indicates that the chosen step size is at the edge of violating the theoretical upper bound.
Overall, we see that \textit{our algorithm handles data heterogeneity better than the baselines, gaining more improvement from increasing the number of in-subnet D2D communication rounds.}
\vspace{-0.1in}
\subsection{Varying Network Structures}

\label{sec:V-C}
\textcolor{black}{\textbf{Impact of D2D Subnet Connectivity.} 
In Figure~\ref{fig:setting4}, we vary the D2D communication radius within each subnet to evaluate the effect of D2D connectivity on training convergence speed, for the CIFAR10 dataset. The device radius range across devices is varied from $[6.5, 7.5]$ (giving a mixing rate $q = 0.87$), corresponding to well-connected subnets, down to $[0.5, 1.5]$ ($q = 0.11$), representing sparse connectivity. As the subnet connectivity decreases, both SD-FedAvg and {\tt SD-GT} experience an increase in rounds required to reach the target accuracy. This observation aligns with Theorem~\ref{thm1}, where smaller values of $q$ results in slower convergence. However, the impact on SD-FedAvg is much more pronnounced. The performance drop can be attributed to reduced effectiveness of the D2D consensus process in sparsely connected subnets, which increases client drift. In contrast, our method corrects this drift using the global gradient tracking terms $y_i^t$.}

\textbf{Impact of Subnet Sizes.} 
In Figure \ref{fig:CIFAR10_exp3}, we consider the impact of varying the total number of clients $n$ for a fixed number of subnets $S = 6$, on the CIFAR10 and CIFAR100 datasets. For this experiment, we hold the data samples contained in each subnet constant as the subnet size $m_s$ increases, resulting in each client receiving a smaller portion of the data. This implies that the data heterogeneity within each subnet varies, while the data heterogeneity between subnets remains constant. Hence, this experiment assesses the impact of the within-subnet gradient tracking terms $z_i^t$ and D2D communication in {\tt SD-GT}.
%the improvement of our algorithm when the number of clients increases relies on the within-subnet gradient tracking terms and D2D communication.

Note that for a fixed SD-FL configuration, CIFAR10 has lower data heterogeneity than CIFAR100, since it has fewer labels. As a result, for CIFAR10, no matter how the subnet size varies, the D2D communications of SD-FedAvg do not bring much extra improvement compared to SCAFFOLD's global gradient tracking. On the other hand, for CIFAR100, SD-FedAvg is able to outperform SCAFFOLD when subnet sizes are small, given the benefit of local model synchronization. Still, SD-FedAvg's performance drops as the number of clients within each subnet increases. In all cases, {\tt SD-GT} is able to maintain superior performance over the baselines. This shows that \textit{within-subnet gradient tracking stablizes the convergence behavior when data heterogeneity within subnets increases}. 

% From Figure \ref{fig:CIFAR10_exp2}, we can see that for scenarios where global aggregation is more important, our method is able to outperform both methods with the usage of the between-subnet gradient tracking. For scenarios where D2D communication is more important, the within-subnet gradient tracking term is able to help our method perform the best among the baselines. \textit{This shows that the two gradient tracking terms we introduced is essential for stable convergence under different network structures.}

\begin{figure*}
\centerline{\includegraphics[width=\textwidth]{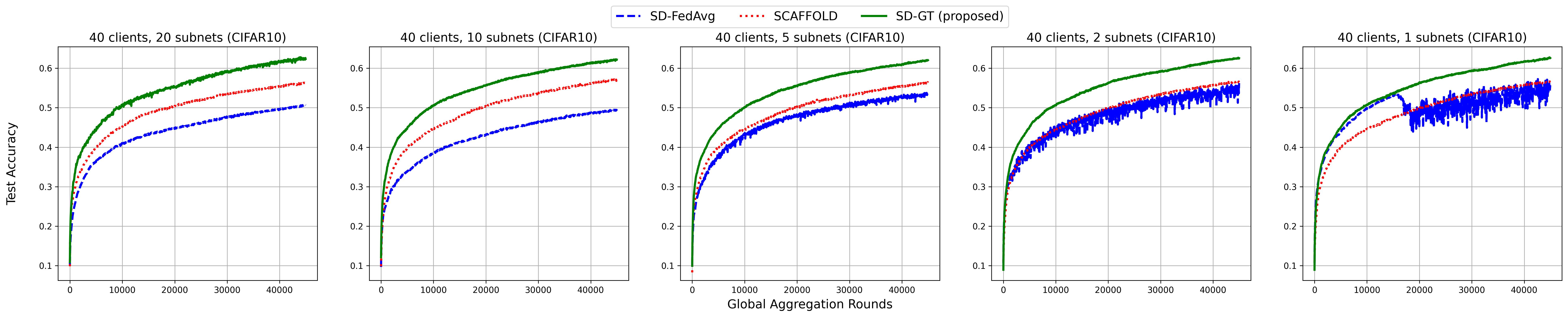}}
% \centerline{\includegraphics[width=\textwidth]{CIFAR100/setting1_CIFAR100-09-07-22-16-51.jpg}}
\caption{Impact of changing the number of subnets $S$ with fixed total number of clients $n = 40$ ($K = 3$). 
We can see that when varying subnet compositions, SD-FedAvg struggles to mix the information between subnets as within-subnet data heterogeneity grows, causing instability. {\tt SD-GT} is able to outperform both baselines using the combination of between-subnet and within-subnet gradient tracking on top of D2D communications.\vspace{-0.15in}}
\label{fig:CIFAR10_exp2}
\end{figure*}

\begin{figure}
    \centering
    \begin{subfigure}[t]{0.9\linewidth}
        \centering
        \includegraphics[width=\linewidth]{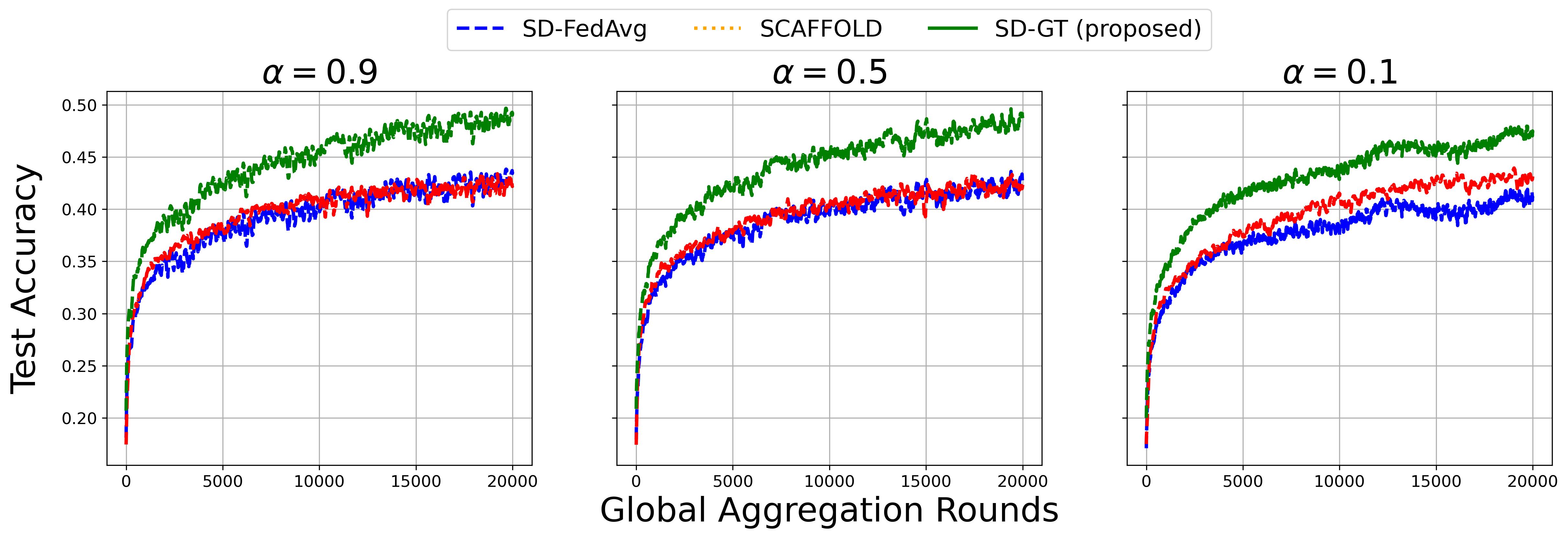}
        \captionsetup{justification=centering}
        \caption{Grouping by proximity (larger intra-subnet data heterogeneity)}
        \label{fig:dirich-group-proximity}       
    \end{subfigure}
    
    \begin{subfigure}[t]{0.95\linewidth}
        \centering
        \includegraphics[width=\linewidth]{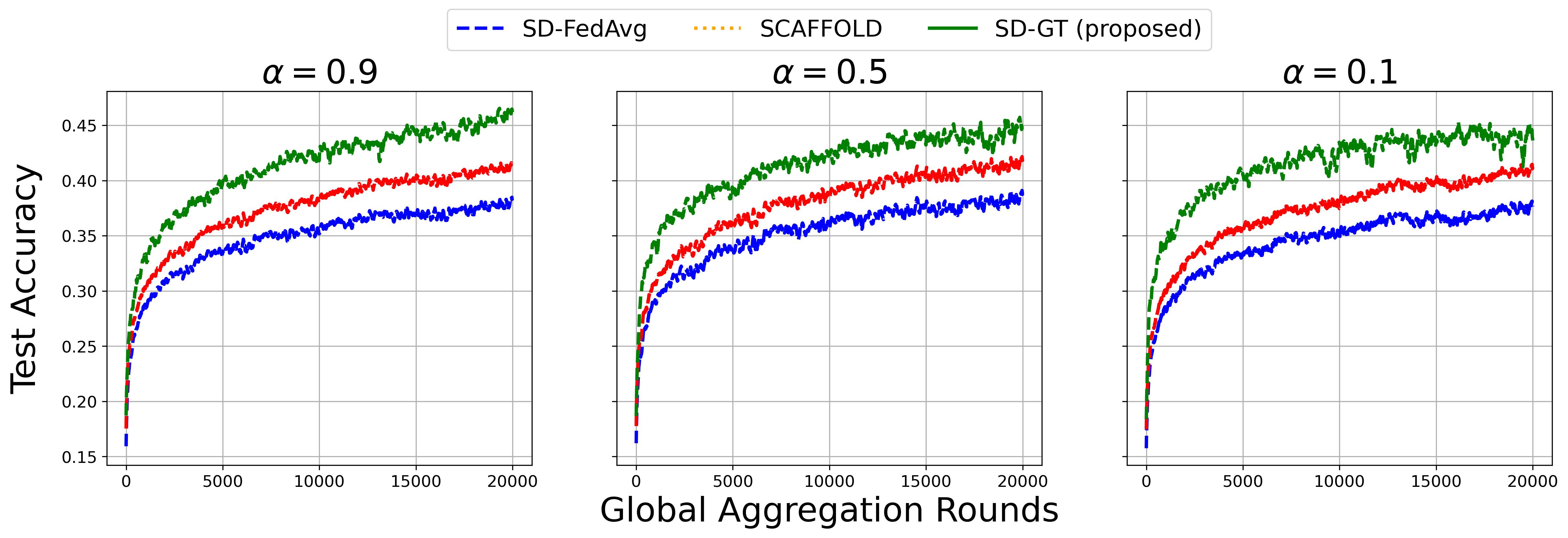}
        \captionsetup{justification=centering}
        \caption{Grouping by data similarity (larger inter-subnet data heterogeneity)}
        \label{fig:dirich-group-similarity}
    \end{subfigure}
    %\captionsetup{justification=centering}
    \caption{\textcolor{black}{Comparison of two subnet grouping policies with data heterogeneity controlled by Dirichlet concentration $\alpha$, on the CIFAR10 dataset.}}
    \label{fig:dirich-comparison}
\end{figure}

\begin{figure}
    \centering    \centerline{\includegraphics[width=0.45\textwidth]{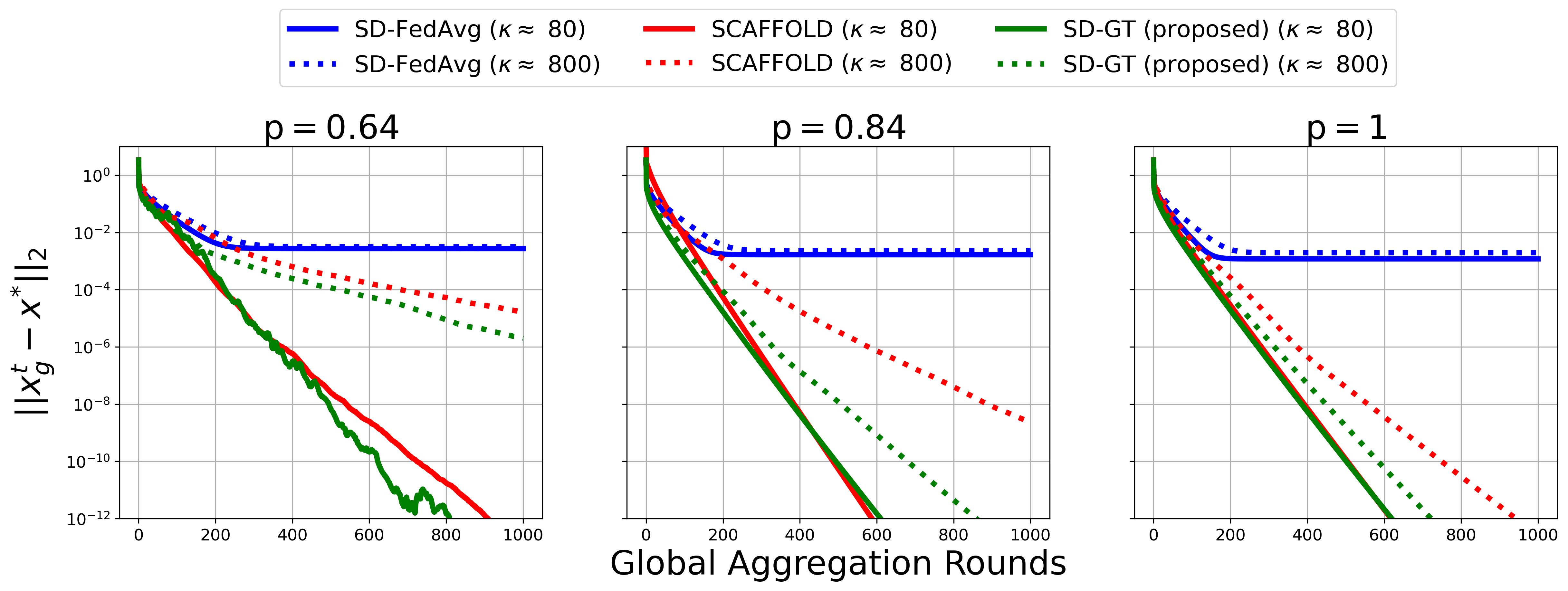}}
\caption{\textcolor{black}{Experimental results from synthetic dataset with $n = 30$ total clients and $S = 6$ total subnets $(K = 40)$. We see that {\tt SD-GT} obtains a linear convergence rate for each choice of sampling rate ($p = 0.64, 0.84, 1$) and condition number ($\kappa \approx 80, 800$), consistent with our theoretical results. Further, we can also observe that sampling rate effects the performance significantly under high conditional number $\kappa$.}\vspace{-0.15in}}
\label{fig:scvx}
\end{figure}
%is able to improve from increasing the sampling rate of each subnet while maintaining linear rate.

\textbf{Impact of Number of Subnets.} 
In Figure \ref{fig:CIFAR10_exp2}, we consider varying the number of subnets $S$ that partition a fixed number of clients $n = 40$ (thus varying number of clients per subnet $m_s$) on CIFAR10. In this setting, each client is assigned a fixed data across all experiments, ensuring that the data heterogeneity across the network remains the same, while both the within-subnet and between-subnet data heterogeneity vary.
%we aim to simulate the scenario where clients are disconnected from their original subnet and then reconnected to a different subnet without altering their local data. In this setting, each client is assigned a fixed subset of data across all experiments, ensuring the data heterogeneity across the network remains the same, while both the data-heterogeneity between and within subnets varies. We fix the total number of clients at $n = 40$, and change both the number of clients within each subnet $m_s$ and the total number of subnets $S$, conducting experiments on the CIFAR10 and CIFAR100 datasets

When the network contains a large number of small subnets (i.e., the left plots), {\tt SD-GT} training depends more on global aggregations and between-subnet gradient tracking $y_i^t$. Conversely, on networks that contains a small number of large subnets (i.e., the right plots), {\tt SD-GT} relies more on D2D communications and the impact of the within-subnet gradient tracking terms $z_i^t$. A key observation is that the performance of SD-FedAvg varies significantly depending on the subnet grouping. Additionally, {\tt SD-GT} outperforms both SCAFFOLD and SD-FedAvg in all cases. \textit{This emphasizes the importance of correcting client drift through both the gradient tracking terms within-subnet $z_i^t$ and between-subnet $y_i^t$,} as solutions like SD-FedAvg based on D2D communication with local gradient updates are not sufficient when the data-heterogeneity within/between subnets changes.
%, only using D2D communication with local gradient updates isn't enough, emphasizing the importance of correcting client drift through within-subnet gradient tracking term $z_i^t$ and between-subnet gradient tracking term $y_i^t$.}

% This can be shown by the CIFAR10 results of SD-FedAvg, where the performance improves when subnet set sizes are large (the importance of D2D communication increases). However, we can also observe that SD-FedAvg failed is struggling to gain improvement from D2D communication in CIFAR100, showing that \textit{to resolve high data-heterogeneity, only using D2D communication with local gradient update isn't enough, hence showing the importance of correct client drift through within-subnet gradient tracking.}

% Since all subnets are connected to the server in a centralized way, we can see that for all three methods, performance difference between networks with small amount of subnets is small when comparing to networks with large amount of subnets, indicating that if the dataset is spreaded evenly throughout the network, fixing subnet sizes and increasing subnet numbers has a similar effect to increasing the total number of agents in a fully centralized federated learning network. This is also indicated by Theorem \ref{thm1}. Since under this experiment, not only we are fixing the D2D communication round $K$, but the sampling rate is also fixed, which means we have a constant $p$. Plus all subnets have the same size (and same generation rule), which means $q$ is also approximately the same.

\textcolor{black}{\textbf{Varying data heterogeneity and subnetting:}
In Figure~\ref{fig:dirich-comparison}, we examine two interrelated notions: varying levels of non-i.i.d. data and subnetting strategies. We use the Dirichlet($\alpha$) distribution, for concentration parameter $\alpha > 0$, to allocate data labels across clients, where smaller values of $\alpha$ produce higher levels of non-i.i.d. Additionally, we consider different policies for grouping devices into subnets, which can impact the subnet graph connectivity and the data heterogeneity within and between subnets.}
% \textcolor{blue}{
% In our methodology, we assume the subnetting strategy is fixed beforehand (e.g., according to network operator policy), and make no particular assumptions on how the subnets are determined, similar to in prior works on semi-decentralized FL~\cite{lin2021semi,hosseinalipour2022multi,yemini2022semi}. The major impact of such policies on our algorithm would manifest from changes in subnet graph connectivity (i.e., $\rho_s$ in Assumption 2), and changes in local dataset heterogeneity within and between subnets (i.e., terms $z$ and $y$ in Algorithm 1). To investigate this, we have conducted experiments on two distinct subnet grouping policies: }

\textcolor{black}{\textit{(i) Grouping by proximity (larger intra-subnet data heterogeneity):} This corresponds to our standard system setup, except now, samples are allocated across devices according to Dirichlet($\alpha$). Thus, devices are grouped into subnets based on their ability to form short distance D2D links, which has communication efficiency advantages, but also
%This ensures that intra-subnet communication is more efficient than device-server (DS) communication.
does not control data heterogeneity within subnets (intra-subnet heterogeneity).}

\textcolor{black}{\textit{(ii) Grouping by data similarity (larger inter-subnet data heterogeneity):} In this case, we apply $K$-means across the distribution of classes possessed by each device, to subnet by similar data distributions. This helps reduce intra-subnet heterogeneity, but does not control proximity, and also increases data heterogeneity across subnets. We increase the geometric graph radius from $r \in [2.5, 3.5]$ to $r \in [6, 7]$ here to provide similar subnet connectivity in each case.}
%In this case, D2D communication within a subnet may not be more efficient than DS communication, as physically distant clients might still be grouped together. This setup leads to reduced intra-subnet heterogeneity but increased inter-subnet heterogeneity.

% \textcolor{blue}{
% In Fig.~\ref{fig:dirich-comparison}, we compare the effect of data heterogeneity across subnets and within subnets. We use Dirichlet distribution with $\alpha = 0.1, 0.5, 0.9$ to simulate different level of heterogeneity, we fixed the network to be 30 clients, 6 subnets, and $K=3$, the results shows that the inter-subnet have a higher impact on accuracy than the intra-subnet for both our method and SD-FedAvg. Which aligns with our intuition since D2D communication helps mitigate the heterogeneity within each subnet.}

\textcolor{black}{In Figure~\ref{fig:dirich-comparison}, we see that, for each value of $\alpha$, the policy of grouping by proximity tends to achieve a higher performance than grouping by data similarity. The gap between SD-FedAvg and SCAFFOLD closes when grouping by proximity, as D2D cooperative consensus rounds help mitigate intra-subnet data heterogeneity similar to gradient tracking within subnets. In each case, our SD-GT methodology obtains the highest performance, showing our joint design of cooperative consensus and gradient tracking provides improvements for different subnetting policies and non-i.i.d. distributions.}
%, showing our joint design of cooperative consensus and gradient tracking provides improvements for different subnetting policies
\vspace{-0.1in}

\subsection{Strongly Convex Learning Tasks}
\label{sec:V-D}
% \textbf{Experiment 4: Convergence results for varying DS sample rates.} 
\textcolor{black}{Figure \ref{fig:scvx} compares the learning performance between the algorithms on the synthetic dataset, a strongly convex task. We set the number of D2D rounds to $K = 40$ for all experiments and consider (i) different sampling rates $p$, varying from $0.64$ (corresponding to $\beta_s = 0.6$ for each subnet) to $1$ (full participation, $\beta_s = 1$), and (ii) different condition numbers $\kappa = 80$ (lower complexity task) and $800$ (higher complexity). We see that {\tt SD-GT} has a linear convergence to the globally optimal solution for each combination of $p$ and $\kappa$. This aligns with our result in Theorem~\ref{thm3} (note there is no stochasticity in the gradients here, i.e., $\sigma^2 = 0$). Also, we observe that SD-FedAvg converges to a non-optimal solution related to data heterogeneity, since we are using a constant step size \cite{nemirovski2009robust, li2019convergence}. 
% We observe from Figure \ref{fig4} that all experiments on SD-FedAvg converge to a similar radius from the optimal solution.
%When $\kappa$ is large, D2D communication improves the overall convergence speed, which gives our algorithm a faster linear convergence speed compared to SCAFFOLD. 
Additionally, SCAFFOLD, while converging linearly, does so at a slower rate than {\tt SD-GT} when $\kappa$ is large. Without D2D communications, the model uploaded from each device to the server in SCAFFOLD contains gradient information of only a single device, instead of from the whole subnet. Once again, this emphasizes the performance benefits of combining D2D communications with gradient tracking in our algorithm. Finally, while the impact of $p$ is not as obvious for $\kappa = 80$, under $\kappa = 800$, we can see that higher $p$ gives {\tt SD-GT} a considerable improvement in convergence speed.}

\begin{figure*}
\centerline{\includegraphics[width=0.95\textwidth]{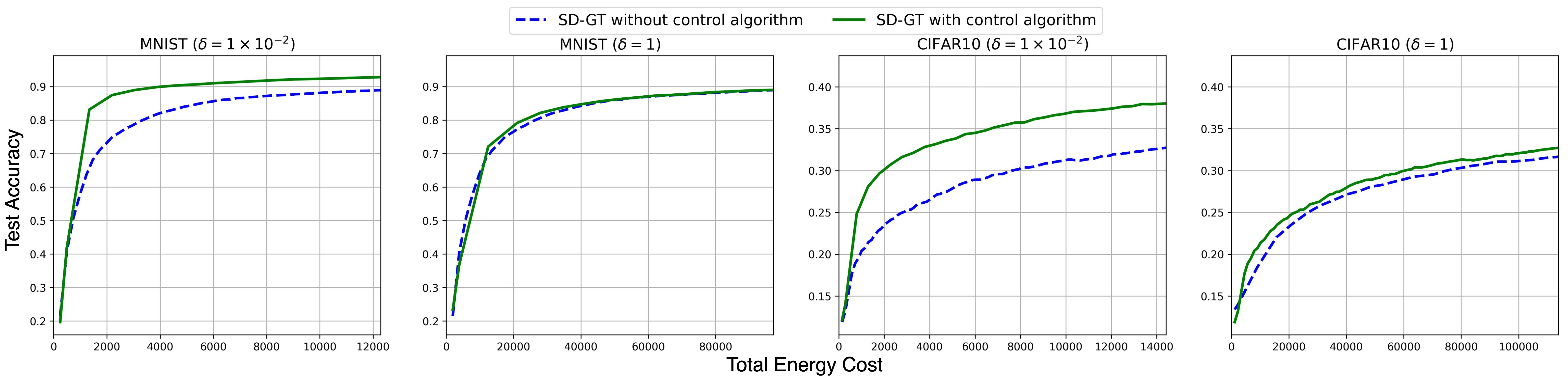}}
\caption{Experiments comparing {\tt SD-GT}'s performance-efficiency tradeoff with and without the proposed control algorithm. When $\delta$ is small (D2D communication is cheap), the co-optimization algorithm is able to adaptively select the sampling rate and the number of D2D communication rounds such that a higher quality trained model is obtained using the same amount of communication energy.\vspace{-0.15in}}
\label{fig:control}
\end{figure*}

\vspace{-0.1in}
\subsection{Adaptive Control Algorithm}
\label{sec:V-E}
Figure \ref{fig:control} evaluates the impact of {\tt SD-GT}'s learning-efficiency co-optimization procedure described in Sec.~\ref{sec:IV-E}. We conduct experiments on MNIST and CIFAR10, comparing two versions of {\tt SD-GT}. The first version is the {\tt SD-GT} without dynamic control, where a constant sampling rate $\frac{h_s}{m_s} = 40\%$ and a constant D2D communication update round $K = 3$ are used, decided before the training process begins. The second version is {\tt SD-GT} employing Algorithm \ref{alg:2} as the dynamic control mechanism. The initial sampling rate is set to $\frac{h_s}{m_s} = 20\%$ and the initial number of D2D communication rounds is set to $K_1 = 1$. We consider two values of energy ratios, $\delta = 0.01$ (D2D is cheap) and $\delta = 1$ (D2D and DS are the same). The DS communication costs are sampled randomly over the range $[1, 100]$ for each subnet, with D2D costs calculated according to $\delta$.
%These numbers will dynamically update throughout the training process.

\begin{figure}[t]
    \centering
    \includegraphics[width=0.48\textwidth]{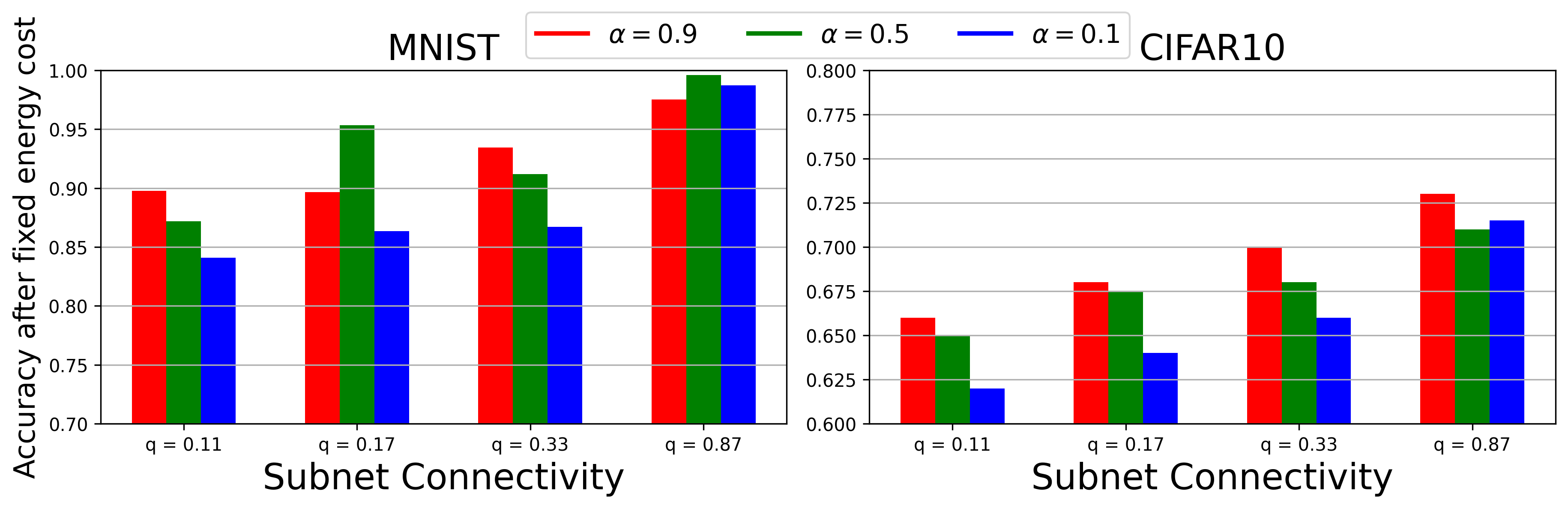}
    \caption{\textcolor{black}{Impact of subnet connectivity (mixing rate $q$) and inter-subnet data heterogeneity (Dirichlet concentration $\alpha$) on the performance of Algorithm~\ref{alg:2}, corroborating the convergence analysis in Appendix~\ref{appen:ctrl_convergence}.}}
    \label{fig:ctrl-results}
    \vspace{-0.15in}
\end{figure}
\begin{figure}[t]
    \centering
    \includegraphics[width=0.48\textwidth]{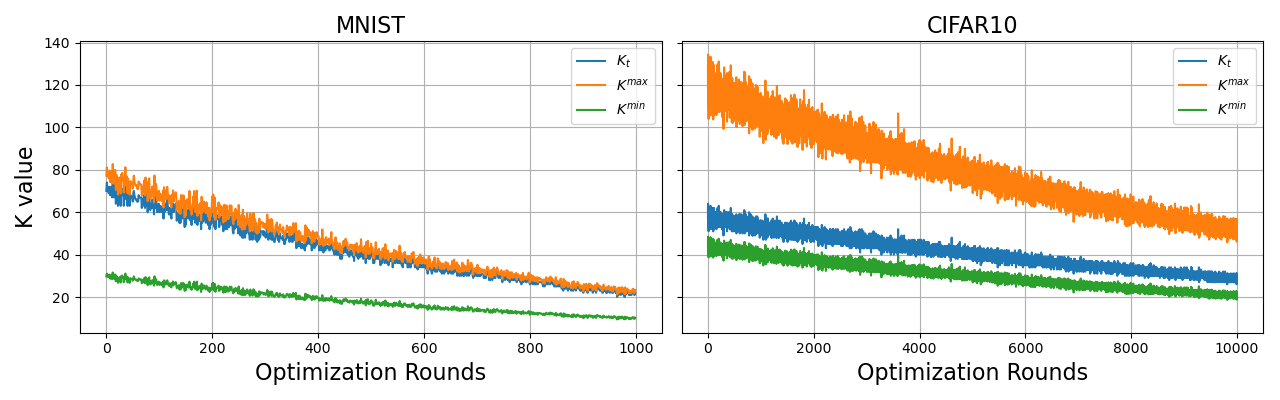}
    \caption{\textcolor{black}{Evolution of the aggregation interval $K_t$ obtained from solving Algorithm~\ref{alg:2} in comparison with the theoretically derived minimum $K^{min}$ and maximum $K^{max}$ provided in Appendix~\ref{appen:ctrl_convergence}.}}
    \label{fig:K_minmax}
    \vspace{-0.15in}
\end{figure}

Overall, we see that the adaptive control algorithm leads to improved {\tt SD-GT} performance in terms of the total energy incurred to reach a given testing accuracy. Our algorithm is able to dynamically choose the sampling rate and D2D update rounds based on the problem's training status and the given energy constraints. The largest gains are observed when $\delta$ is small, allowing the co-optimization algorithm to balance between convergence speed and communication efficiency, resulting in better convergence while incurring the same communication cost. By contrast, when $\delta$ is large, DS and D2D communications have similar cost, and thus there is less benefit to performing multiple D2D communication rounds.

{\color{black} In Figure~\ref{fig:ctrl-results}, we evaluate our adaptive control algorithm under varying degrees of inter-subnet data heterogeneity $\alpha$ and subnet connectivity $q$. The experiments follow our default setup, with cost ratio $\delta = 1 \times 10^{-2}$. Accuracy is measured after the total energy cost reaches $12 \times 10^3$ and $50 \times 10^3$ for MNIST and CIFAR-10, respectively. We see that increasing the value of $q$ tends to enhance performance across both datasets, a trend which aligns with the theoretical findings in Appendix~\ref{appen:ctrl_convergence}, where the convergence rate of Algorithm~\ref{alg:2} is increasing in $q$. On the other hand, while varying $\alpha$ has some impact on convergence, the resulting variation remains minor, particularly for well-connected subnets (large $q$). This can be attributed to (i) the gradient tracking mechanism in {\tt SD-GT} reducing model drift, and (ii) the adaptive control algorithm dynamically balancing energy efficiency with the overall algorithmic performance.

Finally, in Figure~\ref{fig:K_minmax}, we investigate the evolution of the length of aggregation interval $K_t$ over optimization iterations $t$. For both datasets, we see that $K_t$ lies between the minimum ($K^{min}$) and maximum ($K^{max}$) values obtained from analysis of the optimization~\eqref{control_alg_prob} that we provide in Appendix~\ref{appen:ctrl_convergence}.}

%, because D2D communication is not more communication efficient than global aggregation after every gradient computation. When $\delta$ is small, 
%We compare the improvement gained from our co-optimization algorithm under a larger delta ($\delta=1$) and a smaller delta ($\delta = 10^{-2}$). 

\vspace{-0.1in}
\section{Conclusion}
In this paper, we developed {\tt SD-GT}, the first gradient-tracking based semi-decentralized federated learning (SD-FL) methodology. {\tt SD-GT} incorporates dual gradient tracking terms to mitigate the subnet-drift challenge. We provided a convergence analysis of our algorithm under non-convex and strongly-convex settings, revealing conditions under which linear and sub-linear convergence rates are obtained. Based on our convergence results, we developed a low-complexity co-optimization algorithm that trades-off between learning quality and communication cost. Our experimental results demonstrated the improvements provided by {\tt SD-GT} over baselines in SD-FL and gradient tracking literature.

% \chris{Future work can explore ...}
%This work gives rise to several potential future directions. One is to explore the extension of {\tt SD-GT} to multi-layer network hierarchies, using multiple gradient tracking terms. Another potential direction is to explore the impact of adversarial attacks on gradient tracking-enhanced FL methods.

% \newpage
\vspace{-0.1in}
\bibliographystyle{IEEEtran}
\bibliography{egbib}

% \begin{IEEEbiography}[{\includegraphics[width=1in,height=1.25in,clip,keepaspectratio]{fig1.png}}]{Evan Chen}(Student Member, IEEE)
% \end{IEEEbiography}

% \begin{IEEEbiography}[{\includegraphics[width=1in,height=1.25in,clip,keepaspectratio]{fig2.png}}]{Shiqiang Wang}(Senior Member, IEEE)
% \end{IEEEbiography}

% \begin{IEEEbiography}[{\includegraphics[width=1in,height=1.25in,clip,keepaspectratio]{fig3.png}}]{Christopher G. Brinton}(Senior Member, IEEE)
% \end{IEEEbiography}

\newpage

\begin{IEEEbiographynophoto}{Evan Chen (Student Member, IEEE)}
received the B.S. degree in electronic engineering from National Chiao Tung University, Taiwan, in 2021. He is currently pursuing the Ph.D. degree in electrical and computer engineering with Purdue University, USA. 
His research interests lie in distributed and federated machine learning, large-scale optimization, and edge intelligence systems, with a particular focus on communication-efficient learning, adaptive optimization, and privacy-preserving techniques.
\end{IEEEbiographynophoto}

\begin{IEEEbiographynophoto}{Shiqiang Wang (Fellow, IEEE)}
is a Professor of Artificial Intelligence in the Department of Computer Science, University of Exeter, United Kingdom. He was a researcher at IBM T. J. Watson Research Center, NY, United States until Oct. 2025. He received his Ph.D. from Imperial College London, United Kingdom, in 2015. His research focuses on the intersection of artificial intelligence (AI), distributed computing, and optimization, with a broad range of applications including large language models (LLMs), agentic AI, efficient model training and inference, and AI in distributed systems. He has made foundational contributions to edge computing and federated learning that generated both academic and industrial impact. Dr. Wang served as an associate editor of the IEEE Transactions on Mobile Computing, IEEE Transactions on Parallel and Distributed Systems, and IEEE Transactions on Computational Social Systems. He also served as an area chair of major AI and machine learning conferences, including AAAI, ICLR, ICML, and NeurIPS. He received the IEEE Communications Society (ComSoc) Leonard G. Abraham Prize in 2021, IEEE ComSoc Best Young Professional Award in Industry in 2021, Best Paper Runner-Up of ACM MobiHoc 2025, IBM Outstanding Technical Achievement Awards (OTAA) in 2019, 2021, 2022, and 2023, multiple Invention Achievement Awards from IBM since 2016, and Best Student Paper Award of the Network and Information Sciences International Technology Alliance (NIS-ITA) in 2015. 
\end{IEEEbiographynophoto}
\begin{IEEEbiographynophoto}{Christopher G. Brinton (Senior Member, IEEE)} is the Elmore Associate Professor of Electrical and Computer Engineering (ECE) and Faculty Director of the HIVE Engineering Entrepreneurship Center at Purdue University. His research interest is at the intersection of machine learning, communications, and networking, specifically in fog/edge network intelligence, distributed machine learning, and AI/ML-inspired wireless network optimization. Dr. Brinton is a recipient of five of the US top early career awards, from the National Science Foundation (CAREER), Office of Naval Research (YIP), Defense Advanced Research Projects Agency (YFA and Director’s Fellowship), and Air Force Office of Scientific Research (YIP). He is also a recipient of the IEEE Communication Society William Bennett Prize Best Paper Award, the Intel Rising Star Faculty Award, the Qualcomm Faculty Award, and Purdue College of Engineering Faculty Excellence Awards in Early Career Research, Early Career Teaching, and Online Learning. Dr. Brinton currently serves as an Associate Editor for IEEE/ACM Transactions on Networking, and previously was an Associate Editor for IEEE Transactions on Wireless Communications. Prior to joining Purdue, Dr. Brinton was the Associate Director of the EDGE Lab and a Lecturer of Electrical Engineering at Princeton University. He also co-founded Zoomi Inc., a big data startup company that has provided learning optimization to more than one million users worldwide and holds US Patents in machine learning for education. His book The Power of Networks: 6 Principles That Connect our Lives and associated Massive Open Online Courses (MOOCs) have reached over 400,000 students. Dr. Brinton received the PhD (with honors) and MS Degrees from Princeton in 2016 and 2013, respectively, both in Electrical Engineering.

\end{IEEEbiographynophoto}
\newpage

\appendices
\onecolumn
% \appendix
% \section*{Appendices}
% \addcontentsline{toc}{section}{Appendices}
\renewcommand{\thesection}{\Alph{section}}

\section{Proof of general lemmas}
\label{appendixA}
\noindent In this section, we will discuss all the required results to derive the three convergence theorems. In part 1, we show that the expected model aggregated at the server will only collect the sum of all local computed gradients, i.e., all gradient tracking terms that are passed towards it is not accumulated throughout the iteration. Then, in part 2, by iterating the server's model through global iteration $t$, we control the following local update terms with appropriate step sizes: \\[-0.1in]
\begin{enumerate}
    \item deviation term caused by the server sampling clients,
    \item error term from the local gradient tracking term,
    \item error term from the global gradient tracking term. \\[-0.1in]
\end{enumerate}
Finally, in part 3, we derive the descent lemmas based on the non-convex/convex/strongly-convex assumptions that will decide the convergence speed of the algorithm.

% \subsection*{Part 1: Technical tools}
% The following are some tools that will be commonly used in the later proofs:
% \begin{lemma}
    
% \end{lemma}
\subsection*{\textbf{Part 1: Iteration of the Server's Model}}
\noindent We denote \textbf{$x_g^t$ to be the global model aggregated at the central server at iteration $t$}. We then define the \textbf{client sampling random variable $\mathcal{S}$} and \textbf{the expected global model w.r.t $\mathcal{S}$} to be $\overline{x}_g^t =\mathbb{E}_{\mathcal{C}'_s}[x_g^t]$. After taking expectation on the random variable $\mathcal{S}$, we can see that:
\begin{equation}
\begin{aligned}
        \overline{x}_g^t &= \mathbb{E}_{\mathcal{C}'_s}[x_g^{t-1} + \frac{1}{S\cdot h_{s}}\sum_{s = 1}^{S}\sum_{j=1}^{h_{s}}  \Tilde{x}_{{s,j}}^{t+1}]\\
        &= \overline{x}_g^{t-1} + \mathbb{E}_{\mathcal{C}'_s}\bigg[\frac{1}{S\cdot h_{s}}\sum_{s = 1}^{S}\sum_{j=1}^{h_{s}}  (x_{{s,j}}^{t, K+1} - x_{{s,j}}^{t, 1}+ K\gamma y_{{s,j}}^t)\bigg]\\
        &= \overline{x}_g^{t-1} + \frac{1}{n}\sum_{i = 1}^n \left(x_{{i}}^{t, K+1} - x_{{i}}^{t, 1} + K\gamma \mathbb{E}_{\mathcal{C}'_s}[y_{{i}}^t]\right)
        \label{expectation_of_xg}
\end{aligned}    
\end{equation}
Here we introduce some commonly used terminology for later proofs. We first define $W$ to be a block-wise diagonal matrix that corresponds to the D2D communication performed throughout the whole network, as in~\eqref{eq:W}:
\begin{equation}
    W = \mathrm{diag}(W_{1}, \ldots, W_{S}) \in R^{n\times n}.
\end{equation}
Here $W_s$ corresponds to the D2D communication matrix used for communication within the $s$ subnet. Since all $W_s$ are doubly stochastic, $W$ is also doubly stochastic. We then introduce two other matrices, one that corresponds to the whole network communicating in a fully connected graph:
\begin{equation}
    J = \frac{\mathbf{1}_n \mathbf{1}_n^\top}{n},
\end{equation}
and another that corresponds to each individual subnet having a fully connected graph:
\begin{equation}
    J_c = \mathrm{diag}\left(\frac{\mathbf{1}_{m_{1}} \mathbf{1}_{m_{1}}^\top}{m_{1}},\ldots,\frac{\mathbf{1}_{m_{S}} \mathbf{1}_{m_{S}}^\top}{m_{s}}\right).
\end{equation}
We also define the following gradient matrices that concatenate vectors of gradients:
\begin{equation}
    \begin{aligned}
        \nabla F(x_g^t) &= [\nabla f_1(x_g^t), \ldots, \nabla f_n(x_g^t)],\\
        \nabla F(x^{t,k}) &= [\nabla f_1(x_1^{t,k}), \ldots, \nabla f_n(x_n^{t,k})],\\
        \nabla F(x^{t,k}, \xi^{t,k}) &= [\nabla f_1(x_1^{t,k},\xi_1^{t,k}), \ldots, \nabla f_n(x_n^{t,k},\xi_n^{t,k})].
    \end{aligned}
\end{equation}
Based on the definition of $y_i^t$ and $z_i^t$, we can see that, if we initialize with $\sum_{i = 1}^n y_i^0 = 0$ and $\sum_{j = 1}^{m_s} z_{s,j}^0 = 0$, then:
\begin{equation}
    \begin{aligned}
        \sum_{i = 1}^n \mathbb{E}_{\mathcal{C}'_s}[y_i^t] &= \frac{\sum_{s = 1}^S h_s}{n}\sum_{i=1}^n \mathbb{E}_{\mathcal{C}'_s}[y_i^{t-1}] + \left(1 - (\frac{\sum_{s = 1}^S h_s}{n})\right)\sum_{s=1}^S \mathbb{E}_{\mathcal{C}'_s}[\Psi_{s}] = 0,\\
    \end{aligned}
    \label{sum_of_y}
\end{equation}
\begin{equation}
    \sum_{j=1}^{m_s}z_{s,j}^t = \sum_{j=1}^{m_s}z_{s,j}^{t-1} = 0.
    \label{sum_of_z}
\end{equation}
Plugging \eqref{sum_of_y} and \eqref{sum_of_z} into \eqref{expectation_of_xg} and using the fact that $WJ_c = WJ = 0$, we can get:
\begin{equation}
\begin{aligned}
        \overline{x}_g^t &= \overline{x}_g^{t-1} -  \gamma\frac{1}{n} \sum_{k=1}^K \sum_{i=1}^n\nabla f_i(x_i^{t,k}, \xi_i^{t,k}) - \gamma \sum_{s=1}^S \sum_{j=1}^{m_s}K \mathbb{E}_Q[z_{s,j}^t],\\
        &= \overline{x}_g^{t-1} -  \gamma\frac{1}{n} \sum_{k=1}^K \sum_{i=1}^n\nabla f_i(x_i^{t,k}, \xi_i^{t,k}) .
        \label{expectation_of_xg_2}
\end{aligned}    
\end{equation}

\subsection*{\textbf{Part 2: The Common Lemmas}}
\noindent First, we want to control the deviation of local models when performing D2D updates. The intuition is that if we choose our step size with respect to the inverse of the number of local updates $\frac{1}{K}$, we can control the amount of deviation accumulated over $K$ rounds of D2D communication. Then, after bounding the deviation term $\Delta_t$, the iterative relationship of $\mathcal{Y}_t$, $\mathcal{Z}_t$, and $\Gamma_t$ can all be bounded. \\[-0.1in]

\noindent We define the terms that we want to iteratively control with an appropriate choice of step size: we define $\Gamma_t = \frac{1}{n}\sum_{i=1}^n\mathbb{E}\|x_i^{t-1,K+1} - \overline{x}_\mathrm{g} ^t\|^2$ to be the \textbf{sampling error }term, and $\mathcal{Z}_t = \frac{1}{n}\mathbb{E}\|Z^t + \nabla F(\overline{x}_\mathrm{g} ^t)(I - J_c)\|^2_F$, $\mathcal{Y}_t = \frac{1}{n}\mathbb{E}\|{Y^t} + \nabla F(\overline{x}_\mathrm{g} ^t)(J_c - J)\|_F^2$ to be the \textbf{within-subnet} and \textbf{between-subnet correction} terms, respectively.
\begin{lemma}(Unroll recursion lemma)
    For any parameters $r_0 \geq 0, b \geq 0, e\geq 0, u\geq 0$, there exists a constant step size $\gamma < \frac{1}{u}$ s.t. 
\begin{equation}
    \Psi_T := \frac{r_0}{T}\frac{1}{\gamma} + b\gamma + e\gamma^2 \leq 2\sqrt{\frac{br_0}{T}} + 2e^{\frac{1}{3}}\left( \frac{r_0}{T} \right)^{\frac{2}{3}} + \frac{ur_0}{T}.
\end{equation}
    \label{unroll_lem}
    \vspace{-0.10in}
\end{lemma}
\begin{proof}
    See Lemma C.5 in \cite{liu2023decentralized} or Lemma 15 in \cite{koloskova2020unified}.
\end{proof}
\subsection{\textbf{Proof for Lemma \ref{lem1}}}
\begin{proof}
Before considering $\Delta_t$, we first observe the term $\frac{1}{n}\sum_{i=1}^n\mathbb{E}\|x_i^{t, k} - x_\mathrm{g} ^{t}\|$ for $1 \leq k \leq K$:
    \begin{equation}
        \begin{aligned}
            \frac{1}{n}\sum_{i=1}^n\mathbb{E}\|x_i^{t, k} - \overline{x}_g^{t}\|^2=& \frac{1}{n}\sum_{i=1}^n\mathbb{E}\left\|\sum_{j \in N^{in}_i}w_{ij}\bigg(x_j^{t, k-1} - \gamma (\nabla f_j(x^{t,k-1}, \xi_j^{t, k-1}) + y_j^t + z_j^t)]\bigg) - \overline{x}_g^{t}\right\|^2\\
            =& \frac{1}{n}\sum_{i=1}^n\mathbb{E}\left\|\sum_{j \in N^{in}_i}w_{ij}\bigg(x_j^{t, k-1} - \gamma (\nabla f_j(x^{t,k-1}, \xi_j^{t, k-1}) + y_j^t + z_j^t) - \overline{x}_g^{t}\bigg)\right\|^2\\
            \leq&\frac{1}{n}\sum_{i=1}^n\mathbb{E}\left\|x_i^{t, k-1} - \gamma (\nabla f_i(x^{t,k-1}, \xi_i^{t, k-1}) + y_i^t + z_i^t) - \overline{x}_g^{t}\right\|^2\\
            \leq& (1 + \frac{1}{K-1})\frac{1}{n}\sum_{i=1}^n\mathbb{E}\|x_i^{t,k-1} - \overline{x}_g^t\|^2\\
            &+K\gamma^2\frac{1}{n}\sum_{i=1}^n \mathbb{E}\bigg\Vert\nabla f_i(x^{t, k-1}) - \nabla f_i(\overline{x}_g^t) + \nabla f_i(\overline{x}_g^t) - \frac{1}{m_s}\sum_{j \in \mathcal{C}_s}\nabla f_{j}(\overline{x}_g^t) \\
            &+ \frac{1}{m_s}\sum_{j \in \mathcal{C}_s}\nabla f_{j}(\overline{x}_g^t)- \frac{1}{n}\sum_{j=1}^n \nabla f_j(\overline{x}_g^t) + \frac{1}{n}\sum_{j=1}^n \nabla f_j(\overline{x}_g^t) + y_i^t + z_i^t  \bigg\Vert^2+ \gamma^2 \sigma^2\\
            \leq& (1 + \frac{1}{K-1} + 4K\gamma^2L^2)\frac{1}{n}\sum_{i=1}^n\mathbb{E}\|x_i^{t,k-1} - \overline{x}_g^t\|^2+ 8K\gamma^2\mathcal{Y}_t+ 8K\gamma^2 \mathcal{Z}_t + 2K\gamma^2\mathbb{E}\|\nabla f(\overline{x}_g^t)\|^2 + \gamma^2 \sigma^2.
        \end{aligned}
    \end{equation}
    If we let $(1 + \frac{1}{K-1} + 4K\gamma^2L^2) = \alpha$ we can get:
    \begin{equation}
    \begin{aligned}
            \frac{1}{n}\sum_{i=1}^n\mathbb{E}\|x_i^{t, k} - \overline{x}_g^{t}\|^2
            \leq& 
        \alpha^{k-1}\frac{1}{n}\sum_{i=1}^n\mathbb{E}\|x_i^{t,1} - \overline{x}_g^t\|^2
                + \sum_{k'=0}^{k-2}\alpha^{k'}\bigg(8K\gamma^2\mathcal{Y}_t+ 8K\gamma^2\mathcal{Z}_t + 2K\gamma^2\|\nabla \mathbb{E}f(\overline{x}_g^t)\|^2+ \gamma^2 \sigma^2\bigg).
    \end{aligned}
    \label{lem1_unbounded}
    \end{equation}
    Now, we want to connect the term $\mathbb{E}\|x_i^{t, 1} - \overline{x}_g^{t}\|^2$ with $\Gamma_t$. Since for each iteration, each client $x_i$ has a $\frac{h_s}{m_s}$ probability of being sampled and updated by the server, the initial model at iteration $t$ has probability $\frac{m_s - h_s}{m_s}$ to be the last model updated at iteration $t-1$:
    \begin{equation}
        \mathbb{E}_{\mathcal{C}'_s}[x_i^{t, k}]\leq (\frac{m_s - h_s}{m_s})x_i^{t-1, K+1} + \frac{h_s}{m_s}\overline{x}_g^t.
    \end{equation}
    Thus, if we define $p  = \min(1 - \beta_1^2, \cdots, 1 - \beta_S^2)$, we can see that:
    \begin{equation}
        \frac{1}{n}\sum_{i=1}^n\mathbb{E}\|x_i^{t,1} - \overline{x}_g^t\|^2 \leq (1-p)\Gamma_t.
        \label{lem1_bound2}
    \end{equation}
    If we choose $\gamma < \frac{1}{8KL}$, then we can bound the following terms:
    \begin{equation}
        \begin{aligned}
            \alpha^{k-1} \leq \alpha^K \leq 3,\\
            \sum_{k'=0}^{k-2}\alpha^{k'} \leq 3K.
        \end{aligned}
    \end{equation}
    Combining \eqref{lem1_unbounded} and \eqref{lem1_bound2} together, we can get the following result:
    \begin{equation}
    \begin{aligned}
            \frac{1}{n}\sum_{i=1}^n\mathbb{E}\|x_i^{t, k} - \overline{x}_g^{t}\|^2 
            \leq 3(1-p)\Gamma_t
            + 24K^2\gamma^2\mathcal{Y}_t+ 24K^2\gamma^2\mathcal{Z}_t + 6K^2\gamma^2\mathbb{E}\|\nabla f(\overline{x}_g^t)\|^2 + 3K\gamma^2 \sigma^2.
            \label{lem1_singlek}
    \end{aligned}
    \end{equation}
    Since this bound is the same for any $\frac{1}{n}\sum_{i=1}^n\mathbb{E}\|x_i^{t, k} - \overline{x}_g^{t}\|^2$, $1\leq k \leq K$, multiplying the RHS of \eqref{lem1_singlek} by $K$ yields the upper bound on $\Delta_t$ in Lemma~\ref{lem1}.
\end{proof}

\subsection{\textbf{Proof for Lemma \ref{lem2}}}
\begin{proof}
    For subnet $\mathcal{C}_s$, we define $\beta_s = \frac{m_s - h_s}{m_s}$ to be the ratio of unsampled clients. Then, for each iteration, client $i \in \mathcal{C}_s$ has a probability of $\beta_s$ to stay the same as the previous iteration:
    \begin{equation}
        \mathbb{E}_{\mathcal{C}'_s} y_i^t = \beta_s \mathbb{E}_{\mathcal{C}'_s} y_i^{t-1} + (1 - \beta_s) \mathbb{E}_{\mathcal{C}'_s}\Psi_s^t.
    \end{equation}
    The expected updated information $\mathbb{E}_{\mathcal{C}'_s} \Psi_s$ can be written as:
    \begin{equation}
        \begin{aligned}
            \mathbb{E}_{\mathcal{C}'_s}\Psi_s^t &=  \frac{1}{K}\sum_{k = 1}^K \bigg( \frac{1}{n}\sum_{i=1}^n \nabla f_i(x_i^{t-1,k}, \xi_i^{t-1,k}) - \frac{1}{m_s}\sum_{i\in C_s} \nabla f_i(x_i^{t-1,k}, \xi_i^{t-1,k})\bigg).
    \end{aligned}
    \end{equation}
    For the next part of this proof, it is easier to discuss the properties in matrix form. We define the following matrices:
    \begin{equation}
\begin{aligned}
            \nabla F(x^{t-1,k}, \xi^{t-1,k}) &= \big[\nabla f_1(x_1^{t-1,k}, \xi_1^{t-1,k}), \ldots, \nabla f_n(x_n^{t-1,k}, \xi_n^{t-1,k})\big] \in \mathbb{R}^{d \times n},\\
            \nabla F(x^{t-1,k}) &= \big[\nabla f_1(x_1^{t-1,k}), \ldots, \nabla f_n(x_n^{t-1,k})\big] \in \mathbb{R}^{d \times n}.\\
\end{aligned}
\end{equation}
Under this definition, if we collect all $y_i^t$ into a matrix $Y^t = [y_1^t, \ldots, y_n^t]$, we have:
\begin{equation}
    \mathbb{E}_{\mathcal{C}'_s} Y^t =  \mathbb{E}_{\mathcal{C}'_s} Y^{t-1}\begin{bmatrix}
    \beta_{1}\cdot I_{m_1} & \cdots & 0\\
    \vdots, & \ddots & \vdots\\
    0 & \cdots & \beta_{S}\cdot I_{m_S}
    \end{bmatrix} + \frac{1}{K}\sum_{k=1}^K\nabla F(x^{t-1,k}, \xi^{t-1,k})(J - J_c)\begin{bmatrix}
    (1-\beta_{1})\cdot I_{m_1} & \cdots & 0\\
    \vdots & \ddots & \vdots\\
    0 & \cdots & (1-\beta_{S})\cdot I_{m_S}
    \end{bmatrix}.
\end{equation}
We define two matrices that control the sampling probability as:
\begin{equation}
\begin{aligned}
        {B} &= \begin{bmatrix}
    \beta_{1}\cdot I_{m_1} & \cdots & 0\\
    \vdots & \ddots & \vdots\\
    0 & \cdots & \beta_{S}\cdot I_{m_S}
\end{bmatrix},\\
{B}' &= \begin{bmatrix}
    (1-\beta_{1})\cdot I_{m_1} & \cdots & 0\\
    \vdots & \ddots & \vdots\\
    0 & \cdots & (1-\beta_{S})\cdot I_{m_S}
    \end{bmatrix}.
\end{aligned}
\end{equation}
Then, if we define $p  = \min(1 - \beta_1^2, \cdots, 1 - \beta_S^2)$, we can bound the iteration of $\mathcal{Y}_t$ as follows:
    \begin{equation}
        \begin{aligned}
            n\mathcal{Y}_t =& \mathbb{E}\|Y^t + \nabla f(\overline{x}_g^t)(J_c - J)\|_F^2\\
            =&\mathbb{E}\bigg\| \bigg(Y^{t-1} + \nabla f(\overline{x}_g^{t-1})(J_c - J)\bigg)B+ \bigg(\frac{1}{K}\sum_{k=1}^K\nabla F(x^{t-1,k}, \xi^{t-1,k}) - \nabla F(\overline{x}_g^{t-1})\bigg)(J - J_c)B'\\
            &+ \bigg(\nabla F(\overline{x}_g^t) - \nabla F(\overline{x}_g^{t-1})\bigg)(J_c - J)\bigg\|_F^2\\
            \leq&\mathbb{E}\bigg\| \bigg(Y^{t-1} + \nabla f(\overline{x}_g^{t-1})(J_c - J)\bigg)B+ \bigg(\frac{1}{K}\sum_{k=1}^K\nabla F(x^{t-1,k}) - \nabla F(\overline{x}_g^{t-1})\bigg)(J - J_c)B'\\
            &+ \bigg(\nabla F(\overline{x}_g^t) - \nabla F(\overline{x}_g^{t-1})\bigg)(J_c - J)\bigg\|_F^2 + n\frac{\sigma^2}{K}\\
            \leq& (1 - \frac{p}{2})n\mathcal{Y}_{t-1} + \frac{6}{p}\bigg(\frac{nL^2}{K}\Delta_{t-1} + 2\gamma^2nL^4K\Delta_{t-1} + 2\gamma^2L^2K^2n\mathbb{E}\|\nabla f(\overline{x}_g^{t-1})\|^2 \bigg)+\frac{6\gamma^2L^2Kn}{p}\sigma^2 +  n\frac{\sigma^2}{K}.
            \label{lem2_unbd}
        \end{aligned}
    \end{equation}
    Finally, by choosing $\gamma < \frac{1}{\sqrt{6}KL}$, combining the two $\Delta_{t-1}$ terms in the RHS of \eqref{lem2_unbd} yields the result of Lemma 4.
\end{proof}
\subsection{\textbf{Proof for Lemma \ref{lem3}}}
\begin{proof}
     For each client $i$, the iteration of the tracking term $z_i^t$ can be expressed as:
     \begin{equation}
\begin{aligned}
      z_i^{t+1} &= z_i^t + \frac{1}{K\gamma}\sum_{k=1}^K\left(\Tilde{z}_i^{t,k} - \sum_{j \in \mathcal{N}_i \cup \{i\}} w_{ij}\Tilde{z}_j^{t,k}\right)\\
      &= z_i^t + \frac{1}{K}\sum_{k=1}^K \left(\sum_{j \in \mathcal{N}_i \cup \{i\}} \nabla f_j(x_j^{t,k}, \xi_j^{t,k}) - \nabla f_i(x_i^{t,k}, \xi_i^{t,k}) + \sum_{j \in \mathcal{N}_i \cup \{i\}} z_j^t - z_i^t\right)\\
      &= \sum_{j \in \mathcal{N}_i \cup \{i\}} z_j^t + \frac{1}{K}\sum_{k=1}^K \left(\sum_{j \in \mathcal{N}_i \cup \{i\}} \nabla f_j(x_j^{t,k}, \xi_j^{t,k}) - \nabla f_i(x_i^{t,k}, \xi_i^{t,k})\right).
\end{aligned}
 \end{equation}
 Thus, if we define the matrix $Z^t = [z_1^t, \ldots, z_n^t] \in \mathbb{R}^{d\times n}$, we have:
 \begin{equation}
\begin{aligned}
Z^t = Z^{t-1}W +\frac{1}{K}\sum_{k=1}^K\nabla F(x^{t-1,k}, \xi^{t-1,k})(W - I).
\end{aligned}
 \label{lem3_z_expand}
 \end{equation}
Defining $q = \min(\rho_{1}, \ldots, \rho_{S})\in (0,1]$, we can apply \eqref{lem3_z_expand} to get:
    \begin{equation}
        \begin{aligned}
            n\mathcal{Z}_t =& \mathbb{E}\|Z^t + \nabla F(\overline{x}_g^t)(I - J_c)\|^2_F\\
            =& \mathbb{E}\left\|Z^{t-1}W + \frac{1}{K}\sum_{k=1}^K\nabla F(x^{t-1,k}, \xi^{t-1,k})(W - I) + \nabla F(\overline{x}_g^t)(I - J_c)\right\|_F^2\\
            =& \mathbb{E}\bigg\|Z^{t-1}W + \nabla F(\overline{x}_g^{t-1})(W - J_c) + \frac{1}{K}\sum_{k=1}^K(\nabla F(x^{t-1,k}) - \nabla F(\overline{x}_g^{t-1}))(W - I) \\
            &+ (\nabla F(\overline{x}_g^t) - \nabla F(\overline{x}_g^{t-1}))(I - J_c)\bigg\|_F^2+ \frac{\sigma^2}{K}\\
            \leq&(1-\frac{q}{2})n\mathcal{Z}_{t-1} + \frac{6n}{q}\bigg(\frac{4L^2}{K}\Delta_{t-1} + \|\nabla F(\overline{x}_g^t) - \nabla F(\overline{x}_g^{t-1})\|_F^2\bigg) + \frac{n\sigma^2}{K}\\
            \leq&(1-\frac{q}{2})n\mathcal{Z}_{t-1} + \frac{6n}{q}\bigg(\frac{4L^2}{K}\Delta_{t-1} + 2K\gamma^2L^4\Delta_{t-1} + 2K^2L^2\gamma^2\mathbb{E}\|\nabla f(\overline{x}_g^t)\|^2 + L^2K\gamma^2\sigma^2\bigg)+ \frac{n\sigma^2}{K}.\\
        \end{aligned}
    \end{equation}
With $ \gamma < \frac{1}{\sqrt{6}KL}$, we arrive at the result in Lemma 5:
    \begin{equation}
        \begin{aligned}
            \mathcal{Z}_t \leq &(1-\frac{q}{2})\mathcal{Z}_{t-1} + \frac{26L^2}{q}\Delta_{t-1} + \frac{12}{q}K^2L^2\gamma^2\mathbb{E}\|\nabla f(\overline{x}_g^t)\|^2 + \frac{2\sigma^2}{qK}.\\
        \end{aligned}
    \end{equation}
\end{proof}
\subsection{\textbf{Proof for Lemma \ref{lem4}}}
\begin{proof}
We have:
\begin{equation}
    \begin{aligned}
        \Gamma_t &= \frac{1}{n}\sum_{i=1}^n E\|x_i^{t-1,K+1} - \overline{x}_g^t\|^2\\
        &\leq \frac{1}{n}\sum_{i=1}^nE\left\|x_i^{t-1,K+1} - \overline{x}_g^{t-1} + \gamma\sum_{k=1}^K\frac{1}{n}\sum_{i=1}^n \nabla f_i(x_i^{t-1,k}, \xi_i^{t-1,k})\right\|^2\\
        &= \frac{1}{n}\sum_{i=1}^nE\left\|x_i^{t-1,K+1} - \overline{x}_g^{t-1} + \gamma\sum_{k=1}^K\frac{1}{n}\sum_{i=1}^n \nabla f_i(x_i^{t-1,k}, \xi_i^{t-1,k}) - \gamma \sum_{k=1}^K\frac{1}{n}\sum_{i=1}^n \nabla f_i(\overline{x}_g^{t-1}) + \gamma \sum_{k=1}^K\frac{1}{n}\sum_{i=1}^n \nabla f_i(\overline{x}_g^{t-1})\right\|^2\\
        &\leq (1 - \frac{p}{2})\Gamma_{t-1} + K\gamma^2\sigma^2\\
        &+ \frac{3\gamma^2}{np}K\sum_{k=1}^K\left\|(\nabla F(x^{t-1,k}) - \nabla F(\overline{x}_g^{t-1}))(I-J) + Z^{t-1} + \nabla F (\overline{x}_g^{t-1})(I - J_c) + Y^{t-1} + \nabla F (\overline{x}_g^{t-1})(J_c - J) + \nabla F (\overline{x}_g^{t-1})J\right\|_F^2.\\
    \end{aligned}
\end{equation}
Now, injecting $L$-smoothness and using the fact that $\|I - J\| \leq 1$, we arrive at Lemma 6:
\begin{equation}
    \Gamma_t\leq (1 - \frac{p}{2})\Gamma_{t-1} + \frac{12}{p}\gamma^2KL^2\Delta_{t-1} + \frac{12}{p}\gamma^2K^2\mathcal{Y}_{t-1} + \frac{12}{p}\gamma^2K^2\mathcal{Z}_{t-1} + \frac{12}{p}\gamma^2K^2\|\nabla f(\overline{x}_g^{t-1})\|^2 + K\gamma^2\sigma^2.
\end{equation}
\end{proof}

\subsection*{\textbf{Part 3: Convergence Lemmas for Different Assumptions}}
\noindent Lastly, we provide proofs for Lemma \ref{lem5}, which uses assumptions on $L$-smooth, and Lemma~\ref{lem6}, which uses assumptions on both convexity and $L$-smooth.

\subsection{\textbf{Proof for Lemma \ref{lem5}}}
\begin{proof}
    Using $\eqref{expectation_of_xg_2}$ and injecting $L$-smoothness, we can see that:
\begin{equation}
    \begin{aligned}
        \mathbb{E}f(\overline{x}_g^{t+1}) &\leq \mathbb{E}f(\overline{x}_g^{t}) + \langle \nabla \mathbb{E}f(\overline{x}_g^t), \overline{x}_g^{t+1} - \overline{x}_g^t\rangle + \frac{L}{2}\mathbb{E}\|\overline{x}_g^{t+1} -\overline{x}_g^{t}\|^2\\
        & \leq \mathbb{E}f(\overline{x}_g^{t}) + \underbrace{\mathbb{E}\bigg\langle \nabla f(\overline{x}_g^t), - \frac{1}{n}\gamma \sum_{i,k}\nabla f_i(x_i^{t,k}, \xi_i^{t,k})\bigg\rangle}_\text{term 1} 
         + \underbrace{\frac{L}{2}\mathbb{E}\left\|\frac{1}{n}\gamma \sum_{i,k}\nabla f_i(x_i^{t,k}, \xi_i^{t,k})\right\|^2}_\text{term 2}.
    \end{aligned}
\end{equation}
For term 1:
\begin{equation}
    \begin{aligned}
        \mathbb{E}\bigg\langle \nabla f(\overline{x}_g^t), - \frac{1}{n}\gamma \sum_{i,k}\nabla f_i(x_i^{t,k}, \xi_i^{t,k})\bigg\rangle
        =& \gamma\bigg\langle \nabla f(\overline{x}_g^t), \frac{1}{n} \sum_{i,k}-\nabla f_i(E[x_i^{t,k}]) - K\nabla f(\overline{x}_g^t) + K\nabla f(\overline{x}_g^t)\bigg\rangle\\
        =& -\gamma K\mathbb{E}\|\nabla f(\overline{x}_g^t)\|^2-\gamma\mathbb{E}\bigg\langle \nabla f(\overline{x}_g^t), \frac{1}{n} \sum_{i,k}\nabla f_i(x_i^{t,k}, \xi_i^{t,k}) - K\nabla f(\overline{x}_g^t) \bigg\rangle.
    \end{aligned}
\end{equation}
Since $\mathbb{E}_{\xi_i} \nabla f_i(x_i, \xi_i) = \nabla f_i(x_i)$, by $L$-smoothness, we can show that:
\begin{equation}
    \mathbb{E}\bigg\langle \nabla f(\overline{x}_g^t), - \frac{1}{n}\gamma \sum_{i,k}\nabla f_i(x_i^{t,k}, \xi_i^{t,k})\bigg\rangle \leq -\frac{\gamma K}{2}\mathbb{E}\|\nabla f(\overline{x}_g^t)\|^2 + \frac{\gamma L^2}{2}\frac{1}{n}\sum_{i,k}\mathbb{E}\|x_i^{t,k} - \overline{x}_g^t\|^2.
\end{equation}
For term 2:
\begin{equation}
    \begin{aligned}
        \frac{L}{2}\mathbb{E}\left\|\frac{1}{n}\gamma \sum_{i,k}\nabla f_i(x_i^{t,k})\right\|^2 
        =& \frac{L\gamma^2}{2}\mathbb{E}\left\|\frac{1}{n} \sum_{i,k}\nabla f_i(x_i^{t,k}) - K\nabla f(\overline{x}_g^t) + K\nabla f(\overline{x}_g^t)\right\|^2\\
        \leq& L^3K\gamma^2\frac{1}{n}\sum_{i,k}\mathbb{E}\|x_i^{t,k} - \overline{x}_g^t\|^2 + L^2K^2\gamma^2\mathbb{E}\|\nabla f(\overline{x}_g^t)\|^2 + \frac{L\gamma^2K}{2n}\sigma^2.
    \end{aligned}
\end{equation}
Combining both terms and choosing the step size as $\gamma \leq \frac{1}{4KL}$, we arrive at the result:

\begin{equation}
    \begin{aligned}
    \mathbb{E}f(\overline{x}_g^{t+1}) \leq & \mathbb{E}f(\overline{x}_g^{t}) + \bigg( L^3K\gamma^2 + \frac{\gamma L^2}{2}\bigg)\frac{1}{n}\sum_{i,k}\mathbb{E}\|x_i^{t,k} - \overline{x}_g^t\|^2 + \bigg(L^2K^2\gamma^2-\frac{\gamma K}{2}\bigg)\mathbb{E}\|\nabla f(\overline{x}_g^t)\|^2\\
    \leq & \mathbb{E}f(\overline{x}_g^{t}) + \gamma L^2\frac{1}{n}\sum_{i,k}\mathbb{E}\|x_i^{t,k} - \overline{x}_g^t\|^2 -\frac{\gamma K}{4}\mathbb{E}\|\nabla f(\overline{x}_g^t)\|^2 + \frac{L\gamma^2K}{2n}\sigma^2.
    \end{aligned}
\end{equation}
\end{proof}
\subsection{\textbf{Proof for Lemma \ref{lem6}}}
\begin{proof}
We start by bounding the distance between the expected model $\overline{x}_g^t$ at the server and the optimal model $x^\star$:
\begin{equation}
    \begin{aligned}
        \mathbb{E}\|\Bar{x}_g^{t+1} - x^\star\|^2 & =   \mathbb{E}\|\Bar{x}_g^t - x^\star\|^2 - \frac{2\gamma}{n} \mathbb{E}\langle\Bar{x}_g^t - x^\star, \sum_{i=1}^n\sum_{k=1}^K \nabla f_i(x_i^{t,k}, \xi_i^{t,k}) \rangle + \gamma^2 \mathbb{E}\|\frac{1}{n}\sum_{i=1}^n\sum_{k=1}^K \nabla f_i(x_i^{t,k}, \xi_i^{t,k})\|^2\\
        &\leq \mathbb{E}\|\Bar{x}_g^t - x^\star\|^2 - \frac{2\gamma}{n} \mathbb{E}\langle\Bar{x}_g^t - x^\star, \sum_{i=1}^n\sum_{k=1}^K \nabla f_i(x_i^{t,k}) \rangle + 2\gamma^2 \mathbb{E}\|\frac{1}{n}\sum_{i=1}^n\sum_{k=1}^K \nabla f_i(x_i^{t,k})\|^2 + \frac{2\gamma^2 K \sigma^2}{n}.
    \end{aligned}
\end{equation}
With L-smoothness and convexity, we have the inequality:
\begin{equation}
    \langle z - y, \nabla f(x)\rangle \geq f(z) - f(y) + \frac{\mu}{4} \|y - z\|^2 - L\|z - x\|^2,\quad \forall x,y,z \in R^d.
\end{equation}
We can now further bound the terms:
\begin{equation}
    \begin{aligned}
        \mathbb{E}\|\Bar{x}_g^{t+1} - x^\star\|^2 \leq& \mathbb{E}\|\Bar{x}_g^t - x^\star\|^2 + 2\gamma^2 \mathbb{E}\|\frac{1}{n}\sum_{i=1}^n\sum_{k=1}^K \nabla f_i(x_i^{t,k})\|^2 + \frac{2\gamma^2 K \sigma^2}{n}\\
        &- \frac{\mu\gamma}{2}\mathbb{E}\|\Bar{x}_g^t - x^\star\|^2 - 2\gamma K \mathbb{E} (f(\Bar{x}_g^t) - f(x^\star)) + 2\gamma L \Delta_t\\
        \leq& \mathbb{E}\|\Bar{x}_g^t - x^\star\|^2 + 4\gamma^2K^2 \mathbb{E}\|\nabla f (\Bar{x}_g^t)\|^2\\
        &- \frac{\mu\gamma}{2}\mathbb{E}\|\Bar{x}_g^t - x^\star\|^2 - 2\gamma K \mathbb{E} (f(\Bar{x}_g^t) - f(x^\star)) + 2\gamma L (1 + 2\gamma KL) \Delta_t+ \frac{2\gamma^2 K^2 \sigma^2}{n}.
    \end{aligned}
\end{equation}
We first choose $\gamma < \frac{1}{4KL}$ to let $1 + 2\gamma KL \leq \frac{3}{2}$, to control the coefficients of $\Delta_t$. Then, since we are assuming convexity for this lemma, we can use $\mathbb{E}\|\nabla f (\Bar{x}_g^t)\|^2 \leq 2L \mathbb{E}(f (\Bar{x}_g^t) - f (x^\star))$ to remove the $\mathbb{E}\|\nabla f (\Bar{x}_g^t)\|^2$ terms. Finally, we use lemma~\ref{lem1} to replace $\Delta_t$ with $\Gamma_t, \mathcal{Y}_t$ and $\mathcal{Z}_t$:
\begin{equation}
    \begin{aligned}
        \mathbb{E}\|\Bar{x}_g^{t+1} - x^\star\|^2 
        \leq& \mathbb{E}\|\Bar{x}_g^t - x^\star\|^2 + 4\gamma^2K^2 \mathbb{E}\|\nabla f (\Bar{x}_g^t)\|^2\\
        &- \frac{\mu\gamma}{2}\mathbb{E}\|\Bar{x}_g^t - x^\star\|^2 - 2\gamma K \mathbb{E} (f(\Bar{x}_g^t) - f(x^\star)) + 3\gamma L \Delta_t + \frac{2\gamma^2 K \sigma^2}{n}\\
        \leq & \mathbb{E}\|\Bar{x}_g^t - x^\star\|^2 - \frac{\mu\gamma}{2}\mathbb{E}\|\Bar{x}_g^t - x^\star\|^2 + (- 2\gamma K + 8\gamma^2 K^2 L + 36K^3\gamma^3 L^2) \mathbb{E} (f(\Bar{x}_g^t) - f(x^\star))\\
        &+9 (1-p)\gamma KL \Gamma_t + 72K^3L\gamma^3 (\mathcal{Y}_t + \mathcal{Z}_t)\\
        &+ \frac{2\gamma^2 K \sigma^2}{n} + 9K^2\gamma^3L\sigma^2.
    \end{aligned}
\end{equation}
In order to control the coefficients of $\mathbb{E} (f(\Bar{x}_g^t) - f(x^\star))$, we let $\gamma \leq \frac{1}{36KL}$ so that $ - 2\gamma K + 8\gamma^2 K^2 L + 36K^3\gamma^3 L^2 \leq - \gamma K$. Under this assumption, we arrive at the result:

\begin{equation}
    \begin{aligned}
        \mathbb{E}\|\Bar{x}_g^{t+1} - x^\star\|^2 
        \leq & \mathbb{E}\|\Bar{x}_g^t - x^\star\|^2 - \frac{\mu\gamma}{2}\mathbb{E}\|\Bar{x}_g^t - x^\star\|^2 - \gamma K  \mathbb{E} (f(\Bar{x}_g^t) - f(x^\star))\\
        &+9 (1-p)\gamma KL \Gamma_t + 72K^3L\gamma^3 (\mathcal{Y}_t + \mathcal{Z}_t)+ \frac{2\gamma^2 K \sigma^2}{n} + 9K\gamma^3L\sigma^2.
    \end{aligned}
\end{equation}

\end{proof}

\newpage
 
\section{Proofs of Theorems and Corollaries}
\label{appendixB}
\noindent In this section, we provide the proofs for our theorems and corollaries. 

\subsection{\textbf{Proof of Theorem \ref{thm1}}}

\begin{proof}
    We start with defining the Lyapunov function:
    \begin{equation}
        \mathcal{H}_t = \mathbb{E}f(\overline{x}_\mathrm{g} ^t) - \mathbb{E}f(x^\star) + c_0 K^3\gamma^3 \bigg(\frac{1}{p}\mathcal{Y}_t + \frac{1}{q}\mathcal{Z}_t\bigg) + c_1 \frac{K\gamma}{p}\Gamma_t
    \end{equation}
    for constants $c_0, c_1, c_2$. Additionally, by rearranging Lemma \ref{lem1} and introducing a constant $c_2$, we have:
    \begin{equation}
            0
        \leq c_2\gamma L^2 \bigg( -\Delta_t +  \textstyle3(1 - p)K\Gamma_t
        + 24K^3\gamma^2\mathcal{Y}_t   \textstyle+ 24K^3\gamma^2\mathcal{Z}_t + 6K^3\gamma^2\mathbb{E}\|\nabla f(\overline{x}_\mathrm{g} ^t)\|^2 + 3K^2\gamma^2 \sigma^2\bigg).
    \end{equation}
    Using Lemma \ref{lem5}, the non-convex descent lemma, together with Lemmas \ref{lem1}, \ref{lem2}, \ref{lem3}, \ref{lem4}, we can see that:
    \begin{equation}
\begin{aligned}
    \mathcal{H}_t - \mathcal{H}_{t-1} \leq & -DK\gamma\mathbb{E}\|\nabla f(\overline{x}_\mathrm{g} ^{t-1})\|^2
        -D_1\mathcal{Y}_{t-1} - D_2\mathcal{Z}_{t-1}- D_3\Gamma_{t-1} - D_4 \Delta_{t-1}\textstyle+\frac{D_5L^2}{p^4q^2 \cdot K}(K^3\gamma^3)\sigma^2 + \frac{L}{2nK}(K^2\gamma^2)\sigma^2,\\
\end{aligned}
\end{equation}
where here we have:
\begin{equation}
    \begin{cases}
        D &= \frac{1}{4} - \frac{12}{p^2}\gamma^4K^4L^2c_0 - \frac{12}{p^2}\gamma^2K^2c_1 - \frac{12}{q^2}\gamma^4 K^4 L^2c_0 - 6K^2\gamma^2L^2 c_2\\
        D_1 &= (\frac{c_0}{2} - \frac{12}{p}c_1 + 24L^2 c_2)K^3\gamma^3\\
        D_2 &= (\frac{c_0}{2} - \frac{12}{q}c_1 + 24L^2 c_2)K^3\gamma^3\\
        D_3 &= (\frac{c_1}{2} - 3(1-p)L^2c_2) K\gamma\\
        D_4 &= (c_2 - \frac{12}{p^2}\gamma^2 K^2c_1)\gamma L^2
    \end{cases}
\end{equation}
and $D_5$ is another constant. To get convergence, we need $D > 0$ and $D_1, D_2, D_3, D_4 \geq 0$. By choosing:
\begin{equation}
    \begin{cases}
        c_2 &> 2\\
        c_1 &= 6L^2 c_2\\
        c_0 &=\frac{192}{pq}L^2 c_2\\
        \gamma &\leq \frac{p^2 q}{945 KL}
    \end{cases}
\end{equation}
we can achieve this, yielding:
\begin{equation}
\begin{aligned}
    &\mathcal{H}_t - \mathcal{H}_{t-1} \leq  -DK\gamma\mathbb{E}\|\nabla f(\overline{x}_\mathrm{g} ^{t-1})\|^2
        +\frac{D_5L^2}{p^4q^2 \cdot K}(K^3\gamma^3)\sigma^2 + \frac{L}{2nK}(K^2\gamma^2)\sigma^2\\
    \Rightarrow & \frac{1}{T}(\mathcal{H}_{T} - \mathcal{H}_{1}) \leq -DK\gamma \frac{1}{T} \sum_{t=1}^T \mathbb{E}\|\nabla f(\overline{x}_g^t)\|^2 +\frac{D_5L^2}{p^4q^2 \cdot K}(K^3\gamma^3)\sigma^2 + \frac{L}{2nK}(K^2\gamma^2)\sigma^2\\
    \Rightarrow & \frac{1}{T} \sum_{t=1}^T \mathbb{E}\|\nabla f(\overline{x}_g^t)\|^2\leq \frac{1}{DTK\gamma}(\mathcal{H}_{1}) +\frac{D_5L^2}{Dp^4q^2 \cdot K}(K^2\gamma^2)\sigma^2 + \frac{L}{2DnK}(K\gamma)\sigma^2\\
\end{aligned}
\end{equation}
Based on our initialization, we have $\mathcal{H}_1 = \mathbb{E}f(\overline{x}_\mathrm{g} ^1) - \mathbb{E}f(x^\star) + c_0 K^3\gamma^3 \bigg(\frac{1}{p}\mathcal{Y}_1 + \frac{1}{q}\mathcal{Z}_1\bigg) + c_1 \frac{K\gamma}{p}\Gamma_1 \leq \mathbb{E}f(\overline{x}_\mathrm{g} ^1) - \mathbb{E}f(x^\star) + c_0 K^3\gamma^3 \bigg(\frac{1}{p}\sigma^2 + \frac{1}{q}\sigma^2\bigg)$. Under the condition that $T > K$, we can move the stochastic variance $\sigma^2$ from $\mathcal{H}_1$ and merge it into $D_5$, thus yielding the final result:
\begin{equation}
    \frac{1}{T} \sum_{t=1}^T \mathbb{E}\|\nabla f(\overline{x}_g^t)\|^2\leq \frac{1}{DTK\gamma}(\mathbb{E}f(\overline{x}_\mathrm{g} ^1) - \mathbb{E}f(x^\star)) +\frac{D_5L^2}{Dp^4q^2 \cdot K}(K^2\gamma^2)\sigma^2 + \frac{L}{2DnK}(K\gamma)\sigma^2.
\end{equation}
\end{proof}

\subsection{\textbf{Proof of Theorem \ref{thm2}}}

\begin{proof}
The descent lemma for this case is Lemma \ref{lem6}. Since we are assuming $\mu = 0$, we aim to use the $- \gamma K  \mathbb{E} (f(\Bar{x}_g^t) - f(x^\star))$ term in the RHS of the inequality to control the descent instead of the $- \frac{\mu\gamma}{4}\mathbb{E}\|\Bar{x}_g^t - x^\star\|^2$ term. From Lemma \ref{lem6}, we can write:
\begin{equation}
    \begin{aligned}
    \mathbb{E}\|\Bar{x}_g^{t+1} - x^\star\|^2 
        \leq & \mathbb{E}\|\Bar{x}_g^t - x^\star\|^2 - \frac{\mu\gamma}{2}\mathbb{E}\|\Bar{x}_g^t - x^\star\|^2 - \gamma K  \mathbb{E} (f(\Bar{x}_g^t) - f(x^\star))\\
        &+9 (1-p)\gamma KL \Gamma_t + 72K^3L\gamma^3 (\mathcal{Y}_t + \mathcal{Z}_t)+ \frac{2\gamma^2 K \sigma^2}{n} + 9K^2\gamma^3L\sigma^2\\
\Rightarrow \mathbb{E} (f(\Bar{x}_g^t) - f(x^\star)) \leq & \frac{1}{\gamma K}\Big(\mathbb{E}\|\Bar{x}_g^{t} - x^\star\|^2 - \mathbb{E}\|\Bar{x}_g^{t+1} - x^\star\|^2\Big)+9 (1-p)L \Gamma_t + 72K^2L\gamma^2 (\mathcal{Y}_t + \mathcal{Z}_t)+ \frac{2\gamma \sigma^2}{n} + 9K\gamma^2L\sigma^2\\
\end{aligned}
\end{equation}
From this, we define the following two terms:
\begin{equation}
\begin{aligned}
        \mathcal{G}_t &= 9 (1-p)L \Gamma_t + 72K^2L\gamma^2c_0 (\mathcal{Y}_t + \mathcal{Z}_t),\\
    \mathcal{E}_t &= \mathbb{E} (f(\Bar{x}_g^t) - f(x^\star)),
\end{aligned}
\end{equation}
where we choose $c_0 = \frac{6}{\min(p,q)^2} > 1$. For simplicity of notation, let $p' = \min(p,q)$. Then, we can bound the sum of $\mathcal{E}_t$ with:
\begin{equation}
    \label{eq:convT2}
    \frac{1}{T}\sum_{t=0}^{T-1}\mathcal{E}_t \leq \frac{\mathbb{E}\|\overline{x}_g^0 - x^\star\|^2}{\gamma KT} + \frac{1}{T} \sum_{t=0}^{T-1} \mathcal{G}_t + \frac{2\gamma \sigma^2}{n} + 9K\gamma^2 L \sigma^2.
\end{equation}
Now, in order to unfold this recursion, we need to bound $\mathcal{G}_t$ with an appropriate choice of step size. To do so, we first expand the terms using Lemmas \ref{lem2}, \ref{lem3}, \ref{lem4}:
\begin{equation}
    \begin{aligned}
        \mathcal{G}_t \leq & 9L(1-p)\bigg(1 - \frac{p}{2}\bigg) \Gamma_{t-1}\\
        &+\frac{432}{p'^2}K^2L\gamma^2 \bigg(\frac{p'}{4} + 1 - \frac{p'}{2}\bigg)(\mathcal{Y}_{t-1} + \mathcal{Z}_{t-1})\\
        &+ \bigg(108\frac{1-p}{p}\gamma^2 K^2L + \frac{24K^2L^2\gamma^2}{p'}\frac{432}{p'^2}K^2L\gamma^2\bigg)\mathbb{E}\|\nabla f (\Bar{x}_g^{t-1})\|^2\\
        &+ \bigg(108\frac{1-p}{p}\gamma^2 K L^3 + \frac{52L^2}{p'}\frac{432}{p'^2}K^2L\gamma^2\bigg)\Delta_{t-1}\\
        &+9L(1-p)K\gamma^2 \sigma^2 + \frac{4}{p'} \frac{432}{p'^2}K L\gamma^2 \sigma^2.
    \end{aligned}
\end{equation}
Next, we merge the coefficients of $\Delta_{t-1}$ using the fact $0 < p' \leq p < 1$ and $K \geq 1$ for a simpler notation:
\begin{equation}
    \begin{aligned}
        108\frac{1-p}{p}\gamma^2 K L^3 + \frac{52L^2}{p'}\frac{432}{p'^2}K^2L\gamma^2 &\leq (\frac{108}{p} + \frac{22464}{p'^3})\gamma^2K^2L^3 \leq \frac{22572}{p'^3} \gamma^2K^2L^3.
    \end{aligned}
\end{equation}
Then we expand the $\Delta_{t-1}$ term:
\begin{equation}
    \begin{aligned}
        \mathcal{G}_t \leq & 9L(1-p)\bigg(1 - \frac{p}{2} + \frac{7524}{p'^3}\gamma^2 K^3 L^2\bigg) \Gamma_{t-1}\\
        &+\frac{432}{p'^2}K^2L\gamma^2 \bigg(\frac{p'}{4} + 1 - \frac{p'}{2} + \frac{1254}{p'}L^2 K^3 \gamma^2\bigg)(\mathcal{Y}_{t-1} + \mathcal{Z}_{t-1})\\
        &+ \bigg(\frac{108}{p'^3}\gamma^2 K^2L + \frac{24K^2L^2\gamma^2}{p'}\frac{432}{p'^2}K^2L\gamma^2 + \frac{22572}{p'^3}\gamma^2 K^2 L^3 6 K^3 \gamma^2\bigg)\mathbb{E}\|\nabla f (\Bar{x}_g^{t-1})\|^2\\
        &+9L(1-p)K\gamma^2 \sigma^2 + \frac{4}{p'} \frac{432}{p'^2}K L\gamma^2 \sigma^2 + \frac{67716}{p'^3} \gamma^4 K^4 L^3 \sigma^2
    \end{aligned}
\end{equation}
In order to control all these coefficients, if we impose the following constraints:
\begin{equation}
    \begin{cases}
        \gamma \leq \frac{p'^2}{173 \sqrt{K^3 L^2}}\\
        \gamma \leq \frac{p'}{101\sqrt{K^3 L^2}}\\
        \gamma \leq \frac{p'}{53 \sqrt{K^3 L^2}}\\
        \gamma \leq \frac{1}{87\sqrt{K^3 L^2}}
    \end{cases}
\end{equation}
Then the following holds:
\begin{equation}
    \begin{cases}
        \frac{7524}{p'^3}\gamma^2 K^3 L^2 &\leq \frac{p'}{4}\\
        \frac{1254}{p'}L^2 K^3 \gamma^2 &\leq \frac{p'}{8}\\
        \frac{108}{p'^3}\gamma^2 K^2L + \frac{24K^2L^2\gamma^2}{p'}\frac{432}{p'^2}K^2L\gamma^2 + \frac{22572}{p'^3}\gamma^2 K^2 L^3 6 K^3 \gamma^2 &\leq \frac{120L}{p'}\gamma^2 K^2\\
        9L(1-p)K\gamma^2 \sigma^2 + \frac{4}{p'} \frac{432}{p'^2}K L\gamma^2 \sigma^2 + \frac{67716}{p'^3} \gamma^4 K^4 L^3 \sigma^2 &\leq \frac{1740}{p'^3} \gamma^2 KL\sigma^2
    \end{cases}
\end{equation}
Thus, we have:
\begin{equation}
\begin{aligned}
        \mathcal{G}_t &\leq (1 - \frac{p'}{8}) \mathcal{G}_{t-1} + \frac{120L}{p'}\gamma^2 K^2\mathbb{E}\|\nabla f (\Bar{x}_g^{t-1})\|^2 + \frac{1740}{p'^3} \gamma^2 KL\sigma^2\\
        &\leq (1 - \frac{p'}{8}) \mathcal{G}_{t-1} + \frac{240L}{p'}\gamma^2 K^2L \mathbb{E} (f(\Bar{x}_g^{t-1}) - f(x^\star)) + \frac{1740}{p'^3} \gamma^2 KL\sigma^2,
\end{aligned}    
\end{equation}
where in the second inequality we use $\mathbb{E}\|\nabla f (\Bar{x}_g^{t-1})\|^2 \leq 2L \mathbb{E} (f(\Bar{x}_g^{t-1}) - f(x^\star))$.
Now we can unroll the recursion of $\mathcal{G}_t$, defining $\delta = 1 - \frac{p'}{8}$:
\begin{equation}
\begin{aligned}
        \frac{1}{T}\sum_{t=2}^{T}\mathcal{G}_t &\leq \frac{2 \mathcal{G}_1}{(1- \delta)T} + \frac{240}{p' T} \gamma^2 K^2 L \sum_{t=1}^{T} \sum_{l = 0}^{t-1} (\frac{1+\delta}{2})^{t-1-l} \mathcal{E}^l + \frac{3480 \gamma^2 KL \sigma^2}{p'^3 (1 - \delta)}\\
        &\leq \frac{2 \mathcal{G}_1}{(1- \delta)T} + \frac{240}{p'} \frac{\gamma^2 K^2 L}{T (1 - \delta)} \sum_{t=1}^{T} \mathcal{E}_t + \frac{3480 \gamma^2 KL \sigma^2}{p'^3 (1 - \delta)}\\
        \Rightarrow \frac{1}{T}\sum_{t=1}^{T}\mathcal{G}_t 
        &\leq \frac{3 \mathcal{G}_1}{(1- \delta)T} + \frac{240}{p'} \frac{\gamma^2 K^2 L}{T (1 - \delta)} \sum_{t=1}^{T} \mathcal{E}_t + \frac{3480 \gamma^2 KL \sigma^2}{p'^3 (1 - \delta)}\\
\end{aligned}    
\end{equation}
Finally, we plug this back into the relationship from~\eqref{eq:convT2}:
\begin{equation}
    \begin{aligned}
        \frac{1}{T}\sum_{t=1}^{T}\mathcal{E}_t &\leq \frac{\mathbb{E}\|\overline{x}_g^1 - x^\star\|^2}{\gamma KT} + \frac{1}{T} \sum_{t=1}^{T} \mathcal{G}_t + \frac{2\gamma \sigma^2}{n} + 9K\gamma^2 L \sigma^2\\
        \Rightarrow \frac{1}{T} \left(1 - \frac{3840\gamma^2 K^2 L}{p'^2} \right)\sum_{t=1}^{T}\mathcal{E}_t &\leq \frac{\mathbb{E}\|\overline{x}_g^1 - x^\star\|^2}{\gamma KT} + \frac{3 \mathcal{G}_1}{(1- \delta)T} + \frac{2\gamma  \sigma^2}{n} + 9K\gamma^2 L \sigma^2 + \frac{55680\gamma^2 KL \sigma^2}{p'^4}\\
    \end{aligned}
\end{equation}
If we choose $\gamma \leq \frac{p'}{88 \sqrt{K^2L}}$, we get:
\begin{equation}
    \begin{aligned}
       \frac{1}{T}\sum_{t=1}^{T}\mathcal{E}_t &\leq \frac{2\mathbb{E}\|\overline{x}_g^1 - x^\star\|^2}{\gamma KT} + \frac{48\mathcal{G}_1}{T} + \frac{2\gamma  \sigma^2}{n} + 9K\gamma^2 L \sigma^2 + \frac{55680\gamma^2 KL \sigma^2}{p'^4}\\
       &\leq \frac{2\mathbb{E}\|\overline{x}_g^1 - x^\star\|^2}{\gamma KT} + \frac{432(1-p)L\Gamma_1}{T} + \frac{3456LK^2\gamma^2(\mathcal{Y}_1 + \mathcal{Z}_1)}{p^2q^2T} + \frac{2\gamma  \sigma^2}{n} + 9K\gamma^2 L \sigma^2 + \frac{55680\gamma^2 KL \sigma^2}{p^4 q^4}\\
    \end{aligned}
\end{equation}
Then, based on the algorithm's initialization, and under the condition that $T >K$, we can merge the terms $\mathcal{Y}_1$ and $\mathcal{Z}_1$ into the last stochastic variance term to obtain the final result:
\begin{equation}
    \begin{aligned}
       \frac{1}{T}\sum_{t=1}^{T}\mathcal{E}_t 
       &\leq \frac{2\mathbb{E}\|\overline{x}_g^1 - x^\star\|^2}{\gamma KT} + \frac{2\gamma  \sigma^2}{n} + 9K\gamma^2 L \sigma^2 + \frac{62592\gamma^2 KL \sigma^2}{p^4 q^4}. \\
    \end{aligned}
\end{equation}

\end{proof}
% \textcolor{red}{Organize the corollary into one}

\subsection{\textbf{Proof of Theorem \ref{thm3}}}

\begin{proof}
As in Theorem \ref{thm2}, the descent lemma for this case is Lemma \ref{lem6}. Since we are assuming $\mu > 0$, we aim to use the $- \frac{\mu\gamma}{4}\mathbb{E}\|\Bar{x}_g^t - x^\star\|^2$ term in the RHS of the inequality to control the descent instead of using the $- \gamma K  \mathbb{E} (f(\Bar{x}_g^t) - f(x^\star))$ term. We can write:
\begin{equation}
    \begin{aligned}
    \mathbb{E}\|\Bar{x}_g^{t+1} - x^\star\|^2 
        \leq & \mathbb{E}\|\Bar{x}_g^t - x^\star\|^2 - \frac{\mu\gamma}{2}\mathbb{E}\|\Bar{x}_g^t - x^\star\|^2 - \gamma K  \mathbb{E} (f(\Bar{x}_g^t) - f(x^\star))\\
        &+9 (1-p)\gamma KL \Gamma_t + 72K^3L\gamma^3 (\mathcal{Y}_t + \mathcal{Z}_t)+ \frac{2\gamma^2 K \sigma^2}{n} + 9K^2\gamma^3L\sigma^2\\
    \leq &  (1 - \frac{\mu\gamma}{2})\mathbb{E}\|\Bar{x}_g^t - x^\star\|^2
        +9 (1-p)\gamma KL \Gamma_t + 72K^3L\gamma^3 (\mathcal{Y}_t + \mathcal{Z}_t)+ \frac{2\gamma^2 K \sigma^2}{n} + 9K^2\gamma^3L\sigma^2.
    \end{aligned}
\end{equation}
Note that by convexity and L-smoothness, we can bound the gradient terms as follows:
\begin{equation}
    \mathbb{E}\|\nabla f(\overline{x}_g^t)\|^2 \leq L^2 \mathbb{E}\|\overline{x}_g^t - x^\star\|^2
\end{equation}
Then, by combining Lemmas \ref{lem1}, \ref{lem2}, \ref{lem3}, \ref{lem4}, and \ref{lem6}, we can form the following recursion:
\begin{equation}
        \textstyle\begin{bmatrix}
            \mathbb{E}\|\overline{x}_\mathrm{g} ^t \!-\! x^\star\|^2\\
            \Gamma_t\\
            \gamma \mathcal{Y}_t\\
            \gamma \mathcal{Z}_t
        \end{bmatrix}
        \leq         
        A\begin{bmatrix}
            \mathbb{E}\|\overline{x}_\mathrm{g} ^{t-1} \!-\! x^\star\|^2\\
            \Gamma_{t-1}\\
            \gamma \mathcal{Y}_{t-1}\\
            \gamma \mathcal{Z}_{t-1}
        \end{bmatrix} + b.
        \label{eq:recur-thm3}
    \end{equation}
where
\begin{equation}
    A = \begin{bmatrix}
        1 - \mu K\gamma/2 & 9\gamma KL(1 - p) & 72K^3L\gamma^2 & 72K^3L\gamma^2\\
        \frac{14}{p}\gamma^2K^2L^2 & 1 - \frac{p}{2} + \frac{36}{p}\gamma^2KL^2 & \frac{14}{p} K^2\gamma & \frac{14}{p} K^2\gamma\\
        \frac{72}{p}\gamma^2K^2L^3 & \frac{30}{p}L^2K\gamma & 1 - \frac{p}{2} + \frac{240}{p}K^3\gamma^2L^2 & \frac{240}{p}K^3\gamma^2L^2\\
        \frac{168}{q}\gamma^2K^2L^3 & \frac{78}{q}L^2K\gamma & \frac{624}{q}K^3\gamma^2L^2 & 1 - \frac{q}{2} + \frac{624}{q}K^3\gamma^2L^2 \\
    \end{bmatrix},
\end{equation}
and 
\begin{equation}
    b = \begin{bmatrix}
        \frac{2\gamma^2K}{n} + 9K^2\gamma^3L\\
        K\gamma^2 + 3K^3\gamma^4L^2\\
        \frac{2\gamma}{qK} + \frac{30K^3\gamma^3L^2}{q}\\
        \frac{2\gamma}{qK} + \frac{78K^3\gamma^3L^2}{q}\\
    \end{bmatrix}\sigma^2.
\end{equation}
To enforce convergence, we aim to upper bound the value of $\|A\|_1$ by appropriate choice of step size, which gives the following conditions:
\begin{equation}
    \begin{cases}
        \frac{254}{\min(p, q)}\Big(\gamma^2K^2L^2 + 2\gamma^2K^2L^3\Big) \leq \frac{\mu K\gamma}{4}\\
         9\gamma KL(1 - p) + 1 - \frac{p}{2} + \frac{36}{p}\gamma^2KL^2 + \frac{30}{p}L^2K\gamma + \frac{78}{q}L^2K\gamma \leq 1 - \frac{\mu K\gamma}{4}\\
          72K^3L\gamma^2 + \frac{14}{p} K^2\gamma + 1 - \frac{p}{2} + \frac{240}{p}K^3\gamma^2L^2 + \frac{624}{q}K^3\gamma^2L^2 \leq 1 - \frac{\mu K \gamma}{4}\\
           72K^3L\gamma^2+  \frac{14}{p} K^2\gamma+ \frac{240}{p}K^3\gamma^2L^2+ 1 - \frac{q}{2} + \frac{624}{q}K^3\gamma^2L^2 \leq 1 - \frac{\mu K \gamma}{4}
    \end{cases}
\end{equation}
Now we define the following Lyapunov function:
\begin{equation}
    \mathcal{L}_t = \mathbb{E}\|\overline{x}_\mathrm{g} ^t - x^\star\|^2 + \Gamma_t + \gamma \mathcal{Y}_t + \gamma \mathcal{Z}_t.
\end{equation}
If we choose our step size to satisfy the following inequality:
\begin{align}
    \gamma \leq \min\Big(\frac{\min(p,q)\mu}{K(14L^2+240L^3)},\frac{1}{18KL}, \frac{4}{\mu K},\frac{\min(p, q)p}{2(45KL + 108KL^2 + K\mu/4)},\frac{\min(p,q)^2}{2(86K^2+864K^2L+K\mu/4)}\Big),
    \label{eq:14}
\end{align}
then we have $\rho(A) \leq \|A\|_1 \leq 1 - \frac{\mu K\gamma}{4} < 1$. Unrolling the stochastic noise part of the recursion in~\eqref{eq:recur-thm3}, we have $\sum_{t=0}^{T-1}A^tb \leq (I - A)^{-1}b$. Therefore:
{\begin{align}
\mathcal{L}_{T+1} \leq (1 - \frac{\mu K\gamma}{4})^T\mathcal{L}_{1} + \|(I - A)^{-1}b\|_1.
\end{align}}
Lower bounding $\mathcal{L}_{T+1}$ with $\mathbb{E}\|x_\mathrm{g} ^{T+1} - x^\star\|^2$ completes the proof.
\end{proof}

\subsection{\textbf{Proof of Corollary \ref{cor1}}}
\begin{proof}
    For the non-convex case, we build on lemma \ref{unroll_lem}. Starting from the final result from Theorem \ref{thm1}, we can choose:
    \begin{equation}
        \begin{cases}
            r_0 &= \frac{\mathcal{E}_1}{D}\\
            b &= \frac{L\sigma^2}{2DnK}\\
            e &= \frac{D_5 L^2\sigma^2}{p^4q^2K}.
        \end{cases}
    \end{equation}
    Then, by choosing $K\gamma \leq \frac{1}{u} = \frac{1}{L}$, we have:
    \begin{equation}
        \begin{aligned}
            \frac{1}{T} \sum_{t=1}^T \mathbb{E}\|\nabla f(\overline{x}_g^t)\|^2 &\leq
            \frac{r_0}{T}\frac{1}{K\gamma} + bK\gamma + e K^2\gamma^2 \\
       &\leq 2\sqrt{\frac{br_0}{T}} + 2e^{\frac{1}{3}}(\frac{r_0}{T})^{\frac{2}{3}} + \frac{ur_0}{T}\\
       &= \mathcal{O}\bigg(\sqrt{\frac{\mathcal{H}_1\sigma^2L}{nTK}} + (\frac{\mathcal{H}_1L\sigma}{\sqrt{K}Tp^2q})^{\frac{2}{3}} + \frac{\mathcal{H}_1L}{Tp^2q}\bigg).
        \end{aligned}
    \end{equation}
\end{proof}

\subsection{\textbf{Proof of Corollary \ref{cor2}}}
\begin{proof}
For the weakly convex case, we still build on Lemma \ref{unroll_lem}. Starting from the final result from Theorem \ref{thm2}, by choosing:
\begin{equation}
    \begin{cases}
        r_0 &= 2\mathbb{E}\|\overline{x}_g^1 - x^\star\|^2\\
        b &= \frac{2\sigma^2}{Kn}\\
        e &= \frac{9L\sigma^2}{K} + \frac{62592L\sigma^2}{Kp^4q^4},
    \end{cases}
\end{equation}
and choosing $K\gamma \leq \frac{1}{u} = \frac{1}{L}$, we have:
\begin{equation}
    \begin{aligned}
        \frac{1}{T}\sum_{t=1}^{T}\mathcal{E}_t 
       &\leq \frac{2\mathbb{E}\|\overline{x}_g^1 - x^\star\|^2}{\gamma KT} + \frac{2\gamma  \sigma^2}{n} + 9K\gamma^2 L \sigma^2 + \frac{62592\gamma^2 KL \sigma^2}{p^4 q^4}\\
       &\leq \frac{r_0}{T}\frac{1}{K\gamma} + bK\gamma + e K^2\gamma^2\\
       &\leq 2\sqrt{\frac{br_0}{T}} + 2e^{\frac{1}{3}}(\frac{r_0}{T})^{\frac{2}{3}} + \frac{ur_0}{T}\\
       &= \mathcal{O}\bigg(\frac{L\mathbb{E}\|\overline{x}_g^1 - x^\star\|^2}{T}\bigg) + \mathcal{O}\bigg(\Big(\frac{\sqrt{L}\mathbb{E}\|\overline{x}_g^1 - x^\star\|^2\sigma}{p^2q^2\sqrt{K}T}\Big)^{\frac{2}{3}} + \sqrt{\frac{\mathbb{E}\|\overline{x}_g^1 - x^\star\|^2\sigma^2}{nKT}}\bigg).
    \end{aligned}
\end{equation}
\end{proof}

\subsection{\textbf{Proof of Corollary \ref{cor3}}}
\begin{proof}
    Starting from the final results of Theorem~\ref{thm3}, we find see that:
{\begin{align}
    \mathbb{E}\|\overline{x}_\mathrm{g} ^{T+1} - x^\star\|^2 &\leq (1 - \frac{\mu K\gamma}{4})^T\mathcal{L}_{1} + \|(I - A)^{-1}b\|_1\notag\\
    &\leq(1 - \frac{\mu K\gamma}{4})^T\mathcal{L}_{1} + \mathcal{O}(\frac{K^5L^5\sigma^2\gamma^5}{pq})\notag\\
    &\leq \exp( - \frac{\mu K\gamma}{2}T)\mathcal{L}_1 + \mathcal{O}\left(\frac{K^5L^5\sigma^2\gamma^5}{pq}\right).
\end{align}}
Defining $\overline{\gamma}$ as
\begin{equation}
    \overline{\gamma} = \min\Big(\frac{\min(p,q)\mu}{K(14L^2+240L^3)},\frac{1}{18KL}, \frac{4}{\mu K},\frac{\min(p, q)p}{2(45KL + 108KL^2 + K\mu/4)},\frac{\min(p,q)^2}{2(86K^2+864K^2L+K\mu/4)}\Big),
\end{equation}
we can choose the step size according to:
\begin{equation}
    \gamma = \min \left(\overline{\gamma}, \; \frac{\ln(\max(1, \mu K(\mathbb{E}\|\overline{x}_\mathrm{g} ^{1} - x^\star\|^2 + \Bar{\gamma} 2\sigma^2)T/\sigma^2))}{\mu K T} \right).
\end{equation} 
Using this step size, we see that:
{\begin{align}
        \mathbb{E}\|\overline{x}_\mathrm{g} ^{T+1} - x^\star\|^2 \leq& \mathcal{O}\bigg(\exp( - \frac{\mu K\overline{\gamma}}{2}T)\mathbb{E}\|\overline{x}_\mathrm{g} ^{1} - x^\star\|^2\bigg)+ \mathcal{O}\bigg(\exp( - \frac{\mu K\overline{\gamma}}{2}T)\overline{\gamma}\sigma^2\bigg) + \mathcal{O}\left(\frac{\sigma^2}{\mu KT}\right) + \Tilde{\mathcal{O}}\left(\frac{L^5\sigma^2}{\mu^5T^5pq}\right).
\end{align}}
\noindent Finally, we plug in $\overline{\gamma}$ to complete the proof.
\end{proof}
% \section*{Appendix C\\Additional Numerical results}

\newpage
\section{Additional Experiments}
\label{appen:additional_experiments}

\subsection{Experiments on weakly convex problems}

\textcolor{black}{In Fig.~\ref{fig:cvx}, we show an experiment for weakly convex loss functions (Theorem~\ref{thm2}). This follows the synthetic data generation process from Fig.~\ref{fig:scvx}, but we change the data dimension $d$ from 200 to 1000, while keeping the number of samples fixed. This changes the problem from strongly-convex to convex, as the Gram matrix becomes positive semi-definite.}

\begin{figure}[h]
\captionsetup{justification=centering}
    \centering
    \includegraphics[width=0.85\linewidth]{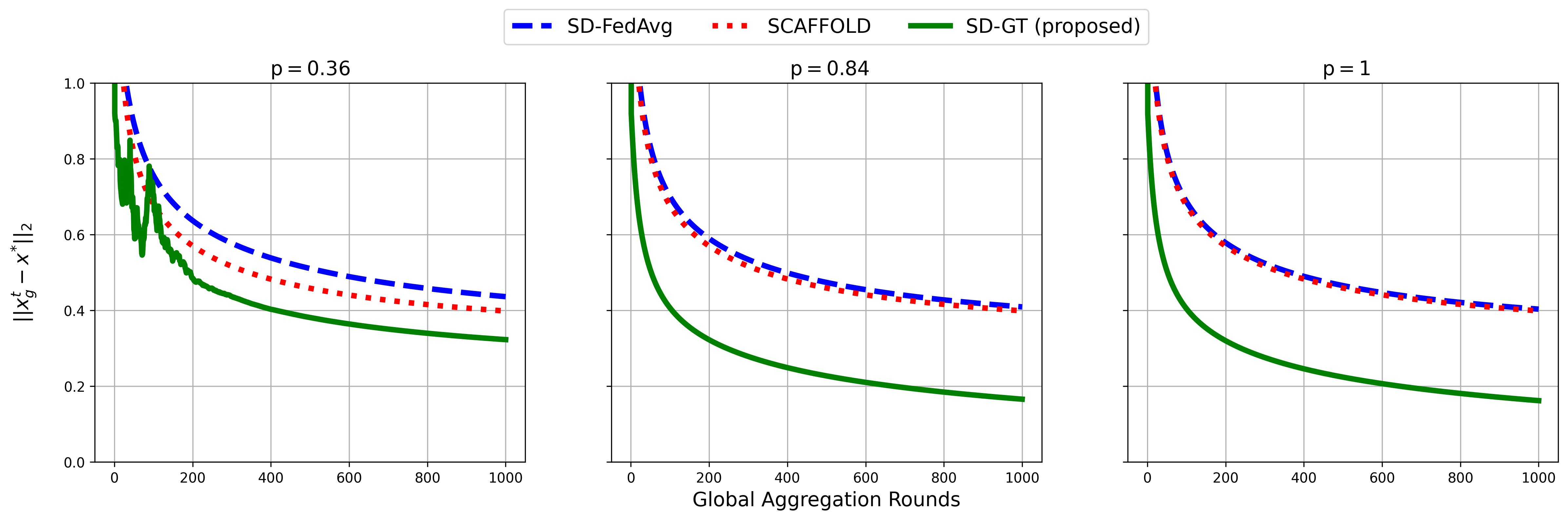}
    \caption{\textcolor{black}{Experiments on a weakly convex loss function, for different subnet sampling rates $p$ (obtained from $\beta_s = 0.2, 0.6, 1$).}}
    \label{fig:cvx}
\end{figure}

\textcolor{black}{As shown in Fig.~\ref{fig:cvx}, the algorithm continues to converge towards the optimal solution in the weekly convex case, although at a sublinear rate. This is consistent with theoretical guarantee for weakly convex problems in Theorem~\ref{thm2}, versus the linear rate that can be achieved in strongly convex problems (Theorem~\ref{thm3} and Fig.~\ref{fig:scvx}) in the manuscript). Despite this slower convergence compared to strongly convex tasks, our method still consistently outperforms SD-FedAvg and SCAFFOLD, highlighting its benefit across a range of loss geometries.}

\subsection{Experiments on strongly-convex with noises ($\sigma^2 > 0$)}
\textcolor{black}{In Fig.~\ref{fig:s-cvx_w_noise}, we evaluate the training performance for the strongly convex synthetic dataset, in the case where we have stochastic noises on the gradients. Based on the convergence analysis in Theorem~\ref{thm3} and Corollary~\ref{cor3}, we know that the linear rate will only be achievable for strongly convex problems with deterministic gradients ($\sigma^2 = 0$); on the other hand, in the presence of stochastic noises, we expect the convergence will be eventually be dominated by sublinear rates as we approach a steady state. This aligns with the observation in Fig.~\ref{fig:s-cvx_w_noise}, where each amount of stochasticity prevents linear convergence, and the final error obtained is increasing in the value of $\sigma^2$.}
% \newpage

\begin{figure}[h]
\captionsetup{justification=centering}
    \centering
    \includegraphics[width=0.85\linewidth]{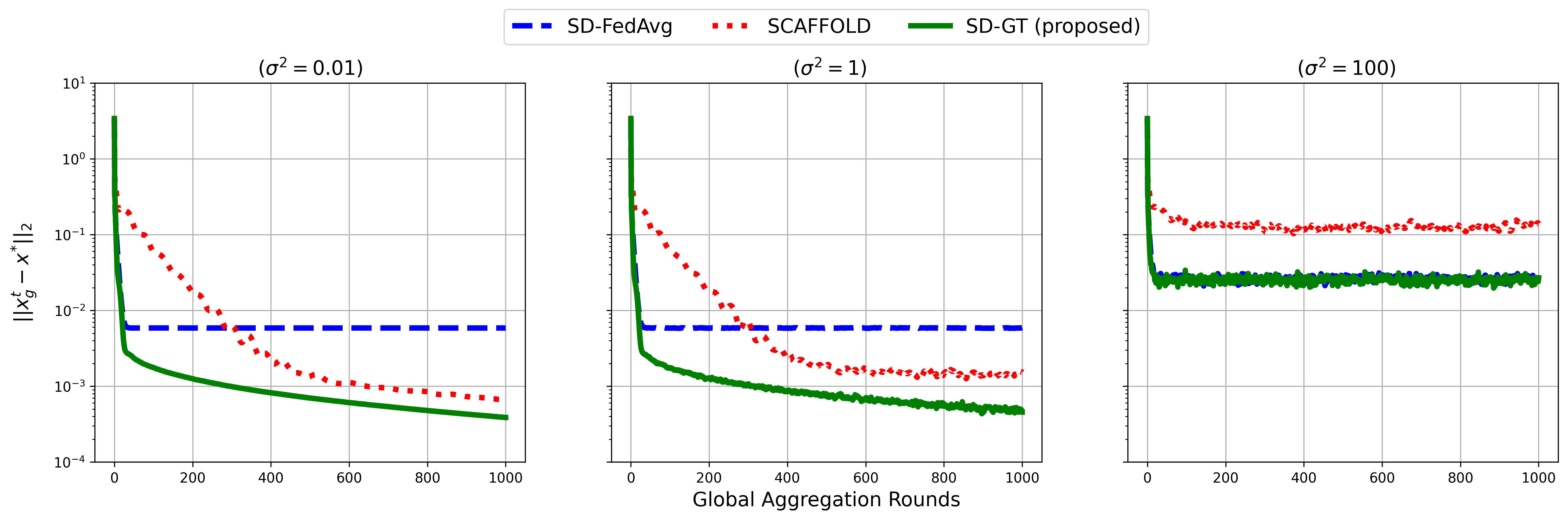}
    \caption{\textcolor{black}{Experiments on the strongly convex problem with gradient noises $\sigma^2 = 0.01, 1, 100$.}}
    \label{fig:s-cvx_w_noise}
\end{figure}

\subsection{Comparison with control algorithm from \cite{chen2024taming}}
\textcolor{black}{
In Fig.~\ref{fig:ctrl-comp}, we compare the adaptive control algorithm in this paper (Algorithm~\ref{alg:2}) with the control algorithm from our initial conference version~\cite{chen2024taming}. The experimental setup is the same as in Sec.~\ref{sec:V-E}. Unlike the previous control method in~\cite{chen2024taming}, Algorithm~\ref{alg:2} updates hyperparameters dynamically during each communication round $t$, resulting more precise control on the energy efficiency. As shown in Fig.~\ref{fig:ctrl-comp}, when the the cost of D2D and DS communication differs by a lot ($\delta = 1\times 10^{-2}$), the new proposed method is able to outperform the prior method. When D2D is energy efficient, there is more room to adapt the number of local training rounds $K$ in-between global aggregations to reduce drift without driving up the energy cost.
}

\begin{figure}[h]
\captionsetup{justification=centering}
    \centering
    \includegraphics[width=1.0\linewidth]{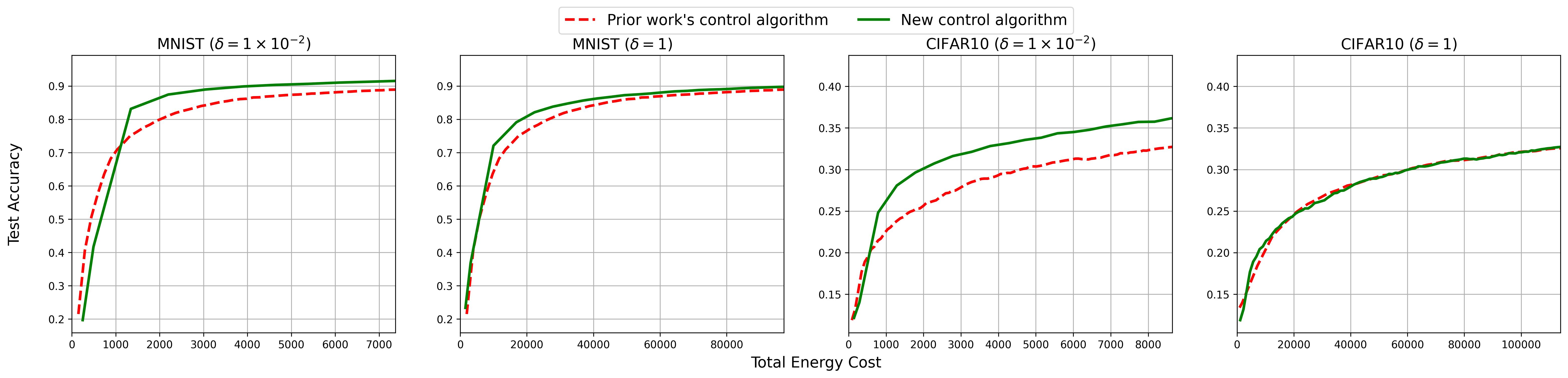}
    \caption{\textcolor{black}{Comparison of Algorithm~\ref{alg:2} with the control algorithm from prior work~\cite{chen2024taming}.}}
    \label{fig:ctrl-comp}
\end{figure}

{\color{black}
\subsection{Impact of Number of Subnets under CIFAR-100}

As shown in Fig.~\ref{fig:CIFAR100_exp2}, when the network consists of a large number of small subnets (i.e., the left plots), {\tt SD-GT} training relies more heavily on global aggregations and between-subnet gradient tracking terms, $y_i^t$. Conversely, when the network contains a small number of large subnets (i.e., the right plots), {\tt SD-GT} depends more on D2D communications and the impact of within-subnet gradient tracking, $z_i^t$. An important observation, now evaluated under the CIFAR-100 dataset, is that the performance of SD-FedAvg is highly sensitive to the specific subnet groupings. Furthermore, {\tt SD-GT} consistently outperforms both SCAFFOLD and SD-FedAvg across all scenarios. \textit{These results underscore the critical role of correcting client drift with both within-subnet ($z_i^t$) and between-subnet ($y_i^t$) gradient tracking, as alternative approaches like SD-FedAvg, which rely solely on D2D communication and local updates, are inadequate when data heterogeneity within and between sububnets varies.}
}

\begin{figure}[h]
% \centerline{\includegraphics[width=\textwidth]{CIFAR10/setting1_CIFAR10-09-07-22-16-19.jpg}}
\centerline{\includegraphics[width=\textwidth]{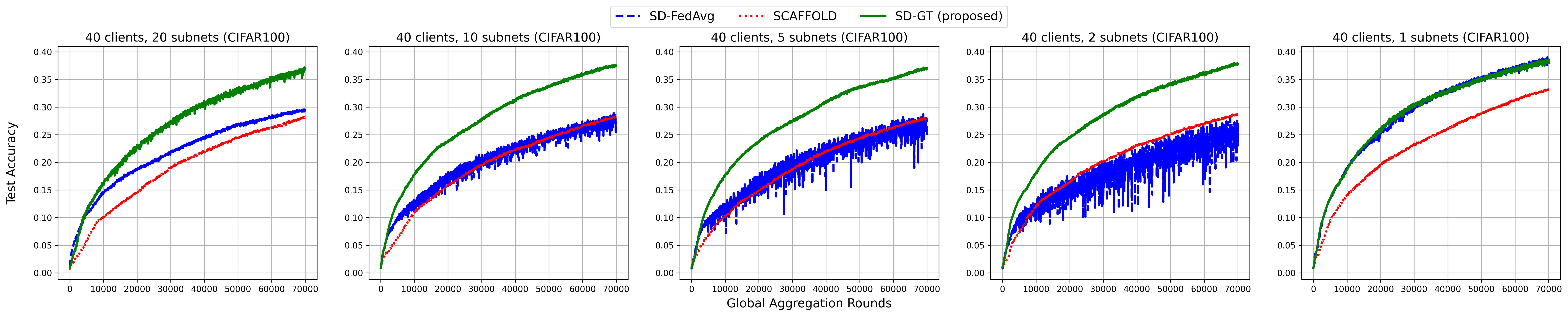}}
\caption{Impact of changing the number of subnets $S$ with fixed total number of clients $n = 40$ ($K = 3$). 
Under the same settings as in Figure~\ref{fig:CIFAR10_exp2} but evaluated under CIFAR100, {\tt SD-GT} is able to outperform both baselines using the combination of between-subnet and within-subnet gradient tracking on top of D2D communications.\vspace{-0.15in}}
\label{fig:CIFAR100_exp2}
\end{figure}

{\color{black}
\subsection{Impact of Device-Server Communication Frequency on MNIST}
In Figure \ref{fig:MNIST_exp1}, we compare the training performance of the algorithms on MNIST as the number of D2D rounds and local updates between global aggregations, denoted by $K$, is varied. As in the CIFAR10 experiments, the central server samples (40\%) of the clients from each subnet, and we consider $K = 3$ to $K = 15$ to study the impact of increasing the frequency of in-subnet consensus operations while reducing the frequency of global aggregations.

We observe that SCAFFOLD converges more slowly than {\tt SD-GT} across all values of $K$. Although both methods rely on gradient correction mechanisms, SCAFFOLD does not utilize D2D communications between global aggregations, whereas {\tt SD-GT} explicitly exploits low-cost in-subnet D2D exchanges to refine local updates. This advantage becomes more pronounced as $K$ increases, where frequent in-subnet communication allows {\tt SD-GT} to better align client updates before global synchronization.
\begin{figure}[h]
\centerline{\includegraphics[width=\textwidth]{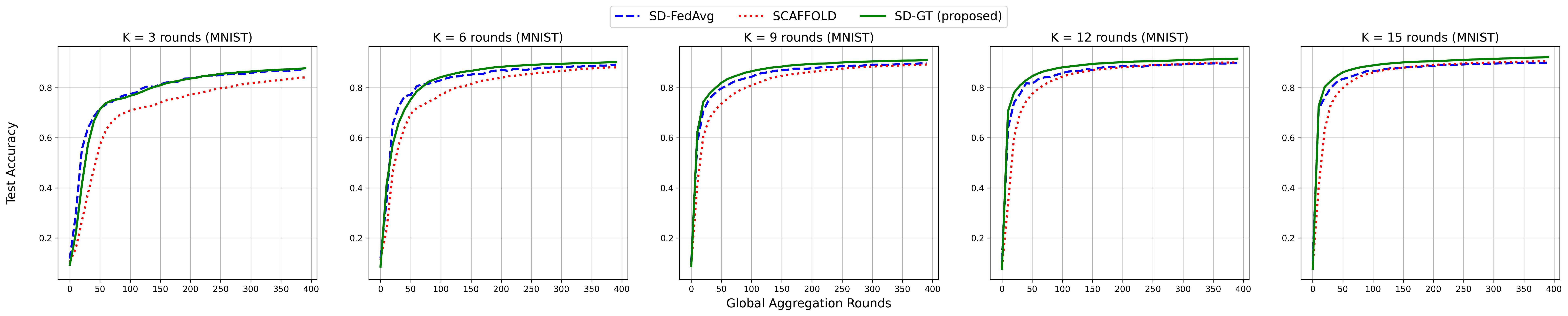}}
% \centerline{\includegraphics[width=\textwidth]{CIFAR10/setting3_CIFAR10-09-07-22-05-14.jpg}}
\caption{Comparison between algorithms on MNIST datasets when changing the number of local client updates and D2D consensus rounds $K$ between global aggregations. Each experiment is conducted with 30 clients and 3 subnets. As $K$ increases, {\tt SD-GT} is able to take advantage of multiple in-subnet model update and consensus iterations while correcting for client drift to achieve better convergence speed.\vspace{-0.15in}}
\label{fig:MNIST_exp1}
\end{figure}

}

\vfill

\newpage

\section{Discussion and Experiments on Time-Varying Topologies}
\label{appen:discuss_topology}
\textcolor{black}{In this section, we briefly discuss extensions to time-varying topologies, e.g., due to device mobility or variable D2D link bandwidth. We consider two cases of dynamics, differing based on the behavior of devices across different subnets:
\begin{enumerate}
    \item \textbf{Topology changes with fixed subnet devices and strongly-connected networks:} This includes scenarios where devices move or experience fluctuating bandwidth, resulting in link reconfigurations but without altering subnet membership. As long as the subnet graphs remain strongly connected at each time step, our convergence analysis can be extended to handle this case by introducing a time-dependent connectivity parameter. Specifically, setting $\rho_s = \min_t \{\rho_s^t\}$, where $\rho_s^t$ is the connectivity constant at global time $t$, preserves the descent inequality in Lemma~\ref{lem3} by updating the subnet connectivity parameter $q$ in Eq.~\eqref{eq:isgt}.
    We conducted additional experiments to validate the impact of these topology changes on model training convergence in {\tt SD-GT}. We begin by setting the mixing matrices $W_s^t$ at the initial time to correspond to random geometric graph topologies with radius selected in $r \in [1.5,5.5]$. Then, we consider three cases of evolution. \textit{(1) Static topology:} Each subnet has fixed $W_s$, $\rho_s$ throughout the whole training phase. \textit{(2) Fully dynamic topology:} Each subnet generates a new graph topology and $W_s^t$ for each global time $t$. \textit{(3) Single subnet shift:} At each global time $t$, one subnet chosen at random changes its topology and $W_s^t$, and the others stay the same from the previous $t$. The results are shown in Fig.~\ref{fig:time_varying}. We can see that our algorithm is still able to converge in all settings. As expected, the static topology performs the best, since its convergence speed will not be dragged down by the worst connectivity $\min_t \{\rho_s^t\}$ experienced for each subnet throughout the whole training process.
    \item \textbf{Topology changing with devices moving across subnets:} If devices migrate between subnets (e.g., moving between base stations), the system may no longer maintain doubly stochastic mixing matrices or track each node's subnet reliably. This results in a violation of Assumption~\ref{asmp2} on strongly connected topologies, and requires a new framework for robust optimization under frequent link interruptions and device mobility. We leave this as a promising direction for future work.
\end{enumerate}}

\begin{figure}[h]
\captionsetup{justification=centering}
    \centering
    \includegraphics[width=0.6\linewidth]{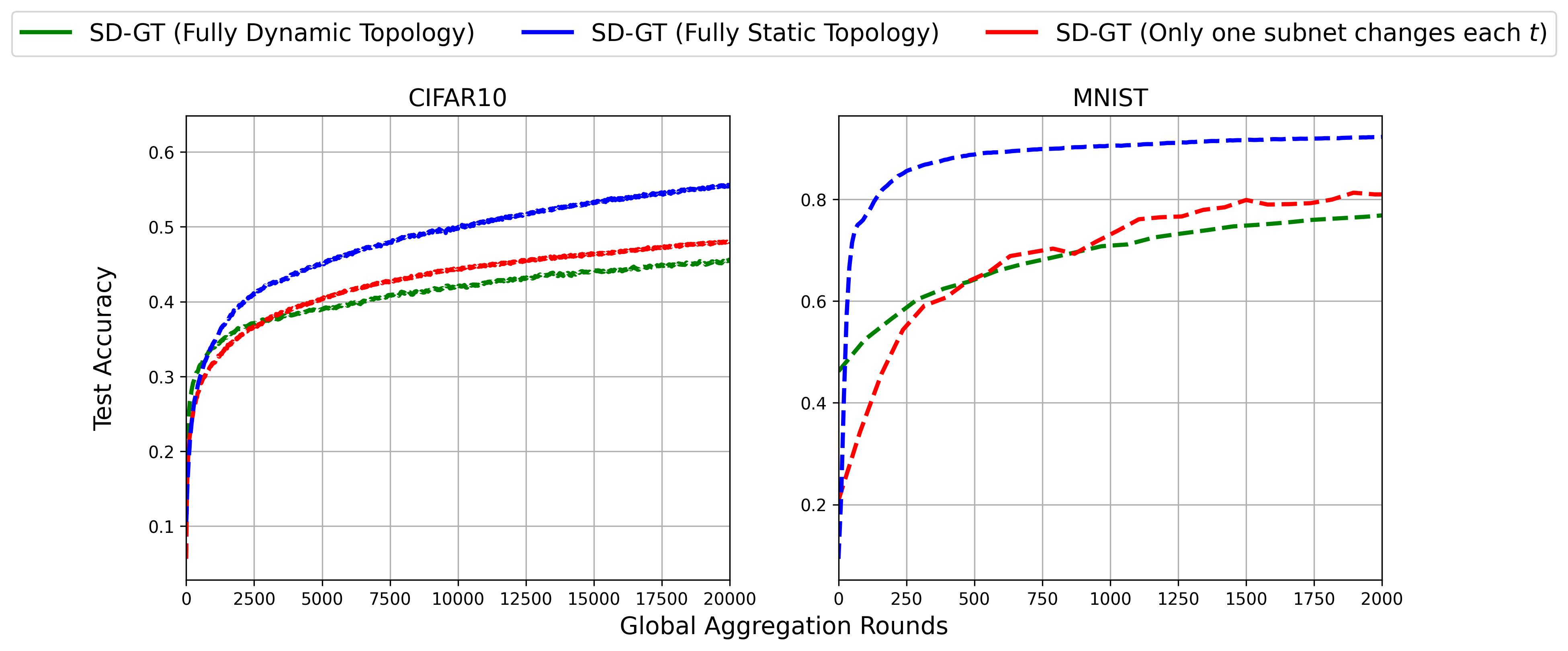}
    \caption{\textcolor{black}{Experimental results on CIFAR10 and MNIST considering the impact of subnet topology variation over global communication rounds $t$.}}
    \label{fig:time_varying}
\end{figure}

\vfill

{\color{black}
% \section{Analysis of {\tt SD-GT} Algorithm~\ref{alg:2}}
\section{Convergence Analysis of {\tt SD-GT} Algorithm~\ref{alg:2}}
\label{appen:ctrl_convergence}
In this section, we extend our convergence analysis to the {\tt DG-GT} adaptive control algorithm, for the non-convex loss case. Given time varying local update iterations $K_t$ and $\beta_s^t$, we can first derive the time-varying version of Lemma \ref{lem1}: 
\begin{equation}
        \begin{aligned}
            \frac{1}{n}\sum_{i=1}^n\mathbb{E}\|x_i^{t, k} - \overline{x}_g^{t}\|^2=& \frac{1}{n}\sum_{i=1}^n\mathbb{E}\left\|\sum_{j \in N^{in}_i}w_{ij}\bigg(x_j^{t, k-1} - \gamma (\nabla f_j(x^{t,k-1}, \xi_j^{t, k-1}) + y_j^t + z_j^t)]\bigg) - \overline{x}_g^{t}\right\|^2\\
            \leq& (1 + \frac{1}{(K_t)-1} + 4K_t\gamma^2L^2)\frac{1}{n}\sum_{i=1}^n\mathbb{E}\|x_i^{t,k-1} - \overline{x}_g^t\|^2+ 8K_t\gamma^2\mathcal{Y}_t+ 8K_t\gamma^2 \mathcal{Z}_t + 2K_t\gamma^2\mathbb{E}\|\nabla f(\overline{x}_g^t)\|^2 \\ & \qquad \qquad \qquad \qquad + \gamma^2 \sigma^2.
        \end{aligned}
    \end{equation}
Using the same unfolding process as Lemma~\ref{lem1}, we get:
\begin{align}
    \textstyle\Delta_t
        \leq&     \textstyle3(1 - p^t)K_t\Gamma_t
        + 24(K_t)^3\gamma^2\mathcal{Y}_t   \textstyle+ 24(K_t)^3\gamma^2\mathcal{Z}_t + 6(K_t)^3\gamma^2\mathbb{E}\|\nabla f(\overline{x}_\mathrm{g} ^t)\|^2 + 3(K_t)^2\gamma^2 \sigma^2.
\end{align}

Similarly, for Lemmas~\ref{lem2},~\ref{lem3},~\ref{lem4}, we can unfold the following iterations:
\begin{itemize}
    \item (Lemma~\ref{lem2})
    \begin{equation}
        \begin{aligned}
            n\mathcal{Y}_t =& \mathbb{E}\|Y^t + \nabla f(\overline{x}_g^t)(J_c - J)\|_F^2\\
            =&\mathbb{E}\bigg\| \bigg(Y^{t-1} + \nabla f(\overline{x}_g^{t-1})(J_c - J)\bigg)B+ \bigg(\frac{1}{K^{t-1}}\sum_{k=1}^{K^{t-1}}\nabla F(x^{t-1,k}, \xi^{t-1,k}) - \nabla F(\overline{x}_g^{t-1})\bigg)(J - J_c)B'\\
            &+ \bigg(\nabla F(\overline{x}_g^t) - \nabla F(\overline{x}_g^{t-1})\bigg)(J_c - J)\bigg\|_F^2\\
            \leq&\mathbb{E}\bigg\| \bigg(Y^{t-1} + \nabla f(\overline{x}_g^{t-1})(J_c - J)\bigg)B+ \bigg(\frac{1}{K^{t-1}}\sum_{k=1}^{K^{t-1}}\nabla F(x^{t-1,k}) - \nabla F(\overline{x}_g^{t-1})\bigg)(J - J_c)B'\\
            &+ \bigg(\nabla F(\overline{x}_g^t) - \nabla F(\overline{x}_g^{t-1})\bigg)(J_c - J)\bigg\|_F^2 + n\frac{\sigma^2}{K^{t-1}}\\
            \leq& (1 - \frac{p^t}{2})n\mathcal{Y}_{t-1} + \frac{6}{p^t}\bigg(\frac{nL^2}{K^{t-1}}\Delta_{t-1} + 2\gamma^2nL^4K^{t-1}\Delta_{t-1} + 2\gamma^2L^2(K^{t-1})^2n\mathbb{E}\|\nabla f(\overline{x}_g^{t-1})\|^2 \bigg)+\frac{6\gamma^2L^2K^{t-1}n}{p^t}\sigma^2 \\ & \qquad\qquad\qquad\qquad +  n\frac{\sigma^2}{K^{t-1}}.
        \end{aligned}
    \end{equation}
    \item (Lemma~\ref{lem3})
        \begin{equation}
        \begin{aligned}
            \mathcal{Z}_t = \frac{1}{n}& \mathbb{E}\|Z^t + \nabla F(\overline{x}_g^t)(I - J_c)\|^2_F\\
            \leq&(1-\frac{q}{2})\mathcal{Z}_{t-1} + \frac{6}{q}\bigg(\frac{4L^2}{K^{t-1}}\Delta_{t-1} + \|\nabla F(\overline{x}_g^t) - \nabla F(\overline{x}_g^{t-1})\|_F^2\bigg) + \frac{\sigma^2}{K^{t-1}}\\
            \leq&(1-\frac{q}{2})\mathcal{Z}_{t-1} + \frac{6}{q}\bigg(\frac{4L^2}{K^{t-1}}\Delta_{t-1} + 2K^{t-1}\gamma^2L^4\Delta_{t-1} + 2(K^{t-1})^2L^2\gamma^2\mathbb{E}\|\nabla f(\overline{x}_g^t)\|^2 + L^2K^{t-1}\gamma^2\sigma^2\bigg)+ \frac{\sigma^2}{K^{t-1}}.\\
            \overset{\gamma < \frac{1}{\sqrt{6}K^{t-1}L}}{\leq} &(1-\frac{q}{2})\mathcal{Z}_{t-1} + \frac{26L^2}{q}\Delta_{t-1} + \frac{12}{q}(K^{t-1})^2L^2\gamma^2\mathbb{E}\|\nabla f(\overline{x}_g^t)\|^2 + \frac{2\sigma^2}{qK^{t-1}}.\\
        \end{aligned}
    \end{equation}
    \item (Lemma~\ref{lem4})
    \begin{equation}
    \begin{aligned}
        \Gamma_t &= \frac{1}{n}\sum_{i=1}^n E\|x_i^{t-1,K+1} - \overline{x}_g^t\|^2\\
        &\leq \frac{1}{n}\sum_{i=1}^nE\left\|x_i^{t-1,K+1} - \overline{x}_g^{t-1} + \gamma\sum_{k=1}^{K^{t-1}}\frac{1}{n}\sum_{i=1}^n \nabla f_i(x_i^{t-1,k}, \xi_i^{t-1,k})\right\|^2\\
        &= \frac{1}{n}\sum_{i=1}^nE\left\|x_i^{t-1,K+1} - \overline{x}_g^{t-1} + \gamma\sum_{k=1}^{K^{t-1}}\frac{1}{n}\sum_{i=1}^n \nabla f_i(x_i^{t-1,k}, \xi_i^{t-1,k}) - \gamma \sum_{k=1}^{K^{t-1}}\frac{1}{n}\sum_{i=1}^n \nabla f_i(\overline{x}_g^{t-1}) + \gamma \sum_{k=1}^{K^{t-1}}\frac{1}{n}\sum_{i=1}^n \nabla f_i(\overline{x}_g^{t-1})\right\|^2\\
    &\leq (1 - \frac{p^t}{2})\Gamma_{t-1} + \frac{12}{p^t}\gamma^2K^{t-1}L^2\Delta_{t-1} + \frac{12}{p^t}\gamma^2(K^{t-1})^2\mathcal{Y}_{t-1} + \frac{12}{p^t}\gamma^2(K^{t-1})^2\mathcal{Z}_{t-1} + \frac{12}{p^t}\gamma^2(K^{t-1})^2\|\nabla f(\overline{x}_g^{t-1})\|^2 \\ & \qquad\qquad\qquad\qquad+ K^{t-1}\gamma^2\sigma^2.
    \end{aligned}
\end{equation}
\end{itemize}

Combining these results using the same process as in Theorem~\ref{thm1}, and letting $K^{\min} = \min_t K_t$, $K^{\max} = \max_t K_t$, and $p^{\min} = \min_t p^t$, we get:

\begin{equation}
    \frac{1}{T} \sum_{t=1}^T \mathbb{E}\|\nabla f(\overline{x}_g^t)\|^2\leq \frac{1}{DTK^{\min}\gamma}(\mathbb{E}f(\overline{x}_\mathrm{g} ^1) - \mathbb{E}f(x^\star)) +\frac{D_5L^2}{D(p^{\min})^4q^2  K^{\min}}((K^{\max})^2\gamma^2)\sigma^2 + \frac{L}{2DnK^{\min}}(K^{\max}\gamma)\sigma^2.
\end{equation}

Finally, by choosing a constant step size
$$\gamma = \min\left\{\left(\frac{2\mathcal{E}_1 K^{\max}n}{L\sigma^2 T}\right)^{1/2}, \left(\frac{2\mathcal{E}_1 (p^{\min})^4q^2 K^{\max}}{L\sigma^2 T}\right)^{1/3}, \frac{(p^{\min})^2q^2}{945 K^{\max}L}\right\},$$ {\tt SD-GT} obtains the following rate:
\begin{align}
    &\frac{1}{T}\sum_{t=1}^T\mathbb{E}\|\nabla f(\overline{x}_\mathrm{g} ^t)\|^2 
    = \mathcal{O}\bigg(\sqrt{\frac{\mathcal{E}_1\sigma^2L}{nTK^{\min}}} + (\frac{\mathcal{E}_1L\sigma}{\sqrt{K^{\min}}T(p^{\min})^2q})^{\frac{2}{3}} + \frac{\mathcal{E}_1L}{T(p^{\min})^2q}\bigg).
\end{align}

Thus we have the following theorem:
\begin{theorem}
\label{thm:SDGT_w_control_rate}
    Under Assumptions \ref{asmp1}, \ref{asmp2}, and \ref{assmp3}, let $\beta_{s} = \frac{m_{s} - h_{s}}{m_{s}}$ be the ratio of unsampled clients from each subnet. Define $p^t = \min(1 - \beta_{1}^2 , \ldots,1 - \beta_{S}^2)\in (0,1]$, $p^{\min} = \min_t p^t$, $q = \min(\rho_{1}, \ldots, \rho_{S})\in (0,1]$, and the function value optimality gap for time $t$ as $\mathcal{E}_t = \mathbb{E}f(\overline{x}_\mathrm{g} ^t) - f(x^\star)$. Then, for a constant step size satisfying $$\gamma = \min\left\{\left(\frac{2\mathcal{E}_1 K^{\max}n}{L\sigma^2 T}\right)^{1/2}, \left(\frac{2\mathcal{E}_1 (p^{\min})^4q^2 K^{\max}}{L\sigma^2 T}\right)^{1/3}, \frac{(p^{\min})^2q^2}{945 K^{\max}L}\right\},$$ the {\tt SD-GT} algorithm with dynamic control obtains the following rate:
\begin{align}
    &\frac{1}{T}\sum_{t=1}^T\mathbb{E}\|\nabla f(\overline{x}_\mathrm{g} ^t)\|^2 
    = \mathcal{O}\bigg(\sqrt{\frac{\mathcal{E}_1\sigma^2L}{nTK^{\min}}} + (\frac{\mathcal{E}_1L\sigma}{\sqrt{K^{\min}}T(p^{\min})^2q})^{\frac{2}{3}} + \frac{\mathcal{E}_1L}{T(p^{\min})^2q}\bigg).
\end{align}
\end{theorem}

% \subsection{Bounds on Control Parameter $K_t$}
% \label{appendix:bounded_ctrl_K}
To further understand the optimizer's effect on this time-varying rate of {\tt SD-GT}, we can obtain upper and lower bounds on the number of D2D rounds $K_t$. 
\begin{theorem}
    The dynamic D2D rounds $K_t$ derived from Algorithm~\ref{alg:2} for each global iteration $t$ can be bounded by:
    \begin{align}
    \min\Bigg( 
        &\left(\frac{\tfrac{1}{2} \sqrt{\lambda_2 \hat{H}_t} + \tfrac{2}{3}\big(\lambda_2 (\max_s m_s^2)\hat{H}_t\big)^{\tfrac{2}{3}}}{\lambda_3 \sum_{s=1}^S E^{D2D}_{s}}\right)^{\tfrac{2}{3}}, \;
        \left(\frac{\tfrac{1}{2} \sqrt{\lambda_2 \hat{H}_t} + \tfrac{2}{3}\big(\lambda_2 (\max_s m_s^2)\hat{H}_t\big)^{\tfrac{2}{3}}}{\lambda_3 \sum_{s=1}^S E^{D2D}_{s}}\right)^{\tfrac{3}{5}}
        \Bigg) \\
        &\leq K_t \leq \max\Bigg(
        \left(\frac{\tfrac{1}{2} \sqrt{\lambda_2 \hat{H}_t} + \tfrac{2}{3}\left(\lambda_2\hat{H}_t\right)^{\tfrac{2}{3}}}{\lambda_3 \sum_{s=1}^S E^{D2D}_{s}}\right)^{\tfrac{2}{3}}, \;
        \left(\frac{\tfrac{1}{2} \sqrt{\lambda_2 \hat{H}_t} + \tfrac{2}{3}\left(\lambda_2\hat{H}_t\right)^{\tfrac{2}{3}}}{\lambda_3 \sum_{s=1}^S E^{D2D}_{s}}\right)^{\tfrac{3}{5}}
        \Bigg)
\end{align}
\end{theorem}

To derive this, we first take the derivative on the objective function w.r.t $K_t$:
\begin{align}
     &\frac{\partial}{\partial K}\sqrt{\frac{\lambda_2\hat{H}_t}{K}} + (\frac{\lambda_2\hat{H}_t}{Kp^2})^{\frac{2}{3}} 
    + \lambda_3 K\sum_{s=1}^S E^{D2D}_{s} = 0\\
    \Rightarrow & \frac{1}{2} \sqrt{\lambda_2 \hat{H}_t} \frac{1}{K^{\frac{3}{2}}} + \frac{2}{3}(\frac{\lambda_2\hat{H}_t}{p^2})^{\frac{2}{3}} \frac{1}{K^{\frac{5}{3}}} = \lambda_3 \sum_{s=1}^S E^{D2D}_{s}
\end{align}
Then, we can see that:
\begin{align}
    K^{\max} &\leq \max\Bigg(
        \left(\frac{\tfrac{1}{2} \sqrt{\lambda_2 \hat{H}_t} + \tfrac{2}{3}\left(\tfrac{\lambda_2\hat{H}_t}{p^2}\right)^{\tfrac{2}{3}}}{\lambda_3 \sum_{s=1}^S E^{D2D}_{s}}\right)^{\tfrac{2}{3}}, \;
        \left(\frac{\tfrac{1}{2} \sqrt{\lambda_2 \hat{H}_t} + \tfrac{2}{3}\left(\tfrac{\lambda_2\hat{H}_t}{p^2}\right)^{\tfrac{2}{3}}}{\lambda_3 \sum_{s=1}^S E^{D2D}_{s}}\right)^{\tfrac{3}{5}}
        \Bigg) \\
    &\leq \max\Bigg(
        \left(\frac{\tfrac{1}{2} \sqrt{\lambda_2 \hat{H}_t} + \tfrac{2}{3}\left(\lambda_2\hat{H}_t\right)^{\tfrac{2}{3}}}{\lambda_3 \sum_{s=1}^S E^{D2D}_{s}}\right)^{\tfrac{2}{3}}, \;
        \left(\frac{\tfrac{1}{2} \sqrt{\lambda_2 \hat{H}_t} + \tfrac{2}{3}\left(\lambda_2\hat{H}_t\right)^{\tfrac{2}{3}}}{\lambda_3 \sum_{s=1}^S E^{D2D}_{s}}\right)^{\tfrac{3}{5}}
        \Bigg)
\end{align}

and 
\begin{align}
    K^{\min} &\geq \min\Bigg( 
        \left(\frac{\tfrac{1}{2} \sqrt{\lambda_2 \hat{H}_t} + \tfrac{2}{3}\left(\tfrac{\lambda_2\hat{H}_t}{p^2}\right)^{\tfrac{2}{3}}}{\lambda_3 \sum_{s=1}^S E^{D2D}_{s}}\right)^{\tfrac{2}{3}}, \;
        \left(\frac{\tfrac{1}{2} \sqrt{\lambda_2 \hat{H}_t} + \tfrac{2}{3}\left(\tfrac{\lambda_2\hat{H}_t}{p^2}\right)^{\tfrac{2}{3}}}{\lambda_3 \sum_{s=1}^S E^{D2D}_{s}}\right)^{\tfrac{3}{5}}
        \Bigg) \\
    &\geq \min\Bigg( 
        \left(\frac{\tfrac{1}{2} \sqrt{\lambda_2 \hat{H}_t} + \tfrac{2}{3}\big(\lambda_2 (\max_s m_s^2)\hat{H}_t\big)^{\tfrac{2}{3}}}{\lambda_3 \sum_{s=1}^S E^{D2D}_{s}}\right)^{\tfrac{2}{3}}, \;
        \left(\frac{\tfrac{1}{2} \sqrt{\lambda_2 \hat{H}_t} + \tfrac{2}{3}\big(\lambda_2 (\max_s m_s^2)\hat{H}_t\big)^{\tfrac{2}{3}}}{\lambda_3 \sum_{s=1}^S E^{D2D}_{s}}\right)^{\tfrac{3}{5}}
        \Bigg)
\end{align}

We observe that both the upper and lower bounds for the control parameter $K_t$ have the same dependence on the D2D communication energy term. In particular, higher D2D costs drive the algorithm to select smaller values of $K_t$, prioritizing communication savings over performance optimization, while lower costs allow larger $K_t$ to improve accuracy. This aligns with the intuition behind the design of our control algorithm.

}

\end{document}